\documentclass[reqno]{amsart}

\usepackage{amscd,amsmath,amsfonts,amssymb,amsthm,cite,mathrsfs,color}

\newtheorem{theorem}{Theorem}[section]
\newtheorem{lemma}[theorem]{Lemma}
\newtheorem{proposition}[theorem]{Proposition}
\newtheorem{coro}[theorem]{Corollary}
\theoremstyle{definition}
\newtheorem {definition}[theorem]{Definition}
\theoremstyle{remark}
\newtheorem{remark}[theorem]{Remark}
\newtheorem*{acknow}{Acknowledgments}

\numberwithin{equation}{section}
\numberwithin{theorem}{section}

\def\iy{\infty}

\def\be{\begin{equation}}
\def\ee{\end{equation}}
\def\bae{\begin{eqnarray}}
\def\eae{\end{eqnarray}}

\begin{document}
\title[Selberg Integrals and Super hypergeometric functions]{Selberg Integrals, Super hypergeometric functions and Applications to $\beta$-Ensembles of Random Matrices}

\author{Patrick Desrosiers} \address{Instituto Matem\'atica y F\'isica,
Universidad de Talca, 2 Norte 685, Talca,
Chile}\email{Patrick.Desrosiers@inst-mat.utalca.cl}

\author{Dang-Zheng Liu} \address{Instituto Matem\'atica y F\'isica,
Universidad de Talca, 2 Norte 685, Talca,
Chile}\email{dzliu@inst-mat.utalca.cl}

\date{\today}


\begin{abstract}
We study a new Selberg-type integral  with $n+m$ indeterminates, which turns out to be related to the deformed Calogero-Sutherland systems.  We show that the integral satisfies a holonomic system of $n+m$ non-symmetric linear partial differential equations. We also prove that  a particular hypergeometric function defined in terms of super Jack polynomials is the unique solution of the  system. Some properties such as duality relations, integral formulas, Pfaff-Euler and Kummer transformations are also established.  As a direct application, we evaluate the expectation value of ratios of characteristic polynomials in the  classical $\beta$-ensembles of Random Matrix Theory.
\end{abstract}

\maketitle

\small
\tableofcontents
\normalsize

\section{Introduction}

\subsection{Kaneko's integral and symmetric polynomials}

In the early 1990s, Kaneko \cite{kaneko} thoroughly studied the following generalization of the Selberg integral:
\begin{equation} \label{kanekoint}K_N(\lambda_1,\lambda_2,\lambda;t)=\int_{[0,1]^{N}}
 D_{\lambda_{1},\lambda_{2},\lambda}(x) \prod_{j=1}^N\prod_{k=1}^n(x_j-t_k)^\mu\,d^{N}x \end{equation}
where  $t$ denotes the set of variables $t_1,\ldots,t_n$, the parameter $\mu$  is either equal  to  $1$ or $-\lambda$,  and
\begin{align}\label{selbergdensity} D_{\lambda_{1},\lambda_{2},\lambda}(x)&= \prod_{i=1}^{N}x_{i}^{\lambda_{1}}(1-x_{i})^{\lambda_{2}} \,\prod_{1\leq j<k\leq
N}|x_{j}-x_{k}|^{2\lambda}.\end{align}
Note that a convenient way to guarantee the convergence of the integral is to set
%\label{condintionsint}
$\Re (\lambda) >0$, $\Re(\lambda_1)>-1$, $\Re(\lambda_{2})>-1$,
and for the case with $\mu=-\lambda$, $t_i\in\mathbb{C}\setminus[0,1]$.
 In the absence of variables $t$, which corresponds to the case $n=0$, the value of the integral $K_N$ is given by Selberg's celebrated result \cite{selberg}
\begin{align} \label{Selbergformula} S_{N}&(\lambda_{1},\lambda_{2},\lambda)=\prod_{j=0}^{N-1}
\frac{\Gamma(1+\lambda+j\lambda) \Gamma(1+\lambda_{1}+j\lambda) \Gamma(1+\lambda_{2}+j\lambda)}
{\Gamma(1+\lambda)\Gamma(2+\lambda_{1}+\lambda_{2}+(N+j-1)\lambda)}.
\end{align}
See \cite{fw} for a beautiful discussion on the above formula and its various applications.

Kaneko showed in particular that the integral $K_N(\lambda_1,\lambda_2,\lambda;t)$ satisfies the following holonomic system of partial differential equations:
\begin{multline}\label{kanekosyst}
 t_i(1-t_i)\frac{\partial^2 F}{\partial t_i^2}+\left(c-\frac{n-1}{\alpha}\right)\frac{\partial F}{\partial t_i}
 -\left(a+b+1-\frac{n-1}{\alpha}\right)t_i\frac{\partial F}{\partial t_i}-ab F\\+\frac{1}{\alpha}\sum_{j\neq i}\frac{1}{t_i-t_j}\left(t_i(1-t_i)\frac{\partial F}{\partial t_i}-t_j(1-t_j)\frac{\partial F}{\partial t_j}\right)=0,\quad i=1,\ldots,n,
\end{multline}
where $\alpha=\lambda$, $ a=-N$, $b=(\lambda_1+\lambda_2+n+1)/\lambda+N-1$, $c=(\lambda_1+n)/\lambda$ if $\mu=1$, while
$\alpha=1/\lambda$, $ a=\lambda N$, $b=-(\lambda_1+\lambda_2+1)+\lambda(n-N+1)$, $c=-\lambda_1+n\lambda$ if $\mu=-\lambda$.   He further proved  that the above system possesses a unique solution $F(t_1,\ldots,t_n)$ subject to the following conditions:
\begin{enumerate} \item $F(t_1,\ldots,t_n)$ is symmetric of $t_1,\ldots,t_n$;
\item $F(t_1,\ldots,t_n)$ is analytic at $(0,\ldots,0)$ and such that $F(0,\ldots,0)=1$.
\end{enumerate}
This  allowed him to conclude for instance that the integral $K_N$ given in Eq. \eqref{kanekoint} with $\mu=1$ is equal to the Selberg constant $S_N(\lambda_1+n,\lambda_2,\lambda)$ multiplied by a Gaussian hypergeometric function in $n$ variables $_2F_1(a,b;c;t_1,\ldots,t_n)$.
As will be explained in Section \ref{superseries}, the precise definition of the latter function requires some theory on symmetric Jack polynomials, a subject that was developed in the late 1980s, principally by Kadell, Macdonald, and Stanley  \cite{kadell, macdonald,stanley}.  Remarkably, the theorem on the  existence and uniqueness of the solution for the system \eqref{kanekosyst} was independently and almost simultaneously obtained by Yan \cite{yan}.  Although not concerned with the full generalized Selberg integral \eqref{kanekoint}, Yan obtained a series of results on hypergeometric functions also defined in terms of Jack polynomial together with some of their integral representations.

One of the main inspirations behind Kaneko's work was an article by Aomoto on a one-variable generalization of the Selberg integral and its relation to the Jacobi polynomials \cite{aom}.  In our notation, Aomoto's integral corresponds to $K_N(\lambda_1,\lambda_2,\lambda;t_1)$ (i.e., case $n=1$).   Another crucial motivation for studying the integral \eqref{kanekoint} was Muirhead's holonomic system of equations \cite{muirhead1} whose unique solution was given as a hypergeometric function  $_2F_1(a,b;c;T)$ whose argument $T$ is a real symmetric $n\times n$ matrix.  For special values of parameters $a$, $b$, $c$, the latter function can be represented as an integral over some symmetric $N\times N$ matrix,  which in turn is closely related to the integral \eqref{kanekoint} with $n$ general but with $\lambda=1/2$ \cite{muirhead}.

All the previous results converged in the seminal article by Baker and Forrester \cite{bf}.  There, the authors used the hypergeometric functions of Kaneko and Yan for studying the Hermite, Laguerre and Jacobi polynomials in many variables as independently defined by Lassalle and Macdonald, both using Jack polynomial theory.  They also showed that the classical polynomials in many variables as well as their associated hypergeometric functions are all related to the famous quantum many systems of Calogero and Sutherland and applied their results for evaluating the ground state density of these systems, which, as is well known, is equivalent to the eigenvalue density of classical ensembles of random matrices.

\subsection{Relations with Random Matrix Theory}

The integrand \eqref{selbergdensity} is indeed known in Random Matrix Theory \cite{forrester,fyodorov,mehta} as the non-normalized eigenvalue probability density function in the Jacobi $\beta$-Ensemble, where
$\beta=2\lambda $ is any positive real.    The latter ensemble is a generalization of the three classical Jacobi Ensembles which can been defined in terms of rectangular matrices with real ($\beta=1$), complex ($\beta=2$) or real quaternions ($\beta=2$) elements  \cite[Chapter 3]{forrester}.  Thanks to the parameters $\lambda_1$ and $\lambda_2$, the Jacobi Ensembles  can also be seen  as generalizing the somewhat  more common Circular,  Gaussian and Chiral  Ensembles of random matrices, whose respective eigenvalue probability densities are proportional to
\begin{eqnarray}\label{eqbetadensities}
 \prod_{1\leq j<k\leq
N}|e^{2\pi ix_{j}}-e^{2\pi i x_{k}}|^{\beta}  &x_j\in [0,1] &\text{Circular }\\
\prod_{i=1}^{N}e^{-\beta x_i^2/2} \,\prod_{1\leq j<k\leq
N}|x_{j}-x_{k}|^{\beta} & x_j\in \mathbb{R} &\text{Gaussian (Hermite)}\\
\prod_{i=1}^{N}x_{i}^{\lambda_1}e^{-\beta x_i} \,\prod_{1\leq j<k\leq
N}|x_{j}-x_{k}|^{\beta}  &x_i\in \mathbb{R}_+& \text{Chiral (Laguerre)}.
\end{eqnarray}

When $\beta$ is any real positive number, the function \eqref{selbergdensity} is proportional  to  the eigenvalue density coming from a real symmetric matrix whose non-zero elements lay on the three main diagonals and are randomly drawn from some specific distributions.  This was first obtained in the article \cite{KN} by extending the work of Dumitriu and Edelman \cite{due} on Laguerre and Hermite $\beta$-Ensembles as models of tri-diagonal random matrices.

The properties and applications of $\beta$-Ensembles have been studied by many other authors.  For instance, it was  established that in the limit where the size $N$ of a random matrix $\beta$-Ensembles goes to infinity, the eigenvalues of  the random matrix become statistically distributed as the eigenvalues of some one-dimensional stochastic differential operators \cite{rrv,vv}. In high energy physics, the $\beta$-Ensembles recently appeared in the study of a conjectured (so called AGT) duality between two quantum field theories defined on space of distinct dimensions: the four-dimensional Seiberg-Witten theory and the two-dimensional conformal field theory \cite{mmm,sulkovski}.  They are also expected to play an important role in topological string theory and are used in the context of ``quantum" Riemann surfaces  \cite{bempf,cem}, which are  parametrized by a formal Planck parameter $\hbar=(\sqrt{\beta/2}-\sqrt{2/\beta})/N$.

 If $x=(x_1,\ldots,x_N)$ denotes the eigenvalues of a $N\times N$ random matrix $X$ in the Jacobi $\beta$-Ensemble, we see that the integral \eqref{kanekoint} is proportional to the the expectation value of products of characteristic polynomials in $X$:
\begin{equation} \label{expectationproc} K_N(\lambda_1,\lambda_2,\lambda;t_1,\ldots,t_n)= S _N(\lambda_1,\lambda_2,\lambda)\, \left\langle \prod_{j=1}^n\det (-t_j+ X)^\mu\right\rangle_{X\in J\beta E}\end{equation}
Similar expressions exist for the Hermite and Laguerre $\beta$-Ensembles.  This means  that the expectation values of product of characteristic polynomials in $\beta$-Ensembles of random matrices can be given in terms of hypergeometric functions involving Jack polynomials.  Other probabilistic quantities such as the cumulation distribution of the largest or smallest eigenvalue, the marginal probability distribution of $k$  eigenvalues have been related to the hypergeometric functions studied by Kaneko and Yan or the their associated classical symmetric polynomials \cite{bf,des,df,due,forrester,ke}.  Despite the fact that the generalized hypergeometric functions  are defined as sums over an infinite number of Jack polynomials, efficient algorithms have been found  that make possible   their numerical evaluation \cite{ke}.

Although interesting in their own right, expectations of product of  characteristic polynomial such as \eqref{expectationproc} are not the most useful probabilistic object.   Expectation of ratios of characteristic polynomials such as
\begin{equation} \label{expectationratio} S _N(\lambda_1,\lambda_2,\lambda;t;s)= S _N(\lambda_1,\lambda_2,\lambda)\, \left\langle \frac{\prod_{k=1}^n\det (1-t_k X)^{\mu_1}}{\prod_{l=1}^m\det (1-s_l X)^{\mu_2}}\right\rangle_{X\in J\beta E}\end{equation}
 are often  preferred since they are of particular interest in the study of the zeros of the Riemann $\zeta$-function \cite{ks} and provide a method for computing   the marginal densities of the eigenvalues also known as the  $n$-point correlation functions (see for instance Proposition 2.16 in \cite{bs}). For the three classical cases $\beta=1,2,4$,  the exact calculation of these expectations have been completed by several authors (see \cite{bds,bs} and references therein).
 In the case where $\beta$ is a general positive integer however, almost nothing is known.

\subsection{Goals and main results}

 Our aim is  to exactly calculate  the expectation  \eqref{expectationproc} of ratios of characteristic polynomials in the Jacobi-$\beta$-Ensemble for $\beta$ general, which is equivalent to the following generalization of the  Selberg integral with $\lambda=\beta/2$:
 \begin{equation} \label{skanekoint}S _N(\lambda_1,\lambda_2,\lambda;t;s)=\int_{[0,1]^{N}}
 D_{\lambda_{1},\lambda_{2},\lambda}(x)\frac{\prod_{j=1}^N\prod_{k=1}^n(1-x_jt_k)^{\mu_1} }{\prod_{j=1}^N\prod_{\ell=1}^m(1-x_js_\ell)^{\mu_2} } d^N x,\end{equation}
 where it is assumed that $ s_\ell \in \mathbb{C} \setminus [1,\infty)$ for all $\ell$. By looking at similar integrals but with $D_{\lambda_{1},\lambda_{2},\lambda}$ replaced by the densities in \eqref{eqbetadensities}, we also want to find exact expressions for the  expectation ratios of characteristic polynomials in the Circular-, Gaussian- and Chiral- $\beta$-Ensembles.  For reasons that will become clear in the next section, we will always suppose
 \begin{equation}\label{conditionsmu} \mu_1=1\qquad\text{and} \qquad\mu_2=\lambda.\end{equation}

One of the few results on expectations of ratios of characteristic polynomials in $\beta$-Ensembles was the observation  in \cite[Section 6]{des}    that such expectation for the Gaussian $\beta$-Ensemble is an eigenfunction of a second order differential operator which can be interpreted as the Hamiltonian  in the so-called deformed Calogero model.   The latter belongs to a more general family of quantum Calogero-Sutherland models  involving two sets of variables (particles) that were extensively studied by Sergeev and Veselov \cite{sv,sv2} and more recently in \cite{dh}.  They were also shown to be related to super-matrices and symmetric super-spaces \cite{gk} as well as to ``multispecies Calogero models" exhibiting particles-antiparticles and strong-weak dualities \cite{bfm}.

In Section \ref{section1}, we provide a stronger result than \cite[Section 6]{des} by showing that the function defined by integral \eqref{skanekoint} satisfies a holonomic system of $n+m$ nonsymmetric differential equations.

\begin{theorem}\label{theodeformedkanekosystem} Let
\begin{equation}\label{condparameters}\alpha=\lambda,\quad a=-N,\quad b=-N+1-(1/\lambda)(1+\lambda_1),\quad c=-2N+2-(1/\lambda)(2+\lambda_1+\lambda_2).\end{equation}
Then $S _N(\lambda_1,\lambda_2,\lambda;t_1,\ldots,t_n;s_1,\ldots,s_m)$ satisfies
\begin{multline}\label{PDEs1}
t_i (1 - t_i) \frac{\partial^2 F}{\partial t_i^2} + \Bigl( c - \bigl(a + b + 1 \bigr) t_i \Bigr)
\frac{\partial F}{\partial t_i} - a b F  +
\frac{1}{\alpha}\sum_{k=1, k \ne i}^{n}
\frac{t_k}{t_i - t_k} \biggl((1 - t_i)
 \frac{\partial F}{\partial t_i}\\
- (1 - t_k) \frac{\partial F}{\partial t_k} \biggr)-\sum_{\substack{k=1}}^{m}
\frac{s_k}{t_i - s_k} \biggl((1 - t_i)
 \frac{\partial F}{\partial t_i}
+ \frac{1}{\alpha}(1 - s_k) \frac{\partial F}{\partial s_k} \biggr)
 = 0,
\end{multline}
\begin{multline}\label{PDEs2}
-\frac{1}{\alpha}s_j (1 - s_j) \frac{\partial^2 F}{\partial s_j^2} + \Bigl( c
 - \bigl(a + b + 1 -(1+\frac{1}{\alpha})\bigr) s_j \Bigr)
\frac{\partial F}{\partial s_j} - (-\alpha)abF
-\sum_{k=1, k \ne j}^{m}
\frac{s_k}{s_j - s_k} \\
\biggl(
(1 - s_j)
 \frac{\partial F}{\partial s_j}
- (1 - s_k) \frac{\partial F}{\partial s_k} \biggr)+\sum_{\substack{k=1}}^{n}
\frac{t_k}{s_j - t_k} \biggl(\frac{1}{\alpha}(1 - s_j) \frac{\partial F}{\partial s_j}+(1 - t_k)
 \frac{\partial F}{\partial t_k} \biggr)
 = 0,
\end{multline}
for all  $i=1,\ldots,n$ and $j=1,\ldots,m$.  Moreover, it also  satisfies the cancellation condition \be
 \left(\frac{\partial F}{\partial t_i}+\frac{1}{\alpha}
 \frac{\partial
  F}{\partial s_j} \right)_{t_i=s_j}=0, \qquad\forall i,j . \ee
  \end{theorem}

 The sum of all the above differential equation gives a single differential equation of the form $\mathcal{L}_{n,m}F=0$, where $\mathcal{L}_{n,m}$ is a linear combination of operators that previously appeared in the article \cite{dh}.  The latter was concerned with the solutions of some deformed Calogero-Sutherland models \cite{sv,sv2,sv3} whose study requires the use of the super-Jack polynomials, here denoted by  $S\!P^{(\alpha)}_\kappa(t;s)$.  These objects were  introduced by Kerov, Okounkov and Olshanski \cite{koo}, and subsequently  studied  by Sergeev and Veselov \cite{sv2}.  They generalize both the usual Jack polynomials $P^{(\alpha)}_\kappa(t)$ (case with $s=(0,\ldots,0)$)  and  supersymmetric Schur polynomials $s_\kappa(t|s)$ (case $\alpha=1$).  The latter were first defined  as characters
in the representation theory of the superalgebra $gl(n|m)$; see
examples 23-24 in Section I.3 in \cite{macdonald} and the article \cite{MVdJ04}.

 A solution for the equation for $\mathcal{L}_{n,m}F=0$ was given in \cite{dh} as a super-hypergeometric function $\,_2S\!F^{(\alpha)}_1(a,b;c;t;s)$, which can be written as an infinite sum of super-Jack polynomials.
This provides us with one possible solution of the system \eqref{PDEs1}-\eqref{PDEs2}.  We will confirm this claim in Section \ref{sectsystem}.
Unfortunately, there is no uniqueness theorem for the equation  $\mathcal{L}_{n,m}F=0$ in \cite{dh} that would allow us to conclude that the deformed Selberg integral in \eqref{skanekoint} is indeed equal to $\,_2S\!F^{(\alpha)}_1(a,b;c;t;s)$.  As we will explain in Section \ref{sectsystem}, the method developed by Muirhead and Kaneko for proving the uniqueness for systems like \eqref{kanekosyst} does not extend directly for our case.

We will tackle the delicate problem of uniqueness for the system  \eqref{PDEs1}-\eqref{PDEs2} in several steps (Sections \ref{SectHoloInfinite} to \ref{SectHoloFinite}).   We will first prove that there is a unique symmetric formal power series that satisfies the infinite version of the system \eqref{PDEs1} (i.e., with $n=\infty$).  This will allow us to conclude that the solution of the system \eqref{PDEs1}-\eqref{PDEs2}, when both $n$ and $m$ are infinite.  We will also show that the latter infinite system is compatible with the rectiction of the number of variable to the finite system.  We will finally use this property and the fact    that the  super-hypergeometric function $\,_2S\!F^{(\alpha)}_1(a,b;c;t;s)$ satisfies some additional constraints, here called the Kaneko's criterion, to prove the following.

\begin{theorem}\label{maintheorem}
The function $\,_2S\!F^{(\alpha)}_1(a,b;c;t;s)$ is the unique solution of the system  \eqref{PDEs1}-\eqref{PDEs2}  subject to the following  additional conditions:
\begin{align}(1)&\quad\text{$F(t;s)$ is symmetric both in $t_1,\ldots,t_n$ and in $s_1,\ldots,s_m$;} \nonumber \\
(2) &\quad\text{$F(t;s)$ is analytic at $(0,\ldots;0,\ldots)$ and such that $F(0,\ldots;0,\ldots)=1$;} \nonumber\\
(3) &\quad\text{$F(t;s)$ satisfies the cancellation condition} \nonumber \\
& \qquad \qquad \left(\frac{\partial F}{\partial t_i}+\frac{1}{\alpha}
 \frac{\partial
  F}{\partial s_j} \right)_{t_i=s_j}=0, \qquad\forall i,j . \label{cancellationproperty}
\end{align}
\end{theorem}

\begin{coro} \label{maincoro} Let $\alpha$, $a$, $b$, and $c$ satisfy the conditions in \eqref{condparameters}.  Then,\begin{equation*} S _N(\lambda_1,\lambda_2,\lambda;t;s)= S _N(\lambda_1,\lambda_2,\lambda) \phantom{}_2S\!F^{(\lambda)}_1(a,b;c;t;s).\end{equation*}
\end{coro}

In Section \ref{sectproperties}, we will obtain more properties related to the super-hypergeometric functions $\phantom{}_pS\!F_q$.  This will make possible in Section \ref{sectapplication}, to show that the expectations of ratios of characteristic polynomials in the  Gaussian-, Laguerre-, and Circular-$\beta$-Ensembles are also super-hypergeometric functions.

Finally, it is worth mentioning that the super-Jack polynomials have already appeared in the study of the Circular-$\beta$-Ensemble.  Indeed, the large $N$ limit calculation of $n$-point correlation function for this ensemble was attempted in the preprint \cite{okounkov} and involved super series of hypergeometric type.  More recently in \cite{matsumoto2}, some formal results on the ratios of characteristic polynomials in the Circular-$\beta$-Ensemble were derived by using super-Jack polynomials and super-series.  The latter preprint also contains interesting dualities relations between  expectation values in the Circular $\beta$ and $4/\beta$ Ensembles.   Such  $\beta\leftrightarrow 4/\beta$ dualities are common in Random Matrix Theory (see \cite{des} and references therein).  For more general ensembles, such as the Jacobi-$\beta$-Ensemble, the duality relations for the expectation values of ratios of characteristic polynomials are much more complicated than their circular counterparts, and will be the subject of a forthcoming article.

\section{Deformed holonomic systems}  \label{section1}

We want to prove Theorem \ref{theodeformedkanekosystem}.  In order to facilitate the presentation of the demonstration, we first introduce some more notations based on the theory of supermatrices \cite{de,gk}.

\subsection{$\mathbb{Z}_2$-graded notation}
Let $w=w_1,\ldots,w_{n+m}$ be an ordered set of variables equipped with the following  $\mathbb{Z}_2$-grading: $w_i$ is even if $i\leq n$ and odd otherwise.   We set
\be \label{defvarw} w_i=\begin{cases} t_i,& i\leq n\\ s_{i-n},& i>n.\end{cases}\nonumber\ee
Let also $\rho$ denote the $\mathbb{Z}_2$-graded function on  $\{1,\ldots,n+m\}$ defined by
\begin{equation}\label{defrho}
\rho_i=\rho(i):=\begin{cases}\phantom{-} 1, & i\leq n;\\ -1/\alpha,& i>n.\end{cases}\nonumber
\end{equation}
We see that if $\alpha=\lambda$ and equation \eqref{conditionsmu} is satisfied,  then
\begin{equation*} \prod_{j=1}^N\prod_{i=1}^{n+m}(1-x_jw_i)^{1/\rho_i}=\frac{\prod_{j=1}^N\prod_{k=1}^n(1-x_jt_k)^{\mu_1} }{\prod_{j=1}^N\prod_{\ell=1}^m(1-x_js_\ell)^{\mu_2}}\end{equation*}
The   $\mathbb{Z}_2$-grading allows moreover  to recast the two types of differential operators given in Theorem \ref{theodeformedkanekosystem} into a single type of deformed  differential operators:
\begin{multline}\label{opsli}
\mathcal{S}\!\mathcal{L}_i=\rho_iw_i(1-w_i)\frac{\partial^{2}}{\partial w_{i}^{2}} +c\frac{\partial}{\partial w_i}-(a+b+\rho_i)w_i\frac{\partial}{\partial w_i}\\
+ \frac{1}{\alpha}\sum_{j=1, j \ne i}^{n+m}
\frac{1}{\rho_i\rho_j}\frac{w_j}{w_i - w_j} \biggl( \rho_i(1-w_i)
 \frac{\partial }{\partial w_i}
-
\rho_j(1-w_j) \frac{\partial }{\partial w_j} \biggr)
\end{multline}
More general deformed operators will be considered in Section \ref{sectsystem}.

\subsection{System for  $S _N(\lambda_1,\lambda_2,\lambda;t;s)$}
We now proceed to show that for one specific choice of parameters $a$, $b$, $c$, and $\alpha$, the function $S _N(\lambda_1,\lambda_2,\lambda;t;s)$ defined in \eqref{skanekoint} is an eigenfunction of each operator $\mathcal{S}\!\mathcal{L}_i$ defined above.     The method presented here is in essence equivalent to Kaneko's, which basically uses Stokes' Theorem and clever choices of functions to be integrated and derived.  However, proceeding  exactly like Kaneko would rapidly lead us to long and almost uncheckable lines of equations.  For this reason, in the proof, we prefer to interpret  \eqref{skanekoint}    not only as $N$ dimensional integral, but also in terms of the scalar product over the Hilbert space $\mathscr{H}_N$ of symmetric functions in $L^2(\Omega; \mu)$, where $\Omega=(0,1)^N$ and $\mu$ is the absolutely continuous measure given by the Selberg density \eqref{selbergdensity}:
\be d\mu(x)=\frac{D_{\lambda_1,\lambda_2,\lambda}(x)}{S_N(\lambda_1,\lambda_2,\lambda)}dx,\qquad \lambda_1,\lambda_2,\lambda>0.\nonumber \ee
In what follows, the expectation value of any $\mu$-integrable function $f$ will be denoted by brackets:
\be \langle f\rangle =\int_\Omega f(x)\,d\mu(x) .\nonumber\ee

\begin{lemma}\label{operatorD}  Let $h:\,[0,1]\to \mathbb{C}$ be rational and regular. Define the following linear differential operator acting on smooth functions of $\mathscr{H}_N$:
\be \label{eqdefD}
\mathscr{D}=\sum_i\left(\partial_{x_i}+V_i(x)\right)h(x_i),\qquad V_i(x)=\partial_{x_i}(\ln D_{\lambda_{1},\lambda_{2},\lambda})(x),
\ee
Then for any symmetric $f:\,\overline{\Omega}\to \mathbb{C}$ that is  rational and regular,
\be
\left\langle \mathscr{D}(f) \right\rangle=\int_\Omega \mathscr{D}(f)d\mu(x)=0.
\ee
\end{lemma}
\begin{proof}
Let us first recall that  the weighted scalar product over $\mathscr{H}_N$ is given by
\be \langle f \,|\, g\rangle=\int_\Omega f(x) \overline{g(x)}d\mu(x), \nonumber\ee
where the bar denotes the complex conjugation.   This means that expectation value of $f$ over $\Omega$ is equal to $ \left\langle f \right\rangle=\langle f \,|\, 1\rangle$.
We now turn our attention to the domain $\Delta$ of the operator $\mathscr{D}$.  Given that
\be V_i(x)=\frac{\lambda_1}{x_i}-\frac{\lambda_2}{1-x_i}+\sum_{j\neq i}\frac{2\lambda}{x_i-x_j},\nonumber
\ee
we choose  $\Delta$ to be the set of all smooth and symmetric functions $f$ in $L^2(\Omega; \mu)$ such that $f/x_i\in L^2(\Omega; \mu)$ and  $f/(1-x_i)\in L^2(\Omega; \mu)$.  Note that $\mathscr{D}$ preserves the symmetry of the functions $f$, so we indeed have $\Delta\subset \mathscr{H}_N$ .  Since $\lambda_1,\lambda_2>0$, the domain $\Delta$ includes all   the functions we are concerned with, namely the rational symmetric functions that are continuous all over the closure $\overline{\Omega}$ of ${\Omega}$.   Moreover, the condition $\lambda_1,\lambda_2>0$ implies that $ D_{\lambda_{1},\lambda_{2},\lambda}(x)=0$ for all $x$ in $\partial\Omega=\overline{\Omega}\setminus\Omega$.  Hence,
\be \int_\Omega\partial_{x_i}\left(f(x)\overline{g(x)}D_{\lambda_{1},\lambda_{2},\lambda}(x)\right)dx=0,\qquad \forall\,i.\nonumber\ee
This allows us to conclude that for all $f$ and $g$ in $\Delta$ ,
\be \langle \mathscr{D}f| g\rangle = \langle f| \mathscr{D}^*g\rangle\qquad\text{with}\qquad \mathscr{D}^*=-\sum_i\overline{h(x_i)}\partial_{x_i}.\nonumber\ee
Consequently, for any $f$ in $\Delta$,  we have $\left\langle \mathscr{D}(f) |1 \right\rangle=\left\langle f |\mathscr{D}^*(1) \right\rangle=0$. \end{proof}

\begin{theorem} \label{teointsystem} Let $F(w):=S _N(\lambda_1,\lambda_2,\lambda;t;s)$  be the function defined in \eqref{skanekoint} with $s_i\in\mathbb{C}\setminus[1,\infty)$ and $\lambda_1,\lambda_2,\lambda>0$. Let also $\mathcal{S}\!\mathcal{L}_i$ be the deformed differential operator given in \eqref{opsli} and assume
\begin{equation} \label{condabc}\alpha=\lambda,\quad a=-N,\quad b=-N+1-(1/\lambda)(1+\lambda_1),\quad c=-2N+2-(1/\lambda)(2+\lambda_1+\lambda_2).\end{equation}   Then $F(w)$ satisfies the following deformed holonomic system of partial differential equations,
\be \,\rho_i\mathcal{S}\!\mathcal{L}_i \,F(w)=\,ab\,F(w),\qquad i=1,\ldots,n+m,\ee
and the cancellation property
\be \left(\rho_i\frac{\partial F}{\partial w_i}-\rho_j \frac{\partial F}{\partial w_j}\right)_{w_i=w_j}=0, \ee
whenever $w_i$ and $w_j$ do not share the same parity.
\end{theorem}
\begin{proof}  Before showing  anything, we note  that $F(w)$  is simply equal to the expectation $\langle \Pi\rangle$, where
 \be \Pi:=\prod_{j}\prod_{i}(1-x_jw_i)^{\sigma_i},\qquad \sigma_i:=\begin{cases}\phantom{-} \mu_1, & i=1,\ldots,n;\\ -\mu_2,& i=n+1,\ldots,n+m.\end{cases}\nonumber\ee
Also, assuming that $\mu_1,\mu_2>0$ and  $w_i\in \mathbb{C}\setminus [1,\infty)$ for all $i>n$, we remark that all the  partial derivatives of $\Pi$ with respect to the variables $w$ belong to $C^\infty(\overline{\Omega})$.  This allows us to take the derivatives in $w$  indifferently inside or outside the brackets (or equivalently, the integral signs).

The first part of the proof consists in obtaining relations between expectation values involving $\Pi$.  For this we successively apply Lemma \ref{operatorD} for different choices of operator $\mathscr{D}$ . Let us
  set  $h(x_i)=x_i$  in the definition \eqref{eqdefD} of the operator $\mathscr{D}$ and get the following operator:
\be
\mathscr{A}=\sum_i\left(\partial_{x_i}+V_i(x)\right)x_i,\qquad  V_i(x)=\frac{\lambda_1}{x_i}-\frac{\lambda_2}{1-x_i}+\sum_{j\neq i}\frac{2\lambda}{x_i-x_j}.\nonumber\ee
According to Lemma \ref{operatorD}, we have $\langle\mathscr{A}\Pi\rangle=0$.  Moreover    \be \partial_{x_i}(\Pi)=-\Pi\left(\sum_j\frac{\sigma_jw_j}{1-x_iw_j}\right). \nonumber\ee
Simple manipulations then lead to
\be\label{eqA} %\langle\mathscr{A}\Pi\rangle=0\quad\Longrightarrow\quad
\lambda_3\left\langle\Pi\right\rangle-\lambda_2\left\langle\sum_i\frac{1}{1-x_i}\Pi\right\rangle-\left\langle\sum_{i,j}\frac{\sigma_jx_iw_j}{1-x_iw_j}\Pi\right\rangle=0.\ee
where $\lambda_3=N(1+\lambda_1+\lambda_2+\lambda(N-1))$.   Similarly,   set $h(x_i)=1-x_i$ in \eqref{eqdefD} and define
\be
\mathscr{B}=\sum_i\left(\partial_{x_i}+V_i(x)\right)(1-x_i).\nonumber\ee
From  $\langle\mathscr{B}\Pi\rangle=0$ we get
\be\label{eqB}
%\langle\mathscr{B}\Pi\rangle=0\quad\Longrightarrow\quad
-\lambda_3\left\langle\Pi\right\rangle+\lambda_1\left\langle\sum_i\frac{1}{x_i}\Pi\right\rangle-\left\langle\sum_{i,j}\frac{\sigma_jw_j(1-x_i)}{1-x_iw_j}\Pi\right\rangle=0.\ee
We need to introduce yet another differential operator:
\be
\mathscr{C}_{w_k}=\sum_i\left(\partial_{x_i}+V_i(x)\right)(1-x_iw_k)^{-1},\nonumber\ee
which corresponds to the operator $\mathscr{D}$ of \eqref{eqdefD} with $h(x_i)=1/(1-x_iw_k)$.  According to Lemma \ref{operatorD}, $\langle\mathscr{C}_{w_k}\Pi\rangle=0$.  This yields
\begin{multline}\label{eqC} %\langle\mathscr{C}_{w_k}\Pi\rangle=0\quad\Longrightarrow\quad
 \left\langle\sum_i\frac{w_k(1-\sigma_k)}{(1-x_iw_k)^2}\Pi\right\rangle-\left\langle\sum_{i,\ell\neq k}\frac{w_\ell\sigma_\ell}{(1-x_iw_k)(1-x_iw_\ell)}\Pi\right\rangle\\+\left\langle\sum_i\frac{1}{1-x_iw_k}\left(\frac{\lambda_1}{x_i}-\frac{\lambda_2}{1-x_i}+\sum_{j\neq i}\frac{2\lambda}{x_i-x_j}\right)\Pi\right\rangle=0.\end{multline}
After some simplifications, the substitution of \eqref{eqA} and \eqref{eqB} gives
\begin{multline}\label{eqCentral}
 (1-\sigma_k)\left\langle\sum_i\frac{1}{(1-x_iw_k)^2}\Pi\right\rangle+\lambda \left\langle\sum_{i\neq j}\frac{1}{(1-x_iw_k)(1-x_jw_k)}\Pi\right\rangle
\\
+\frac{1}{1-w_k}\sum_{\ell}\left\langle\sum_{i}\frac{x_i \sigma_\ell w_\ell}{1-x_iw_\ell}\Pi\right\rangle
-\sum_{\ell\neq k}\frac{w_\ell\sigma_\ell}{w_k-w_\ell}\left\langle\sum_{i}\left(\frac{1}{1-x_iw_k}-\frac{1}{1-x_iw_\ell}\right)\Pi\right\rangle
\\-\frac{\lambda_3}{1-w_k}\left\langle \Pi\right\rangle +\left(\sigma_k+\lambda_1+\frac{\lambda_2}{1-w_k}\right)\left\langle\sum_{i}\frac{1}{1-x_iw_k}\Pi\right\rangle=0.\end{multline}

For the second part of the proof, we will rewrite the  last equation as a differential operator in the variables $w$ acting on $\left\langle\Pi\right\rangle$.  Little work first gives
\be
\label{actw1}
\frac{\partial }{\partial w_{k}}\left\langle\Pi\right\rangle=-\sigma_k \left\langle \sum\limits_{i}\frac{x_{i}}{1-x_{i}w_{k}} \Pi\right\rangle,
\ee
\be
\label{actw2}
w_k\frac{\partial }{\partial w_{k}}\left\langle\Pi\right\rangle=\sigma_k N\left\langle\Pi\right\rangle-\sigma_k \left\langle \sum_{i}\frac{1}{1-x_{i}w_{k}} \Pi\right\rangle,
\ee
and \be
\label{actw3}
\frac{\partial^2 }{\partial w_{k}^2}\left\langle\Pi\right\rangle=-\sigma_k(1-\sigma_k)\left\langle \sum_{i}\frac{x_i^2}{(1-x_{i}w_{k})^2} \Pi\right\rangle+\sigma_k^2 \left\langle \sum_{i\neq j}\frac{x_i x_j}{(1-x_{i}w_{k})(1-x_jw_k)} \Pi\right\rangle.
\ee
Then, by making use of the identity \[\frac{x_{i}x_{j}w_{k}^{2}}{(1-x_{i}w_{k})(1-x_{j}w_{k})}=\frac{1}{(1-x_{i}w_{k})(1-x_{j}w_{k})}-\frac{1}{1-x_{i}w_{k}}-\frac{1}{1-x_{j}w_{k}}+1, \]
and of \eqref{actw1}-\eqref{actw3}, one can prove
\begin{multline}\label{actw4}\frac{1}{\sigma_k}\left(w_{k}^2\frac{\partial^2 }{\partial w_{k}^2}+2(1-\sigma_k N)w_k\frac{\partial }{\partial w_{k}}-\sigma_k N(1-\sigma_k N)\right)\left\langle\Pi\right\rangle=\\ -(1-\sigma_k)\left\langle\sum_i\frac{1}{(1-x_iw_k)^2}\Pi\right\rangle+ \sigma_k \left\langle\sum_{i\neq j}\frac{1}{(1-x_iw_k)(1-x_jw_k)}\Pi\right\rangle
\end{multline}
The last equation is proportional for every $k$ to the two first terms of \eqref{eqCentral}  if $\sigma_k(1-\sigma_k) = -\lambda(1-\sigma_k)$.  As expected, the latter equation  has two solutions $\sigma_k=1$ and $\sigma_k=-\lambda$.  We choose
\be \label{eqchoices} \sigma_i:=\begin{cases}\phantom{-} 1, & i=1,\ldots,n;\\ -\lambda,& i=n+1,\ldots,n+m.\end{cases}\ee
This allows us to combine   \eqref{eqCentral} and \eqref{actw4} and get
\begin{multline}
\frac{1}{\sigma_k}\left(w_{k}^2\frac{\partial^2 }{\partial w_{k}^2}+2(1-\sigma_k N)w_k\frac{\partial }{\partial w_{k}}-\sigma_k N(1-\sigma_k N)\right)\left\langle\Pi\right\rangle\\
+\frac{1}{\lambda}\frac{\sigma_k}{1-w_k}\sum_{\ell}\left\langle\sum_{i}\frac{x_i \sigma_\ell w_\ell}{1-x_iw_\ell}\Pi\right\rangle
-\frac{1}{\lambda}\sum_{\ell\neq k}\frac{\sigma_k\sigma_{\ell} w_k}{w_k-w_\ell}\left\langle\sum_{i}\left(\frac{1}{1-x_iw_k}-\frac{1}{1-x_iw_{\ell}}\right)\Pi\right\rangle
\\-\frac{1}{\lambda}\frac{\lambda_3\sigma_k}{1-w_k}\left\langle \Pi\right\rangle +\frac{1}{\lambda}\left(\sigma_k+\lambda_1+\frac{\lambda_2}{1-w_k}\right)\left\langle\sum_{i}\frac{\sigma_k}{1-x_i w_k}\Pi\right\rangle=0.\end{multline}
Moreover, by substituting \eqref{actw1} and \eqref{actw2} into the last equation, we actually  find that  $\left\langle\Pi\right\rangle$ satisfies the desired equation as long as $\sigma_i=1/\rho_i$, which imposes $\lambda=\alpha$, and the parameters $a,b,c$ are given by \eqref{condabc}.
Finally, one easily verifies that the cancellation property follows from \eqref{actw1} and \eqref{eqchoices}.
\end{proof}

\subsection{System for contour integrals}

In the last theorem, the conditions on  the parameters $\lambda_1$, $\lambda_2$ and $\lambda$ were crucial for ensuring the convergence of the integral defined in \eqref{skanekoint} and its derivatives with respect to the graded variables $w_k$.  It is possible however to interpret differently the integral without affecting the holonomic system it satisfies.  For instance, in the case where $\lambda_1$ is a negative integer and $\lambda$ is a positive integer, then we can replace the volume integration over $[0,1]^N$ by repeated contour integrals such that  each variables $x_j$ follows a closed loop in the complex plane starting at 1 and  encircling 0 in the counterclockwise direction.  When $\lambda_2$ is not an integer, a branch cut is made from 1 to $\infty$.  This defines a new analytic function at $w=(0,\ldots,0)$:
\begin{multline} \label{deffunctionT}{T}_N(\lambda_1,\lambda_2,\lambda;w)=
\int_1^{(0^+)}\cdots \int_1^{(0^+)}\prod_{j=1}^N\prod_{i=1}^{n+m}(1-x_jw_i)^{1/\rho_i}D_{\lambda_1,\lambda_2,\lambda}(x) dx_1\cdots dx_N,\qquad {\lambda_1\in\mathbb{Z}_-},\end{multline}
where it is understood that $\Re \lambda_2\geq 0$ and $\lambda\in\mathbb{N}$.  In principle, the integral can be evaluated by using the residue theorem.

What should be stressed here about the multiple integrals defining  $T_N$ is the following: (1) all the partial derivatives in $w_k$ commute with the integral signs; and  (2) the integral of any ``total derivative" is zero, i.e.,
\be \int_1^{(0^+)}\frac{\partial}{\partial x_i}\Big(\prod_{j=1}^N\prod_{i=1}^{n+m}(1-x_jw_i)^{1/\rho_i}D_{\lambda_1,\lambda_2,\lambda}(x)\Big)dx_i=0.
\nonumber\ee
These two conditions are enough to guaranty that all the steps in the proof of the last theorem remain valid,  so we conclude that ${T}_N$ also satisfies the same deformed holonomic system as $S_N(\lambda,\lambda_1,\lambda_2)$.

For similar reasons, the following contour integrals defines another  function that is analytic at $w=(0,\ldots,0)$ and satisfies the same deformed holonomic system:
 \begin{multline} \label{deffunctionU}{U}_N(\lambda_1,\lambda_2,\lambda;w)=
\int_{0}^{(1^+)}\cdots \int_0^{(1^+)}\prod_{j=1}^N\prod_{i=1}^{n+m}(1-x_jw_i)^{1/\rho_i}D_{\lambda_1,\lambda_2,\lambda}(x) dx_1\cdots dx_N,\qquad {\lambda_2\in\mathbb{Z}_-},\end{multline}
where it is also understood that $\Re \lambda_1\geq 0$ and $\lambda\in\mathbb{N}$. If $\lambda_1$ is not an integer, we consider that each $x_j$-plane is cut along $(-\infty,0]$.

\subsection{System for  $K_N(\lambda_1,\lambda_2,\lambda;t;s)$}

Theorem \ref{teointsystem} does not apply directly to Kaneko-type integrals, such as \eqref{kanekoint}.  There is however a very simple way to determine the deformed holonomic system related to such integrals.  It relies on the following lemmas, which are easy to prove.

\begin{lemma}\label{inverselemma}
Suppose that   $G(w_1,\ldots,w_{m+n})$ is one solution of  the  deformed system
$$\rho_i\,\mathcal{S}\!\mathcal{L}^{(a,b,c)}_i \,F(w)=\, {ab}\,F(w).$$
Then
$\prod_{i=1}^{n+m}w_i^{-a/\rho_i}G(1/w_1,\ldots,1/w_{m+n})$ is a solution of the following deformed system:
\be \,\rho_i\,\mathcal{S}\!\mathcal{L}^{(a',b',c')}_i \,F(w)=\,{a'b'}\,F(w),\nonumber\ee
where $$a'=a,\qquad b'=a-c+1+(1/\alpha)(n-\alpha m-1),\qquad c'=a-b+1+(1/\alpha)(n-\alpha m-1).$$
\end{lemma}

\begin{lemma}\label{inverselemma2}
Suppose that   $G(w_1,\ldots,w_{m+n})$ satisfies the following cancellation property:
\be \left(\rho_i\frac{\partial F}{\partial w_i}-\rho_j \frac{\partial F}{\partial w_j}\right)_{w_i=w_j}=0, \nonumber\ee
whenever $w_i$ and $w_j$ do not share the same parity.  Then
$\prod_{i=1}^{n+m}w_i^{-a/\rho_i}G(1/w_1,\ldots,1/w_{m+n})$  satisfies the same cancellation property.
\end{lemma}

Now, let us define an integral that generalizes both the $\mu=1$ and $\mu=-\lambda$ Kaneko integrals:
\be\label{deformedKaneko}
K_N(\lambda_1,\lambda_2,\lambda;t;s)=\int_{[0,1]^N}\prod_{i=1}^N\frac{\prod_{j=1}^n(x_i-t_j)}{\prod_{k=1}^m(x_i-s_k)^\lambda}D_{\lambda_1,\lambda_2,\lambda}(x)dx
\ee
where it is assumed that $s_k\not\in[0,1]$.  If moreover $t_j\neq 0$, then we have
\be K_N(\lambda_1,\lambda_2,\lambda;t;s)=(-1)^{(n-\lambda m)N} \prod_{j=1}^nt_j^N\prod_{k=1}^m s_k^{-N\lambda}S_N\left(\lambda_1,\lambda_2,\lambda;\frac{1}{t_1},\ldots,\frac{1}{t_n};\frac{1}{s_1},\ldots,\frac{1}{s_m}\right).
\ee
In other words, $K_N(\lambda_1,\lambda_2,\lambda;w)=(-1)^{(n-\lambda m)N} \prod_iw_i^{-a/\rho_i}S_N\left(\lambda_1,\lambda_2,\lambda;1/w\right)$ with $a=-N$ and $\alpha=\lambda$, so we can apply Lemmas \ref{inverselemma}-\ref{inverselemma2} to Theorem \ref{teointsystem}.

\begin{theorem}Let $F(w):=K _N(\lambda_1,\lambda_2,\lambda;t;s)$  be the function defined in \eqref{deformedKaneko} with $s_k\in\mathbb{C}\setminus[0,1]$, $t_j\neq 0$ and $\lambda_1,\lambda_2,\lambda>0$. Assume
\begin{equation} \label{condabc2}\alpha=\lambda,\quad a=-N,\quad b=(N-m-1)+(1/\lambda)(\lambda_1+\lambda_2+n+1),\quad c=-m+(1/\lambda)(n+\lambda_1).\end{equation}   Then $F(w)$ satisfies the following deformed holonomic system of partial differential equations,
\be\label{Z2deformedsystem} \,(\rho_i\,\mathcal{S}\!\mathcal{L}_i-ab) \,F(w)=0,\qquad i=1,\ldots,n+m,\ee
and the cancellation property
\be \label{Z2cancellation}\left(\rho_i\frac{\partial F}{\partial w_i}-\rho_j \frac{\partial F}{\partial w_j}\right)_{w_i=w_j}=0, \ee
where it is understood that $w_i$ and $w_j$ do not share the same parity.  \end{theorem}

In order to extend the range of possible values for the variables $s_j$, and allow for instance $s_k\in[0,1)$, the contours of integration must be modified.  Let $\mathscr{C}$ denote a closed path in the complex plane starting at 1, encircling  the $m$ variables $s_j$ in the positive direction, and returning to 1.  Suppose $ \lambda_2,\lambda\in\mathbb{N}$ and $\lambda_1\not\in\mathbb{Z}$ (for instance we can suppose that $\lambda_1$ has a non zero imaginary part).  Then we define
\be\label{deformedKanekoC}
K^\mathscr{C}_N(\lambda_1,\lambda_2,\lambda;t;s)=\frac{1}{(e^{2\pi\mathrm{i}(\lambda_1-\lambda m)}-1)^N}\int_{\mathscr{C}}\ldots\int_{\mathscr{C}} \prod_{i=1}^N\frac{\prod_{j=1}^n(x_i-t_j)}{\prod_{k=1}^m(x_i-s_k)^\lambda}D_{\lambda_1,\lambda_2,\lambda}(x)dx_1\cdots dx_N.\ee

The function $K^\mathscr{C}_N(\lambda_1,\lambda_2,\lambda;t;0,\ldots,0)$ is in fact  an analytic continuation of the usual Kaneko integral $K_N(\lambda_1-\lambda m,\lambda_2,\lambda;t)$ with $\mu=1$.  Indeed, suppose as before $ \lambda_2,\lambda\in\mathbb{N}$, $\lambda_1\not\in\mathbb{Z}$, and additionally   $\Re \lambda_1-\lambda m>0$.  Then, using the classical methods of complex analysis (such as those related to  Hankel's contour for the Gamma function or  Pochhammer's contour for the Beta function; see also \cite[section 5]{bf}), we conclude that each $\int_{\mathscr{C}}f(x,t)dx_i$, where $f(x,t)$ denotes the integrand of \eqref{deformedKanekoC},   is equal to  $(e^{2\pi\mathrm{i}(\lambda_1-\lambda m)}-1)\int_{0}^1f(x,t)dx_i$.

\begin{coro}\label{corodeformedKanekoC}Let $F(w):=K^\mathscr{C} _N(\lambda_1,\lambda_2,\lambda;t;s)$  be the function defined in \eqref{deformedKanekoC} with $s_k\neq 1$, $ \lambda_2,\lambda\in\mathbb{N}$, $\lambda_1\not\in\mathbb{Z}$. Assume
\begin{equation} \label{condabc2}\alpha=\lambda,\quad a=-N,\quad b=(N-m-1)+(1/\lambda)(\lambda_1+\lambda_2+n+1),\quad c=-m+(1/\lambda)(n+\lambda_1).\end{equation}   Then $F(w)$ satisfies the  deformed holonomic system \eqref{Z2deformedsystem} and the cancellation property \eqref{Z2cancellation}.   \end{coro}

\section{Supersymmetric functions and associated hypergeometric series}
\label{superseries}
This section first provides a brief review of some aspects of symmetric polynomials and especially Jack polynomials.  The classical references  on the subject are Macdonald's book \cite{macdonald}
and Stanley's article \cite{stanley}.  This will allow us to introduce the  supersymmetric Jack polynomials \cite{sv2} and especially their associated hypergeometric functions \cite{dh}.  A few results proved here will be used later in the article.

We stress that in the following pages, we always assume $\mathbb{F}=\mathbb{Q}(\alpha)$, which means that $\mathbb{F}$ is equal to the field of rational function in the formal parameter $\alpha$.
\label{superfunction}
\subsection{Partitions}

A partition $\kappa = (\kappa_1,\kappa_2,\ldots,\kappa_i,\ldots)$ is a sequence of non-negative integers $\kappa_i$ such that
\begin{equation*}
    \kappa_1\geq\kappa_2\geq\cdots\geq\kappa_i\geq\cdots
\end{equation*}
and only a finite number of the terms $\kappa_i$ are non-zero. The number of non-zero terms is referred to as the length of $\kappa$, and is denoted $\ell(\kappa)$. We shall not distinguish between two partitions that differ only by a string of zeros. The weight of a partition $\kappa$ is the sum
\begin{equation*}
    |\kappa|:= \kappa_1+\kappa_2+\cdots
\end{equation*}
of its parts, and its diagram is the set of points $(i,j)\in\mathbb{N}^2$ such that $1\leq j\leq\kappa_i$.
Reflection in the diagonal produces the conjugate partition
$\kappa^\prime=(\kappa_1',\kappa_2',\ldots)$.

Let $\kappa, \sigma $ be partitions. We define $\kappa \cup  \sigma $ to be the partition whose parts are those of $\kappa$ and $\sigma $, arranged
in descending order. For example, if $\kappa=(321)$ and $\sigma= (21)$, then $\kappa \cup  \sigma=(32211)$.

The set of all partitions of a given weight are partially ordered
by the dominance order: $\kappa\leq \sigma $ if and only if $\sum_{i=1}^k\kappa_i\leq \sum_{i=1}^k \sigma_i$ for all $k$.
One easily verifies that $\kappa\leq\sigma$ if and only if $\sigma'\leq\kappa'$. We shall also require the inclusion order on the set of
all partitions, defined by $\sigma\subseteq\kappa$ if and only if $\sigma_i \leq\kappa_i$ for all $i$, or equivalently, if and only if the
diagram of $\sigma$ is contained in that of $\kappa$.

\subsection{Jack symmetric functions}

Let $\Lambda_N(x)$ denote the algebra of symmetric polynomials in $N$ indeterminates $x_1,\ldots,x_N$.  As a ring, it is generated by the power sums:
\be p_k(x):=x_1^k+\ldots+x_N^k.\nonumber\ee
The ring of symmetric polynomials is naturally graded: $\Lambda_N(x)=\oplus_{k\geq 0}\Lambda^k_N(x)$, where $\Lambda^k_N(x)$ denotes the sets of homogeneous polynomials of degree $k$.   As a vector space, $\Lambda^{k}_N(x)$ is equal to the span over $\mathbb{F}$ of all symmetric monomials $m_\kappa(x)$, where $\kappa$ is a partition of weight $k$ and
\be  m_\kappa(x):=x_1^{\kappa_1}\cdots x_N^{\kappa_N}+\text{distinct permutations}.\nonumber
\ee
Note that if the length of the partition $\kappa$ is larger than $N$,  we set $m_\kappa(x)=0$.

The whole ring $\Lambda_N(x)$ is invariant under the action of  homogeneous differential operators related to the Calogero-Sutherland models \cite{bf}:
\be E^k=\sum_{i=1}^N x_i^k\frac{\partial}{\partial x_i},\qquad D^k=\sum_{i=1}^N x_i^k\frac{\partial^2}{\partial x_i^2}+\frac{2}{\alpha}\sum_{1\leq i\neq j \leq N}\frac{x_i^k}{x_i-x_j}\frac{\partial}{\partial x_i},\qquad k\geq 0.
\ee
The operators $E^1$ and $D^2$ are special since they also preserve each $\Lambda^k_N(x)$.  They can be used to define the Jack polynomials.  Indeed, for each partition $\kappa$, there exists a unique symmetric polynomial $P^{(\alpha)}_\kappa(x)$ that satisfies the the two following conditions \cite{stanley}:
\begin{align} \label{Jacktriang}(1)\qquad&P^{(\alpha)}_\kappa(x)=m_\kappa(x)+\sum_{\mu<\kappa}c_{\kappa\mu}m_\mu(x)&\text{(triangularity)}\\
\label{Jackeigen}(2)\qquad &\left(D^2-\frac{2}{\alpha}(N-1)E^1\right)P^{(\alpha)}_\kappa(x)=\epsilon_\kappa P^{(\alpha)}_\kappa(x)&\text{(eigenfunction)}\end{align}
where the coefficients $\epsilon_\kappa$ and  $c_{\kappa\mu}$ belong to $\mathbb{F}$.  As a condition of the tiangularity condition, $\Lambda_N(x)$ is equal to the span over $\mathbb{F}$ of all Jack polynomials  $P^{(\alpha)}_\kappa(x)$, with $\kappa$ a partition of length less or equal to $N$.

Following the standard nomenclature of algebraic combinatorics\cite{macdonald}, when the number $N$ of indeterminates is infinite, we say  ``symmetric functions" instead of ``symmetric polynomials" (although in general the ``functions" cannot be considered as mapping from a set to another).   The ring of symmetric functions is given by $\Lambda=\oplus_{k\geq 0}\Lambda^k$, where $\Lambda^k$ is equal to the inverse (or projective) limit of $\Lambda^k_N$.   In other words, if $z=(z_1,z_2,\ldots)$ stands for an infinite set of intederminates, each element of $\Lambda(z)$ is equal to a \emph{finite} linear combination of monomial symmetric {functions} $m_\kappa(z)$ with coefficients in $\mathbb{F}$, where
\be m_\kappa(z)=(m_\kappa(z_1),m_\kappa(z_1,z_2),m_\kappa(z_1,z_2,z_3),\ldots).\nonumber\ee
The addition  and the multiplication in $\Lambda(z)$  acts component by component. Similarly, we can define the Jack symmetric function $P^{(\alpha)}_\kappa(z)$ as an inverse limit of polynomials:
\be \label{eqdefJackfunct}P^{(\alpha)}_\kappa(z)=(P^{(\alpha)}_\kappa(z_1),P^{(\alpha)}_\kappa(z_1,z_2),P^{(\alpha)}_\kappa(z_1,z_2,z_3),\ldots).\ee
This object is well defined since the coefficients of the Jack polynomials are stable with respect to the number $N$ of variables: the coefficients $c_{\kappa\mu}$ in \eqref{Jacktriang} do not depend upon $N$.  The  power sum polynomials also share this stability, so we can  form the $k$-th power sum function  $p_k(z)$ similarly.   Thus, the algebra $\Lambda(z)$ of symmetric functions can  be seen as the $\mathbb{F}$-span of  either $\{m_\kappa(z)\}_\kappa$, $\{P^{(\alpha)}_\kappa(z)\}_\kappa$, or $\{p_\kappa(z)\}_\kappa$, where $\kappa$ runs through the set of all partitions and
\be p_\kappa(x):=p_{\kappa_1}(z)\cdots p_{\kappa_\ell}(z).\nonumber\ee

\subsection{Super-symmetric functions}

Let $z=(z_{1}, z_{2}, \ldots)$, $x=(x_{1}, x_{2}, \ldots)$ and $y=(y_{1}, y_{2}, \ldots)$ be three  sequences  of independent indeterminates.  As explained above, $\Lambda(z)$ is  the graded algebra consisting of symmetric functions over $\mathbb{F}$.  The algebra of bi-symmetric functions is given by $\Lambda(x)\otimes \Lambda(y)$, that is, each element $f(x,y)$ of this algebra  is symmetric in $x$ and  $y$ separately.

Now let $S\!\Lambda(x,y)$ be the subalgebra of $\Lambda(x)\otimes \Lambda(y)$ formed by the bi-symmetric functions $f(x,y)$ that satisfy  the following cancellation conditions:
\begin{equation}\label{cancellationcondition}
 \left(\frac{\partial f}{\partial x_i}+\frac{1}{\alpha}
 \frac{\partial f}{\partial y_j} \right)_{x_i=y_j}=0, \forall i,j \ \in \mathbb{N}.
\end{equation}
Each element of $S\!\Lambda(x,y)$ is called a supersymmetric function.  It is very easy to check that  the  homogeneous elements of $S\!\Lambda(x,y)$ also satisfy (\ref{cancellationcondition}), so we have the grading
$
S\!\Lambda=\mathop{\oplus}_{k\geq 0}S\!\Lambda^{k}$.
 Moreover, there is a natural ring homomorphism $\varphi:\,\Lambda(z)\to S\!\Lambda(x,y)$ such that
\be{ \label{homomorphismproject}
\varphi(p_{k}(z))=p_{k,\alpha}(x,y):=p_{k}(x)-{\alpha}p_k(y).
}\ee
Such a homomorphism was first used by Kerov et al.\ in \cite{koo}.  We call $p_{k,\alpha}(x,y)$ the $k$-th deformed power sum (in infinitely many variables).  Let $\mathcal{N}_{\alpha}(x,y)$ denote  the algebra generated by these power sums. In other words,
\be  \mathcal{N}_{\alpha}(x,y)=\mathrm{Span}_\mathbb{F}\{ p_{\kappa, \alpha}(x,y)\}_\kappa\nonumber\ee
where
\be p_{\kappa, \alpha}(x,y)=\prod_{\kappa_{i}>0}p_{\kappa_{i}, \alpha}(x,y).\nonumber\ee

 Obviously, we have $\mathcal{N}_{\alpha}(x,y)\subseteq S\!\Lambda(x,y)$. As proved below, the algebra of supersymmetric functions coincide with the algebra $\mathcal{N}_{\alpha}(x,y)$ for specific choices of the parameter $\alpha$ and in particular, for all real $\alpha>0$, which is the relevant case for the study of the deformed Selberg integral .

\begin{lemma}If $\alpha$ is neither a negative rational number nor zero, then $S\Lambda(x,y)$ is generated by the deformed power sums $p_{r,\alpha}(x,y),  r\in \mathbb{N}$.
\end{lemma}

\begin{proof}

Since the $p_{r}(x)$, with $r\geq 1$, are algebraically independent over $\mathbb{F}$, so are the $p_{r,\alpha}(x,y)$ of $\mathcal{N}_{\alpha}(x,y)$. The mapping $\varphi$ defined in \eqref{homomorphismproject} is then an isomorphism between $\Lambda(z)$ and $\mathcal{N}_{\alpha}(x,y)$.
 As we have already pointed out  $\mathcal{N}_{\alpha}(x,y)\subseteq S\!\Lambda(x,y)$.  Thus, to show they are actually the same it is sufficient to prove that the dimension of the homogeneous component of degree $k$ of $\mathcal{N}_{\alpha}(x,y)$ is not greater than the number of partitions $\kappa$ of $k$.

 For $n, m \in \mathbb{N}_{0}=\{0,1,2,\ldots\}$, set $x=(x_{1},  \ldots, x_{n})$ and $y=(y_{1}, \ldots, y_{m})$. Let us denote by $S\!\Lambda_{n,m}(x,y)$ the subalgebra of $\Lambda_{n}(x)\otimes \Lambda_{m}(y)$ consisting of bi-symmetric polynomials $f_{n,m}(x,y)$ which satisfy the cancellation conditions. For given $k$, taking $n,m\geq k$, we define a homomorphism $\phi_{n,m}$ from $S\Lambda(x,y)$ to $S\Lambda_{n,m}(x,y)$ by setting $x_{n+1}=\cdots =0, y_{m+1}=\cdots=0$, that is, \be \label{homomorphismevaluate}\phi_{n,m}(f)=f(x_{1},  \ldots, x_{n}, y_{1},  \ldots, y_{n}),\qquad \forall f \in S\!\Lambda(x,y).\ee

Since the  homogeneous components of degree $k$ of the rings  $\Lambda(x)\otimes \Lambda(y)$ and  $\Lambda_{n}(x)\otimes \Lambda_{m}(y)$ are isomorphic,  the homomorphism $\phi_{n,m}$,  once restricted to the homogeneous components of degree $k$, is injective whenever $k\leq\mathrm{min}(n,m)$. When  $\alpha$ is neither a negative rational number nor zero, by Proposition 2 of \cite{sv}, the dimension of the  homogeneous component of degree $k$ of $S\!\Lambda_{n,m}(x,y)$ is not greater than the number of partitions $\kappa$ of $k$, which implies that it is also true for $S\!\Lambda(x,y)$.
\end{proof}

Let us go back to the  homomorphism \eqref{homomorphismevaluate}, which restricts the number of indeterminates. If we combine with the
 map $\varphi$ of (\ref{homomorphismproject}), we obtain a homomorphism between the symmetric functions and the super-symmetric polynomials:
\be \varphi_{n,m}=\phi_{n,m}\circ\varphi:\,\Lambda(z)\longrightarrow S\Lambda_{n,m}(x,y).\ee
The image of a Jack symmetric function, which is defined by \eqref{Jacktriang}-\eqref{Jackeigen} and \eqref{eqdefJackfunct}, under the action of  $\varphi_{n,m}$   is called a super Jack polynomial  \cite{okounkov,sv2}. We write
\be S\!P^{(\alpha)}_\kappa(x_1,\ldots,x_n;y_1,\ldots,y_m)=\varphi_{n,m}\big(P^{(\alpha)}_\kappa(z)\big).
\ee Likewise, we call the image of a Jack symmetric function under $\varphi$ a super Jack symmetric function.

Let $H_{n,m}$ be the set of partitions which consists of the partitions $\kappa$ such that $\kappa_{n+1}\leq m$.  In other words, $\kappa\in H_{n,m}$ if its diagram does not contain the square $(n+1,m+1)$.  Recently, Sergeev and Veselov have proved the following very important properties of super Jack polynomials (Theorem 2, \cite{sv2}): If  $\alpha$ is neither a negative rational number nor zero, then the kernel of $\varphi_{n,m}$ is spanned by the Jack symmetric functions $P_{\kappa}^{(\alpha)}(z)$ with $(n+1,m+1)\in \kappa$. Moreover, the polynomials
\[
S\!P_{\kappa}^{(\alpha)}(x_{1},  \ldots, x_{n}; y_{1},  \ldots, y_{n}), \qquad  \kappa \in H_{n,m},
\]
form a linear basis of  $S\!\Lambda_{n,m}$.

In the present article we will treat our objects in the  super-symmetric space $S\Lambda$(or $S\Lambda_{n,m}$) or as the associated super series given below, in particular when $\alpha>0$.
\subsection{Super hypergeometric series}\label{subsecseries}
Recall that the arm-lengths and leg-lengths of the box $(i,j)$ in the partition $\kappa$  are respectively given by
\begin{equation}\label{lengths}
    a_\kappa(i,j) = \kappa_i-j\qquad\text{and}\qquad l_\kappa(i,j) = \kappa^\prime_j-i.
\end{equation}We define the hook-length of a partition $\kappa$  as the following product:
\begin{equation} \label{defhook}
    h_{\kappa}^{(\alpha)}=\prod_{(i,j)\in\kappa}\Big(1+a_\kappa(i,j)+\frac{1}{\alpha}l_\kappa(i,j)\Big).
\end{equation}
Closely related is the following $\alpha$-deformation of the Pochhammer symbol:
\begin{equation}\label{defpochhammer}
    [x]^{(\alpha)}_\kappa = \prod_{1\leq i\leq \ell(\kappa)}\Big(x-\frac{i-1}{\alpha}\Big)_{\kappa_i} = \prod_{(i,j)\in\kappa}\Big(x+a^\prime_\kappa(i,j)-\frac{1}{\alpha}l^\prime_\kappa(i,j)\Big)
\end{equation}
In the middle of the last equation,  $(x)_n\equiv x(x+1)\cdots(x+n-1)$ stands for the ordinary Pochhammer symbol, to which $\lbrack x\rbrack^{(\alpha)}_\kappa$ clearly reduces for $\ell(\kappa)=1$.  The left-hand side of \eqref{defpochhammer} involves the co-arm-lengths and co-leg-lengths box $(i,j)$ in the partition $\kappa$, which are respectively defined as
\begin{equation}\label{colengths}
    a^\prime_\kappa(i,j) = j-1,\qquad \text{and}\qquad l^\prime_\kappa(i,j) = i-1.
\end{equation}
By looking at the formulas above, one readily proves the following.
\begin{lemma}\label{lemmapochhammer} Let $M$ be a positive integer.  Let $\kappa$ be a partition of length $\ell$ and  first part equal to $\kappa_1$.  Then, for all $\alpha>0$,
$ [-M]^{(\alpha)}_\kappa\neq 0$ {if} $\kappa_1\leq M$, {while} $[-M]^{(\alpha)}_\kappa=0$ {if} $\kappa_1> M$.
Moreover, for all $\alpha>0$,
$[M/\alpha]^{(\alpha)}_\kappa\neq 0$ {if} $\ell\leq M$, {while} $[M/\alpha]^{(\alpha)}_\kappa=0$ {if} $\ell> M$ \end{lemma}

We are now ready to give the  precise definition of the super hypergeometric series in question.
\begin{definition}\label{HFpqDef}
Fix $p,q\in\mathbb{N}_0$ and let $a_1,\ldots,a_p, b_1,\ldots,b_q$ be complex numbers such that $(i-1)/\alpha-b_j\notin\mathbb{N}_0$ for all $i\in\mathbb{N}_0$.
We then define the $(p,q)$-type super hypergeometric series  by the following formal power series:
\begin{equation}\label{hfpq}
    {}_pS\!F^{(\alpha)}_q(a_1,\ldots,a_p;b_1,\ldots,b_q;x,y) = \sum_{k=0}^{\infty}\sum_{|\kappa|=k} \frac{1}{h_{\kappa}^{(\alpha)}}\frac{\lbrack a_1\rbrack^{(\alpha)}_\kappa\cdots\lbrack a_p\rbrack^{(\alpha)}_\kappa}{\lbrack b_1\rbrack^{(\alpha)}_\kappa
    \cdots\lbrack b_q\rbrack^{(\alpha)}_\kappa}S\!P_{\kappa}^{(\alpha)}(x;y).
\end{equation}
\end{definition}

When we are given finitely many variables, $x=(x_1,\ldots,x_n)$ and $y=(y_1,\ldots,y_m)$, the above series coincide with the super hypergeometric functions introduced in \cite{dh}.  The latter were given as the image  of  formal power series in Jack symmetric functions under the action of the restriction homomorphism $\varphi_{n,m}$.  If we suppose moreover that $y=(0,\ldots,0)$, then we recover
the (generalised)
hypergeometric functions studied, in particular, by Kor\'anyi \cite{koranyi}, Yan \cite{yan}, and Kaneko \cite{kaneko}.   ${}_pS\!F^{(\alpha)}_q$ will be denoted by ${}_pF^{(\alpha)}_q$ whenever $y=(0,0,\ldots)$.

Some peculiarities of the series ${}_pS\!F_q$ with $n+m$ variables deserve to be mentioned.  First, there is always an infinite number of terms in the series \eqref{hfpq} that do not contribute to  ${}_pS\!F_q$ since $S\!P_{\kappa}(x,y)\equiv 0$ unless $\kappa\in H_{n,m}$.  Second, the super series ${}_pS\!F_q$ does not terminate as easily as the usual ${}_pF_q$ .  To understand this point suppose that in \eqref{hfpq}, we have $y=(0,\ldots,0)$, $q>1$, and $a_1=-M$ for some integer $M$.  According to Lemma \ref{lemmapochhammer}, no partition with $\kappa_1>M$ will contribute to the series.  Moreover, the Jack polynomials $P_\kappa(x_1,\ldots,x_n)\equiv 0$ for all partitions $\kappa$ with length $\ell>n$.  This means that for the case with some  $a_i$ equal to a negative integer, the usual hypergeometric series  ${}_pF_q$ is actually a polynomial.  However, a super Jack polynomial  $S\!P_\kappa(x_1,\ldots,x_n,y_1,\ldots,y_m)\not= 0$  even if the $\ell(\kappa)>n+m$, so the series ${}_pS\!F_q$  with some  $a_i$ equal to a negative integer does not terminate in general.

The following proposition shows that the integral formulas relating ${}_pS\!F_q$ and ${}_{p+1}S\!F_{q+1} $ series also have to be treated slightly differently than for the usual hypergeometric series in many variables, which are due to Yan for general $\alpha>0$ \cite{yan}.  These new formulas   are expressed in terms of series  involving both Jack and super Jack polynomials:
\begin{multline}\label{hfpq2}
    {}_p\mathcal{S}\!\mathcal{F}^{(\alpha)}_q(a_1,\ldots,a_p;b_1,\ldots,b_q;z_1,\ldots,z_\ell;x,y) \\
		= \sum_{\kappa} \frac{1}{h_{\kappa}^{(\alpha)}}\frac{\lbrack a_1\rbrack^{(\alpha)}_\kappa\cdots\lbrack a_p\rbrack^{(\alpha)}_\kappa}{\lbrack b_1\rbrack^{(\alpha)}_\kappa
    \cdots\lbrack b_q\rbrack^{(\alpha)}_\kappa}\frac{P_{\kappa}^{(\alpha)}(z_1,\ldots,z_\ell)S\!P_{\kappa}^{(\alpha)}(x;y)}{P_{\kappa}^{(\alpha)}(1^\ell)}.
\end{multline}
Note that the above sums makes sense only if  it extends over all partitions length less of equal to $\ell$.

\begin{proposition}\label{propintsuperseries}Let $a_1,\ldots,a_{p+1}, b_1,\ldots,b_{q+1}$ be complex numbers such that $(i-1)/\alpha-b_j\notin\mathbb{N}_0$ for all $i\in\mathbb{N}_0$.
Suppose moreover that for some positive integer $\ell$ and  some $1\leq i\leq p+1$, we have
$$a_{p+1}=\lambda_1+1+\frac{\ell-1}{\alpha},\qquad  b_{q+1}=\lambda_1+\lambda_2+2+2\frac{\ell-1}{\alpha} \qquad\text{and}\qquad a_i=\frac{\ell}{\alpha}.$$  Then, formally
\begin{multline}   \phantom{F^{(\alpha)}}_{p+1}{S}\!{F}^{(\alpha)}_{q+1}(a_1,\ldots,a_p;b_1,\ldots,b_q;x,y)\\
=\frac{1}{S_{\lambda_1,\lambda_2,1/\alpha}}\int_{[0,1]^\ell}{}_p\mathcal{S}\!\mathcal{F}^{(\alpha)}_q(a_1,\ldots,a_p;b_1,\ldots,b_q;z_1,\ldots,z_\ell;x,y) D_{\lambda_1,\lambda_2,1/\alpha}(z_1,\ldots,z_\ell)d^\ell x.
\end{multline}
 \end{proposition}
\begin{proof}Here we treat ${}_p\mathcal{S}\!\mathcal{F}^{(\alpha)}_q$ as a formal power series, so we suppose that  it can be integrated term by term.  In fact, the proposition is almost a direct consequence of Kadell's integral formula \cite{kadell,kaneko}  (see (12.143) in \cite{forrester}):
\begin{multline}\label{jackintegral}
\frac{1}{S_{\ell}(\lambda_{1},\lambda_{2},1/\alpha)}\int_{[0,1]^{\ell}}
P_{\kappa}^{(\alpha)}(x_{1},\ldots,x_{\ell})\, D_{\lambda_{1},\lambda_{2},1/\alpha}(x_{1},\ldots,x_{\ell})dx\\
=  P_{\kappa}^{(\alpha)}(1^{\ell}) \frac{\lbrack \lambda_1+1+(\ell-1)/\alpha\rbrack^{(\alpha)}_\kappa}{\lbrack \lambda_1+\lambda_2+2+2(\ell-1)/\alpha\rbrack^{(\alpha)}_\kappa}.\end{multline}
We see that the integration of $P_{\kappa}^{(\alpha)}(x)/P_{\kappa}^{(\alpha)}(1^{\ell})$ precisely gives the factors $\lbrack a_{p+1} \rbrack^{(\alpha)}_\kappa$ and $\lbrack b_{q+1} \rbrack^{(\alpha)}_\kappa$ that should appear in ${}_{p+1}{S}\!{F}^{(\alpha)}_{q+1}$.  Given that the series${}_p\mathcal{S}\!\mathcal{F}^{(\alpha)}_q$ does not contain terms associated to partitions of length greater than $\ell$, the integrated series can not contain such terms either.  Recalling Lemma \ref{lemmapochhammer},  we conclude that the integrated series is indeed equal  to ${}_{p+1}{S}\!{F}^{(\alpha)}_{q+1}$ if at least one the parameters $a_i$ of the latter series is equal to  $\ell/\alpha$. \end{proof}

In the subsequent sections, we will pay special attention to the hypergeometric series ${}_0S\!F_0, {}_1S\!F_0, {}_0S\!F_1, {}_1S\!F_1$ and ${}_2S\!F_1$.  As for the classical Gaussian hypergeometric function, these series satisfy simple systems of differential equations of second order.

\section{Existence and uniqueness theorems for holonomic systems} \label{sectsystem}
\setcounter{equation}{0}

The  aim of this section is to establish existence and uniqueness for solution of the deformed holonomic system with finitely or infinitely many variables.  We will prove that the super hypergeometric series ${}_2SF^{(\alpha)}_1(a, b;c; t, s)$ is the desired solution.

\subsection{Holonomic system with infinitely many variables}\label{SectHoloInfinite}

In this subsection we study the special case of the system \eqref{PDEs1} and \eqref{PDEs2} with $n=\infty$ and $m=0$, which corresponds to a non-deformed system but with infinitely many variables:
\begin{multline}\label{holonomicsystemforinfinite}
y_i (1 - y_i) \frac{\partial^2 F}{\partial y_i^2} + \Bigl( c  - \bigl(a + b + 1 \bigr) y_i \Bigr)
\frac{\partial F}{\partial y_i} - abF   \\
+\frac{1}{\alpha}\sum_{k=1, k \ne i}^\infty
\frac{y_k}{y_i - y_k} \biggl( (1 - y_i)
 \frac{\partial F}{\partial y_i}
- (1 - y_k) \frac{\partial F}{\partial y_k} \biggr)
 = 0, \qquad i=1, 2, \ldots.
\end{multline}

The finite case $n<\infty$ is already well understood. Indeed, by generalizing Muirhead's result \cite{muirhead} for the case $\alpha=2$,  Kaneko  \cite{kaneko} and Yan  \cite{yan} independently
showed that ${}_2F^{(\alpha)}_1(a,b;c;y_{1}, \ldots, y_{n})$  is  the unique solution $F(y_1, \ldots, y_n)$ to the following system that is symmetric in $y_1, \ldots, y_n$ and
 analytic at $(0,\ldots,0)$ with value $F(0, \ldots, 0)=1$:
\begin{multline}\label{holonomicsystemforfinite}
y_i (1 - y_i) \frac{\partial^2 F}{\partial y_i^2} + \Bigl( c -
\frac{1}{\alpha}(n - 1) - \bigl(a + b + 1 -
\frac{1}{\alpha}(n - 1)\bigr) y_i \Bigr)
\frac{\partial F}{\partial y_i} - abF   \\
+\frac{1}{\alpha}\sum_{k=1, k \ne i}^n
\frac{1}{y_i - y_k} \biggl(y_i(1 - y_i)
 \frac{\partial F}{\partial y_i}
- y_k(1 - y_k) \frac{\partial F}{\partial y_k} \biggr)
 = 0,\qquad i=1, \ldots,n.
\end{multline}
The systems  \eqref{holonomicsystemforinfinite} and  \eqref{holonomicsystemforfinite} are easily shown to be equivalent when the number of variables is finite.   However, in the infinite case, only the former makes sense.

We first demonstrate the following uniqueness property. Its proof can be  given in a similar way  to that of the finite case  by Muirhead \cite{muirhead1} and Kaneko  \cite{kaneko}.

\begin{theorem}
\label{uniquenesstheorem}
Assume that $c-(1/\alpha)(i-1)$ is neither negative integers nor zero for any $i\in \mathbb{N}$. Then each of the infinitely many differential equations in the system
(\ref{holonomicsystemforinfinite}) has the same unique formal power series solution $F(y)$ subject to the conditions

  (a) $F$ is a symmetric function of $\{y_i\}_{i=1}^{\infty}$, and

 (b) $F$ has a formal power series expansion at $y=(0,0,\ldots)$ with $F(0,0\ldots)=1$.

\end{theorem}

\begin{proof}
We first transform (\ref{holonomicsystemforinfinite}) to partial differential equations in terms of elementary symmetric functions  $\{r_j\}_{j=1}^{\infty}$, which are defined as
$$r_j=\sum_{i_1<\ldots<i_j}y_{i_1}\cdots y_{i_j}.$$
 Let $r_{j}^{(i)}$ denote the $j$th elementary symmetric function formed from the variables $\{y_j\}_{j=1}^{\infty}\backslash \{y_i\}$. Then we clearly have
\be
\label{elementaryrelation}
r_{j}=y_i r_{j-1}^{(i)}+r_{j}^{(i)},\qquad  j=1, 2, \ldots,\qquad \text{with}\qquad r_0=r_{0}^{(i)}=1.
\ee
Given that
  $$ \frac{\partial }{\partial y_i}=\sum_{\nu=1}^{\infty}
 r_{\nu-1}^{(i)}\frac{\partial }{\partial r_\nu}\qquad\text{and}\qquad
 \frac{\partial^{2} }{\partial y_i^{2}}=\sum_{\mu,\nu=1}^{\infty}
r_{\mu-1}^{(i)} r_{\nu-1}^{(i)}\frac{\partial^{2} }{\partial r_\mu\partial r_\nu},$$
we easily see that the differential equation in \eqref{holonomicsystemforinfinite} is equal to
\begin{multline}\label{aftertransform1}
\sum_{\mu,\nu=1}^{\infty}y_i (1 - y_i)
r_{\mu-1}^{(i)} r_{\nu-1}^{(i)}\frac{\partial^{2} F }{\partial r_\mu\partial r_\nu}+\\
 \sum_{k=1}^{\infty}\Bigl\{( c  - (a + b + 1) y_i)r_{k-1}^{(i)}+\frac{1}{\alpha}
\sum_{\substack{j=1 \\ j \ne i}}^\infty
\frac{y_j}{y_i - y_j} \biggl( (1 - y_i)r_{k-1}^{(i)}
- (1 - y_j)r_{k-1}^{(j)} \biggr)
 \Bigr\}
\frac{\partial F}{\partial r_k} - abF=0.
\end{multline}

Now, for $i\neq j$, let us denote by  $r_{k}^{(i,j)}$ the $k$th elementary symmetric function in the variables $\{y_k\}_{k=1}^{\infty}\backslash \{y_i, y_j\}$.  We set $r_{k}^{(i,j)}=0$ if $k<0$.
By successively making use of the relation \eqref{elementaryrelation}, we get
\begin{multline}
\sum_{j=1, j \ne i}^\infty
\frac{y_i}{y_i - y_j} \biggl( (1 - y_i)r_{k-1}^{(i)}
- (1 - y_j)r_{k-1}^{(j)} \biggr)
=
\sum_{j=1, j \ne i}^\infty
\frac{y_i}{y_i - y_j} \biggl( r_{k-1}^{(i)}-r_{k-1}^{(j)}+r_{k}^{(i)}-r_{k}^{(j)} \biggr)\qquad\qquad \qquad\\ \nonumber
=-\sum_{j=1, j \ne i}^\infty
y_i  \biggl( r_{k-2}^{(i,j)}+r_{k-1}^{(i,j)}\biggr)
=
-\sum_{j=1, j \ne i}^\infty
 \biggl(r_{k-1}^{(j)} - r_{k-1}^{(i,j)}+r_{k}^{(j)} -r_{k}^{(i,j)}\biggr).
\end{multline}
Writing $y_j$ as $y_j-y_i+y_i$, we then find
\begin{multline}
\sum_{j=1, j \ne i}^\infty
\frac{y_j}{y_i - y_j} \biggl( (1 - y_i)r_{k-1}^{(i)}
- (1 - y_j)r_{k-1}^{(j)} \biggr)\\ \nonumber
=-\sum_{j=1, j \ne i}^\infty
 \biggl(r_{k-1}^{(j)} - r_{k-1}^{(i,j)}+r_{k}^{(j)} -r_{k}^{(i,j)}\biggr)-\sum_{j=1, j \ne i}^\infty
 \biggl( r_{k-1}^{(i)}-r_{k-1}^{(j)}+r_{k}^{(i)}-r_{k}^{(j)} \biggr)\nonumber\\
 =-\sum_{j=1, j \ne i}^\infty
 \biggl(r_{k-1}^{(i)} - r_{k-1}^{(i,j)}+r_{k}^{(i)} -r_{k}^{(i,j)}\biggr)
=-(k-1)r_{k-1}^{(i)}-k r_{k}^{(i)}\nonumber,
\end{multline}
where we have used
\be \label{cancelsum}
\sum_{j=1, j \ne i}^\infty
 \biggl(r_{k}^{(i)} -r_{k}^{(i,j)}\biggr)=k r_{k}^{(i)}.
\ee
In fact, given a sequence $i_{1}<i_{2}<\dots<i_{k}$  such that $i_{l}\neq i, l=1, 2, \ldots, k$, the factor $y_{i_{1}}y_{i_{2}}\dots y_{i_{k}}$ occurs exactly $k$
times in the summation of the left-hand side of \eqref{cancelsum}. This can be made more precise: $y_{i_{1}}y_{i_{2}}\dots y_{i_{k}}$ does not disappear in the term $(r_{k}^{(i)} -r_{k}^{(i,j)})$ if and only if $j=i_{t}$ for some $t \in \{1,2, \ldots,k\}$.
Hence,
\begin{multline}\label{eqproofuniquediff}
\big( c  - (a + b + 1) y_i\big)r_{k-1}^{(i)}+\frac{1}{\alpha}
\sum_{j=1, j \ne i}^\infty
\frac{y_j}{y_i - y_j} \biggl( (1 - y_i)r_{k-1}^{(i)}
- (1 - y_j)r_{k-1}^{(j)} \biggr)\\
 =\big( c-\frac{1}{\alpha}(k-1)\big)r_{k-1}^{(i)}+\big( a+b+1-\frac{1}{\alpha}k\big)r_{k}^{(i)}
- (a+b+1)r_{k}.
\end{multline}

Further identities are needed for simplifying  the differential equation \eqref{aftertransform1}.  We note that
\begin{align*} y_i r_{\nu-1}^{(i)} r_{\mu-1}^{(i)}&=r_{\nu} r_{\mu-1}^{(i)}-r_{\nu}^{(i)} r_{\mu-1}^{(i)}\\
&=r_{\nu} r_{\mu-1}^{(i)}-r_{\nu}^{(i)}(r_{\mu-1} -y_{i}r_{\mu-2}^{(i)})\\
&=r_{\nu} r_{\mu-1}^{(i)}-r_{\nu}^{(i)}r_{\mu-1} +y_{i}r_{\nu}^{(i)}r_{\mu-2}^{(i)}.\end{align*}
By iterating this relation, we get
\begin{multline*}
y_i r_{\nu-1}^{(i)} r_{\mu-1}^{(i)}=r_{\nu} r_{\mu-1}^{(i)}-r_{\mu-1} r_{\nu}^{(i)}+r_{\nu+1} r_{\mu-2}^{(i)}-r_{\mu-2} r_{\nu+1}^{(i)}
+\cdots\\
+r_{\nu+\mu-1} r_{0}^{(i)}-r_{0} r_{\nu+\mu-1}^{(i)}+y_{i}r^{(i)}_{\nu+\mu-1} r_{-1}^{(i)}\qquad \qquad\end{multline*}
and $r_{-1}^{(i)}=0$. Furthermore,
\begin{multline}\label{eqproofuniquediff2}
y_i (1 - y_i) r_{\nu-1}^{(i)} r_{\mu-1}^{(i)}=\\
r_{\nu} (r_{\mu-1}^{(i)}-r_{\mu}+r_{\mu}^{(i)})-r_{\mu-1} (r_{\nu}^{(i)}-r_{\nu+1}+r_{\nu+1}^{(i)})
+
r_{\nu+1} (r_{\mu-2}^{(i)}-r_{\mu-1}+r_{\mu-1}^{(i)})\qquad\qquad \qquad\\
-r_{\mu-2} (r_{\nu+1}^{(i)}-r_{\nu+2}+r_{\nu+2}^{(i)})+\dots+r_{\nu+\mu-1}( r_{0}^{(i)}-r_{1}+r_{1}^{(i)})-r_{0} (r_{\nu+\mu-1}^{(i)}-r_{\nu+\mu}+r_{\nu+\mu}^{(i)})
\\=\sum_{k=1}^{\infty}a_{\mu\nu}^{(k)}(r_{k-1}^{(i)}-r_{k}+r_{k}^{(i)}),\end{multline}
where (see  James \cite{james}, p.372--373 for more details) $a_{\mu\nu}^{(k)}=a_{\nu\mu}^{(k)}$ and for $\mu\leq \nu$,
\begin{equation}
a_{\mu\nu}^{(k)} = \begin{cases}
r_{\mu+\nu-k} &\text{for}\  1\leq k\leq \mu \\
0  & \text{for} \ \mu <k\leq \nu\\
-r_{\mu+\nu-k} &\text{for} \ \nu< k\leq \mu+\nu \\
0  & \text{for} \ \mu+\nu <k\ \ \ k=1,2,\ldots .\\
\end{cases}
\end{equation}

Therefore, according to \eqref{eqproofuniquediff} and \eqref{eqproofuniquediff2}, the differential equation \eqref{aftertransform1}  can be rewritten as
\begin{multline}\label{aftertransform2}
\sum_{\mu,\nu=1}^{\infty}\Big\{\sum_{k=1}^{\infty}a_{\mu\nu}^{(k)}(r_{k-1}^{(i)}-r_{k}+r_{k}^{(i)})\Big\}
\frac{\partial^{2} F }{\partial r_\mu\partial r_\nu}+\\
 \sum_{k=1}^{\infty}\Bigl\{\big( c  - \frac{1}{\alpha}(k-1)\big)r_{k-1}^{(i)}+(a + b + 1-\frac{1}{\alpha} k) r_{k}^{(i)}-(a + b + 1) r_{k}
 \Bigr\}
\frac{\partial F}{\partial r_k} - abF=0.
\end{multline}
In the latter equation,  we can equate coefficients of $r_{k-1}^{(i)}$ to zero for $k=1,2,\ldots $, according to the following, which will be showed in Appendix \ref{Appendixsymmetrylemma}.

\begin{lemma}\label{symmetrylemma}
Assume that $\{y_j\}_{j=1}^{\infty}$ are infinitely many indeterminates. Let $\{r_j\}_{j=1}^{\infty}$ and $\{r^{(i)}_j\}_{j=1}^{\infty}$ be the elementary symmetric
functions of  $\{y_j\}_{j=1}^{\infty}$ and $\{y_j\}_{j=1}^{\infty}\backslash\{y_i\}$, respectively. If $\lambda_{0}(r),\lambda_{1}(r),\ldots,$ are formal power series of  $\{r_j\}_{j=1}^{\infty}$ such that
\be \label{linearinfintesum}
\sum_{k=0}^{\infty}\lambda_{k}(r) r^{(i)}_k =0,
\ee
then
$\lambda_{k}(r)=0$ for $k=0, 1, \ldots$.
\end{lemma}

Equating the coefficients of  $r_{k-1}^{(i)}$  in \eqref{aftertransform2} to zero, we obtain the system of partial differential equations
\begin{multline}\label{aftertransform3}
\sum_{\mu,\nu=1}^{\infty}(a_{\mu\nu}^{(k-1)}+a_{\mu\nu}^{(k)})
\frac{\partial^{2} F }{\partial r_\mu\partial r_\nu}+\big( c  - \frac{1}{\alpha}(k-1)\big)\frac{\partial F}{\partial r_k}+\big( a+b+1 - \frac{1}{\alpha}(k-1)\big)\frac{\partial F}{\partial r_{k-1}}\\
 -\delta_{1k}\Bigl\{
 \sum_{\mu,\nu=1}^{\infty}\big(\sum_{j=1}^{\infty}a_{\mu\nu}^{(j)}r_j\big)
\frac{\partial^{2} F }{\partial r_\mu\partial r_\nu}+(a+b+1)\sum_{j=1}^{\infty}r_j \frac{\partial F}{\partial r_j}+ab F
  \Bigr\}=0, k=1,2,\ldots .
\end{multline}
Note that by convention, $a_{\mu\nu}^{(0)}=0$ and $\frac{\partial F}{\partial r_0}=0$.

Let $\mathfrak{L}_{n}=\{(i_{1},\ldots,i_{n})| i_{1},\ldots,i_{n} \in \mathbb{N}_{0}\}$ and $\mathfrak{L}=\bigcup_{n=1}^{\infty}\mathfrak{L}_{n}$. We introduce the lexicographic ordering on $\mathfrak{L}$, that is, for $I=(i_{1},i_{2},\ldots,),J=(j_{1},j_{2},\ldots,)\in \mathfrak{L}$, saying $I\leq J$ if either $I= J$ or there exists $n\in \mathbb{N}$ such that $i_{k}-j_{k}=0$ for $k>n$ but $i_{n}-j_{n}<0$.
Now we write
\be \label{formalsum}
F(r)=\sum_{I\in \mathfrak{L}}\gamma(I)r^{I}
\ee
with $\gamma(0,\ldots,0,\ldots)=1$. Here $r^{I}=r_{1}^{i_{1}}r_{2}^{i_{2}}\dots$ is only a finite product since the components of the index $I$ are zero except for the first finite. Order coefficients $\gamma(I)$ according to the lexicographic ordering of the index $I$. To determine the coefficient $\gamma(i_{1},\ldots,i_{n},0,\ldots)$ where $i_{n}>0$, substituting (\ref{formalsum}) in (\ref{aftertransform3}) with $k=n$ and putting $r_{j}=0$ for $j>n$, we obtain a recurrence relation expresses $i_{n}\big(i_{n}-1+c-(1/\alpha)(n-1)\big)\gamma(i_{1},\ldots,i_{n},0,\ldots)$ in terms of coefficients of lower order. Since $i_{n}-1+c-(1/\alpha)(n-1)$ is not zero by assumption, one can iterate this reduction until one reaches $\gamma(0,\ldots,0,\ldots)$. Hence all coefficients $\gamma(I)$ of (\ref{formalsum}) are uniquely determined by the recurrence relations. On the other hand, each of the infinitely many differential equations in (\ref{aftertransform2}) gives rise to the same system (\ref{aftertransform3}), so it follows that each equation in the system (\ref{holonomicsystemforinfinite}) has the same unique solution $F$ subject to conditions (a) and (b). This completes the proof of Theorem \ref{uniquenesstheorem}.
\end{proof}

It remains to show that the hypergeometric series ${}_2F_{1}^{(\alpha)}(a,b; c; y)$ actually satisfies the
infinite holonomic system \eqref{holonomicsystemforinfinite}.
\begin{theorem}
\label{hypergfsolution}
 ${}_2F_{1}^{(\alpha)}(a,b; c; y)$ is the unique solution of   each of the infinitely many differential equations in the system
(\ref{holonomicsystemforinfinite}) subject to the conditions

  (a) $F$ is a symmetric function of $\{y_i\}_{i=1}^{\infty}$, and

 (b) $F$ has a formal power series expansion at  $y=(0,0,\ldots)$ with $F(0,0\ldots)=1$..

\end{theorem}
\begin{proof}
The series (\ref{formalsum}) can be rearranged as a series of Jack symmetric functions
\be
\label{jacksum}
F(y)=\sum_{k=0}^{\infty}\sum_{|\kappa|=k}\gamma_{\kappa}P_{\kappa}^{(\alpha)}(y).
\ee
It suffices to show that
$$\gamma_\kappa=\frac{[a]_{\kappa}^{(\alpha)}  [b]_{\kappa}^{(\alpha)}}{[c]_{\kappa}^{(\alpha)} h_\kappa^{(\alpha)}}$$

For any partition $\sigma$ and given $n>\max\{1,\ell(\sigma)\}$, putting $y_{j}=0$ for $j>n$, $F(y_{1},\ldots,y_{n},0,\ldots)=\sum_{k=0}^{\infty}\sum_{|\kappa|=k}\gamma_{\kappa}P_{\kappa}^{(\alpha)}(y_{1},\ldots,y_{n},0,\ldots)$ clearly satisfies the finite system \eqref{holonomicsystemforfinite}. The theorem finally follows from a result due to Yan \cite[Theorem 2.13]{yan}:
%\label{yantheorem}
\emph{There exists a unique sequence $\{A_\kappa\}$ with $A_{(0)}=1$ such that $F_{n}(y_1, \ldots, y_n)=\sum_{\kappa}A_\kappa P_{\kappa^{(\alpha)}}(y_1, \ldots, y_n)/h_\kappa^{(\alpha)}$ satisfies \eqref{holonomicsystemforfinite} for $n=2, 3, \ldots$. Moreover, $A_\kappa={[a]_{\kappa}^{(\alpha)}  [b]_{\kappa}^{(\alpha)}}/{[c]_{\kappa}^{(\alpha)} }$.
}
\end{proof}

By proceeding as in the proof of Theorem \ref{uniquenesstheorem}, one can get the following useful generalization.

\begin{theorem}
\label{generaluniquenesstheorem} Fix $p_{0}  \in\mathbb{N} \cup \{\infty\}$  and let $\alpha_{1},  \alpha_{2}, \beta_0, \beta_1, \gamma_0$ be complex numbers such that  $\beta_{0}-(\frac{i-1}{\alpha}-j)\alpha_{1}$ are never zero for any $1\leq i \leq p_{0}$, $j\in \mathbb{N}_{0}$. Then each of the countably many differential equations in the system
\begin{multline}\label{generalholonomicsystem}
y_i (\alpha_{1} + \alpha_{2}y_i) \frac{\partial^2 F}{\partial y_i^2} + (\beta_{0} + \beta_{1}y_i)
\frac{\partial F}{\partial y_i} +\gamma_{0}F   \\
+\frac{1}{\alpha}\sum_{k=1, k \ne i}^{p_{0}}
\frac{y_k}{y_i - y_k} \biggl( (\alpha_{1} + \alpha_{2}y_i)
 \frac{\partial F}{\partial y_i}
- (\alpha_{1} + \alpha_{2}y_k) \frac{\partial F}{\partial y_k} \biggr)
 = 0
\end{multline} has the same unique formal power series solution $F(y)$ subject to the conditions

  (a) $F$ is a symmetric function of $\{y_i\}_{i=1}^{p_{0}}$, and

 (b) $F$ has a formal power series expansion at  $y=(0,0,\ldots)$ such that $F(0,0,\ldots)=1$.
 \end{theorem}

\begin{coro} Consider the system \eqref{generalholonomicsystem}  when the parameters  $(\alpha_{1},  \alpha_{2}, \beta_0, \beta_1, \gamma_0)$ are  evaluated at $$(1,0,c,-1,-a), \, (0,-1,1,-a-b-1,-ab), \, (1,0,c,0,-1),\, (0,0,1,-1,-a),\text{ and }(0,0,1,0,-1).$$
Then, the unique solutions of the system satisfying conditions (a)-(b) are respectively given by
$${}_1F_{1}^{(\alpha)}(a; c; y),\quad {}_2F_{0}^{(\alpha)}(a,b; y),\quad {}_0F_{1}^{(\alpha)}(c; y),\quad {}_1F_{0}^{(\alpha)}(a;y),\quad\text{ and }\quad {}_0F_{0}^{(\alpha)}(y) .$$In particular, $${}_1F_{0}^{(\alpha)}(a;y)=\prod_{i=1}^{p_{0}}(1-y_{i})^{-a}\quad\text{ and }\quad {}_0F_{0}^{(\alpha)}(y)=\prod_{i=1}^{p_{0}}e^{y_{i}}.$$
 \end{coro}
\begin{proof} One simply specializes the values of $\alpha_{1},  \alpha_{2}, \beta_0, \beta_1, \gamma_0$ and uses
the confluence relations, such as
 \[\lim_{b\rightarrow \infty}{}_2F_{1}^{(\alpha)}(a,b; c; \frac{1}{b}y)={}_1F_{1}^{(\alpha)}(a; c; y)\]
and
\[\lim_{a\rightarrow \infty}{}_1F_{1}^{(\alpha)}(a; c; \frac{1}{a}y)={}_0F_{1}^{(\alpha)}(c; y).\]The systems for ${}_1F_{0}^{(\alpha)}$ and ${}_1F_{0}^{(\alpha)}$ are trivial  and can be explicitly solved.
\end{proof}

\subsection{Deformed holonomic system with infinitely many variables }

Here we want to show that the super hypergeometric series ${}_2S\!F^{(\alpha)}_1(a,b;c;t,s)$ is one solution of the infinite deformed  system, which corresponds to \eqref{PDEs1} and\eqref{PDEs2} with $n=\infty$ and $m=\infty$.

We first introduce non-symmetric versions of the differential operators $E^l$ and $D^k$, which naturally appear in the study of Calogero-Sutherland \cite{bf,dh}:
\be E_{i}^{l}=z_{i}^{l}\frac{\partial}{\partial z_{i}},\ \qquad  l\in \mathbb{N}_{0}\ee
and
\be D_{i}^{k}=z_{i}^{k}\frac{\partial^{2}}{\partial z_{i}^{2}}+\frac{1}{\alpha}\sum_{j=1, j \ne i}^\infty
\frac{z_j}{z_i - z_j} \biggl( z_{i}^{k-1}
 \frac{\partial }{\partial z_i}
- z_{j}^{k-1} \frac{\partial }{\partial z_j} \biggr),  \qquad \ k\in \mathbb{N}.\ee
 Obviously, these operators do not preserve ring $\Lambda(z)$ of symmetric function in the indeterminates $z=z_1,z_2,\ldots$,
but as we will see below, their action is sufficiently ``nice'' to be exploited. Note that this new notation, the non deformed infinite system of Theorem \ref{generaluniquenesstheorem} can be written as:
\begin{equation} \mathcal{L}_i F= \gamma_0 F,\qquad \mathcal{L}_i=\alpha_1D^1_i+\alpha_2D^2_i+\beta_0E^0_i+\beta_1E^1_i,\qquad i=1,2,\ldots .
\end{equation}

We then generalize the $\mathbb{Z}_2$-grading introduced in Section \ref{section1} by collecting the two sequences of indeterminates $x=(x_1, x_2,\ldots)$ and $y=(y_1, y_2,\ldots)$ into a single sequence $w=(w_1, w_2,\ldots)$, where
\begin{equation*}
    w_i=\left\lbrace\begin{array}{ll}
        x_{(i+1)/2}, & \textrm{ odd} \ i\\
        y_{i/2}, &   \textrm{even} \ i
    \end{array}\right..
\end{equation*}
Also, we introduce a map $\rho:\lbrace 1, 2, \ldots,\rbrace\to\mathbb{F}$ by specifying the value of $\rho(i)$ according to
\begin{equation*}
    \rho_i=\rho(i)=\left\lbrace\begin{array}{ll}
        \phantom{-}1, &  \textrm{ odd} \ i\\
        -1/\alpha, &\textrm{even} \ i
    \end{array}\right..
\end{equation*}
As explained in Section \ref{superseries}, the algebra $S\!\Lambda(x,y)$ of supersymmetric functions is generated by the deformed power sum \eqref{homomorphismproject}, which can now be written as
\begin{equation*}
    p_{r,\alpha}(x,y) = \sum_{i=1}^{\infty}\frac{1}{\rho_i} w_i^{r}.
\end{equation*}
Furthermore, we define the following deformations of the operators $E^l_i$ and $D^k_{i}$:
\be S\!E_{i}^{l}=w_{i}^{l}\frac{\partial}{\partial w_{i}},\ \  l\in \mathbb{N}_{0}\ee
and
\begin{multline}
S\!D_{i}^{k}=\rho_iw_{i}^{k}\frac{\partial^{2}}{\partial w_{i}^{2}}+\frac{1}{\alpha}\sum_{j=1, j \ne i}^\infty
\frac{1}{\rho_i\rho_j}\frac{w_j}{w_i - w_j} \biggl( \rho_iw_{i}^{k-1}
 \frac{\partial }{\partial w_i}
-
\rho_jw_{j}^{k-1} \frac{\partial }{\partial w_j} \biggr)\\
 -(k-1)\left(1-\rho_i\right)w_i^{k-1}\frac{\partial}{\partial w_i}, \ \ k\in \mathbb{N}.
\end{multline}
The infinite dimensional equivalent of the operator $\mathcal{S}\!\mathcal{L}_i$ defined in \eqref{opsli}, which reads
\begin{multline}\label{opsliinfinite}
\mathcal{S}\!\mathcal{L}_i=\rho_iw_i(1-w_i)\frac{\partial^{2}}{\partial w_{i}^{2}} +c\frac{\partial}{\partial w_i}-(a+b+\rho_i)w_i\frac{\partial}{\partial w_i}\\
+ \frac{1}{\alpha}\sum_{\substack{j=1\\ j \ne i}}^{\infty}
\frac{1}{\rho_i\rho_j}\frac{w_j}{w_i - w_j} \biggl( \rho_i(1-w_i)
 \frac{\partial }{\partial w_i}
-
\rho_j(1-w_j) \frac{\partial }{\partial w_j} \biggr),
\end{multline}
can be simplified as
\be \mathcal{S}\!\mathcal{L}_i=S\!D_{i}^{1}-S\!D_{i}^{2}+cS\!E_{i}^{0}-(a+b+1)S\!E_{i}^{1}.\ee

By summing up all the operators $\mathcal{S}\!\mathcal{L}_i$ one obtains a deformed operator of Calogero-Sutherland type that preserves the algebra $S\!\Lambda=S\!\Lambda(x,y)$.    However, the image of a super symmetric function under the action of  $\mathcal{S}\!\mathcal{L}_i$ is not super symmetric, so we need a bigger ring than $S\!\Lambda$.  We first define $\Lambda[z_1, \ldots, z_n]$ as the  ring of polynomials in indeterminates $z_1,\ldots, z_n$ with coefficients belonging to $\Lambda=\Lambda(z)$, where $n \in \mathbb{N}$.  Second, we form the direct limit
\be\tilde{\Lambda}=\mathop{\bigcup}_{n\geq 1}\Lambda[z_1, \ldots, z_n].\nonumber\ee
It is clear that $E_{i}^{l}$ and $D_{i}^{k}$ map $\Lambda$ into $\tilde{\Lambda}$.
Likewise, $S\!\tilde{\Lambda}$ is the direct limit of the  rings $S\!\Lambda[w_1, \ldots, w_n]$ of polynomials in $w_1,\ldots,w_n$ with coefficients in $S\!\Lambda(w)$.  The deformed operators $S\!E_{i}^{l}$, $S\!D_{i}^{k}$ map $S\!\Lambda$ into $S\!\tilde{\Lambda}$.

Now we interpret both $\Lambda$ and $S\!\Lambda$ as vector spaces.  This means that the homomorphisms of algebras   $\varphi:\, \Lambda\rightarrow S\!\Lambda$, defined by  $\varphi(p_{r})=p_{r,\alpha}$, is now seen as a linear map.
We finally extend  $\varphi$  to the linear map $\tilde{\varphi}: \tilde{\Lambda}\rightarrow S\!\tilde{\Lambda}$ as follows:
\be\tilde{\varphi}(f z_{i_{1}}^{j_{1}}\cdots z_{i_{q}}^{j_{q}}):=\varphi(f) \frac{w_{i_{1}}^{j_{1}}}{\rho(i_{1})}\cdots \frac{w_{i_{q}}^{j_{q}}}{\rho(i_{q})},\nonumber\ee
  for all  ``coefficients''  $f$ in $\Lambda$, $i_{1}<\cdots<i_{q}$, and $j_{1},\ldots, j_{q}>0$.
We stress that $\tilde{\varphi}$ is not a ring homomorphism.

\begin{proposition}\label{EDdiagram}
For all $l\in\mathbb{N}_0, k, i \in\mathbb{N}$,  the following  diagrams are commutative:
%\begin{subequations}
\begin{equation}\label{Ediagram}
    \begin{CD}
        \Lambda @>E_{i}^l>>\tilde{\Lambda}\\
        @V\varphi VV @VV\tilde{\varphi} V\\
        S\!\Lambda @>S\!E_{i}^l>>S\!\tilde{\Lambda}
    \end{CD}
%\end{equation}
%\begin{equation}\label{Ddiagram}
\qquad\qquad\text{and}\qquad \qquad
   \begin{CD}
        \Lambda @>D^k_{i}>>\tilde{\Lambda}\\
        @V\varphi VV @VV\tilde{\varphi} V\\
        S\!\Lambda @>S\!D^k_{i}>>S\!\tilde{\Lambda}
    \end{CD}
\end{equation}
%\end{subequations}
\end{proposition}

\begin{proof}
Note that   $E^l_i$ are first order differential operators and that $\Lambda$ is generated by the power sums $p_r(z)$ with $r\in \mathbb{N}$,
it is sufficient to compute their action on the power sums. Observe that
$$\tilde{\varphi}(E_{i}^l p_{r})=\frac{r}{\rho_i} w_i^{l+r-1}=S\!E_{i}^l\varphi( p_{r}),$$
This yields the left-hand side of \eqref{Ediagram}.

We now  focus on the differential operators $D^k_i$ and $S\!D^k_i$.  It is sufficient to consider their action on the product $p_r(z)p_s(z)$ and
$p_{r,\alpha}(w)p_{s,\alpha}(w)$, respectively.
 Note that
\begin{equation*}
    S\!D^k_{i}\big( p_{r,\alpha}(w)p_{s,\alpha}(w)\big) = 2rs \frac{1}{\rho_i}w_{i}^{k+r+s-2} + p_{r,\alpha}(w)\big(S\!D^k_{i}\ p_{s,\alpha}(w)\big)+p_{s,\alpha}(w)\big(S\!D^k_{i}\ p_{r,\alpha}(w)\big).
\end{equation*}
Moreover, direct calculations lead to
\begin{multline*}
    \sum_{j=1, j \ne i}^\infty
\frac{1}{\rho_i\rho_j}\frac{w_j}{w_i - w_j} \biggl( \rho_iw_{i}^{k-1}
 \frac{\partial }{\partial w_i}
- \rho_jw_{j}^{k-1} \frac{\partial }{\partial w_j} \biggr) p_{r,\alpha}(w)
= r\sum_{j=1, j \ne i}^\infty
\frac{1}{\rho_i\rho_j}\sum_{t=0}^{k+r-3}w_i^{t}w_j^{k+r-2-t}  \qquad\qquad\qquad\qquad\\
= r\sum_{t=0}^{k+r-3}\left(\frac{w_i^{t}}{\rho_i}p_{k+r-2-t,\alpha}-\frac{w_i^{k+r-2}}{\rho_i}\right)+r(k+r-2)
(1-\frac{1}{\rho_i})\frac{w_i^{k+r-2}}{\rho_i},
\end{multline*}
and
\[\rho_iw_{i}^{k}\frac{\partial^{2}}{\partial w_{i}^{2}}p_{r,\alpha}(w)=r(r-1)w_i^{k+r-2}=r(r-1)\frac{w_i^{k+r-2}}{\rho_i}+r(r-1)(1-\frac{1}{\rho_i})w_i^{k+r-2}.\]
Finally, the last few equations and the use of
\[
    \left(1-\frac{1}{\rho_i}\right)\left(1+\frac{1}{\alpha\rho_i}\right)=0
\]
allow us to write
\begin{multline*}
    S\!D^k_{i}\big( p_{r,\alpha}(w)p_{s,\alpha}(w)\big) = 2rs \frac{1}{\rho_i}w_{i}^{k+r+s-2} +r(r-1) \frac{w_{i}^{k+r-2}}{\rho_i}p_{s,\alpha}(w)
    \\+s(s-1) \frac{w_{i}^{k+s-2}}{\rho_i}p_{r,\alpha}(w)
    +\frac{1}{\alpha}\Big\{r\sum_{t=0}^{k+r-3}\Big(\frac{w_i^{t}}{\rho_i}p_{k+r-2-t,\alpha}-\frac{w_i^{k+r-2}}{\rho_i}\Big)p_{s,\alpha}(w)
    \\+s\sum_{t=0}^{k+s-3}\Big(\frac{w_i^{t}}{\rho_i}p_{k+s-2-t,\alpha}-\frac{w_i^{k+s-2}}{\rho_i}\Big)p_{r,\alpha}(w)\Big\}.
\end{multline*}
When   all $\rho_i$ are equal to 1, we see that the action of $ D^k_{i}$ on $ p_{r}(z)p_{s}(z) $ has the same form as above.
Consequently, the commutativity of the diagram on the right-hand side of \eqref{Ediagram} holds true.
\end{proof}

As long as we are dealing with countably infinite sets of indeterminates, both linear maps $\varphi$ and $\tilde{\varphi}$ are injective and surjective.  The existence and uniqueness of the solution for the infinite deformed system then follows from Theorem \ref{hypergfsolution} and Proposition \ref{EDdiagram}.

\begin{theorem}
\label{hypergfissolution}
 ${}_2S\!F_{1}^{(\alpha)}(a,b; c; t,s)$ is the unique solution of the infinite deformed system defined by \eqref{PDEs1}-
\eqref{PDEs2} with $n,m=\infty$, subject to the following conditions:

(a) $F$  is symmetric in $t=(t_{1}, t_{2}, \ldots)$ and $s=(s_{1}, s_{2}, \ldots)$ separately and satisfies  the cancellation condition \eqref{cancellationproperty},

 (b) $F$ has a formal power series expansion at $(t,s)=(0,0)$ with $F(0,0)=1$.
\end{theorem}

\begin{remark}\label{remarkonfinitesystem}
If we let $x=(x_{1}, \ldots,x_{n}, 0, \ldots)$ or $y=(y_{1}, \ldots,y_{m}, 0, \ldots)$, the  commutative  diagrams as above still hold true.  This immediately implies that the corresponding
 series ${}_2S\!F_{1}^{(\alpha)}(a,b; c; t_1,\ldots,t_n,s_1,\ldots,s_m)$  is one solution of the finite deformed system \eqref{PDEs1}-\eqref{PDEs2}.  The uniqueness is not guarantied however.
\end{remark}

The combination of Theorem \ref{theodeformedkanekosystem} with $n,m=\infty$ and Theorem \ref{hypergfissolution} gives us one of main results of the article.
\begin{theorem}
\label{superselbergintegral} As  formal power series,
\begin{multline}
\frac{1}{S_{N}(\lambda_{1},\lambda_{2},\lambda)}\int_{[0,1]^{N}}
\prod_{i=1}^N\frac{\prod_{j=1}^\infty(1-x_it_j)}{\prod_{k=1}^\infty(1-x_is_k)^\lambda}\, D_{\lambda_{1},\lambda_{2},\lambda}(x_{1},\ldots,x_{N})d^{N}x=\nonumber\\
 {}_2S\!F_{1}^{(\lambda)}(-N, -N+1-(1/\lambda)(1+\lambda_1);-2N+2-(1/\lambda)(2+\lambda_1+\lambda_2);t,s).
 \end{multline}
%where $$\Pi_{\mu}(t)=\prod_{k=1}^{\infty}\prod_{i=1}^{N}(1-x_{i}t_{k})^{\mu}.$$
\end{theorem}

\subsection{Deformed holonomic system with finitely many variables}\label{SectHoloFinite}

As  pointed out in remark \ref{remarkonfinitesystem}, we already know that ${}_2S\!F_1(a,b;c;t,s)$ is a solution of $\rho_i \mathcal{S}\!\mathcal{L}_iF=ab F$ for all $1\leq i\leq n+m$.  It thus remains to prove that it is the only one that can be expanded as $\sum_\kappa A_\kappa S\!P^{(\alpha)}_\kappa(t,s)$ and such that it is equal to 1 when $(t,s)=(0,\ldots,0)$.

We first have to consider the binomial formula for super Jack polynomials.  Let $u$ be an indeterminate and define a $\mathbb{F}$-algebra
homomorphism $\epsilon_u:\,\Lambda\to\mathbb{F}[u]$ by setting
\begin{equation*}
    \epsilon_u(p_r) = u,\quad r\in\mathbb{N}.
\end{equation*}
Stanley \cite{stanley} (see also Section VI.10 in Macdonald \cite{macdonald}) has shown that the corresponding specialisation of Jack symmetric functions is given by
\begin{equation} \label{evaluationjack}
    \epsilon_u\big(P^{(\alpha)}_\kappa\big) = \prod_{(i,j)\in\kappa}\frac{u+\alpha a_\kappa^\prime(i,j)-l_\kappa^\prime(i,j)}{\alpha a_\kappa(i,j)+l_\kappa(i,j)+1},
\end{equation}
c.f.~\eqref{lengths} and \eqref{colengths}.  When $u=n$, the above formula gives the evaluation $P^{(\alpha)}_\kappa(1^n)$.  Analogously, for the super case, we  define the algebra homomorphism $\tilde\epsilon_u:S\!\Lambda\to\mathbb{F}[u]$ such that
\begin{equation*}
    \tilde\epsilon_u(p_{r,\alpha}) =u,\qquad r\in\mathbb{N}.
\end{equation*}
In particular, if $u$ is replaced by $n-\alpha m$ we have
\be
\tilde\epsilon_{n-\alpha m}(p_{r,\alpha}) =n-\alpha m=p_{r,\alpha}(1^{n+m})\nonumber
\ee
and hence for any $f\in S\!\Lambda$
\[\tilde\epsilon_{n-\alpha m}(f)=f(1^{n+m}).\]
In the last equation $f(1^{n+m})$ means that $f(x,y)$ is evaluated at $x_{1}=\cdots=x_{n}=1, x_{n+1}=\cdots=0$ and $y_{1}=\cdots=y_{m}=1, y_{m+1}=\cdots=0$.

We recall that $S\!P^{(\alpha)}_\kappa=\varphi(P^{(\alpha)}_\kappa)$  where  $\varphi: \Lambda\to S\Lambda$ is the algebra isomorphism  defined in \eqref{homomorphismproject}.   Hence,
\begin{equation}\label{evaluationsuperjack}
    \tilde\epsilon_u\big(S\!P^{(\alpha)}_\kappa\big) = \prod_{(i,j)\in\kappa}\frac{u+\alpha a_\kappa^\prime(i,j)-l_\kappa^\prime(i,j)}{\alpha a_\kappa(i,j)+l_\kappa(i,j)+1}.
\end{equation}
In particular, for $n,m\in \mathbb{N}_{0}$,
\begin{equation}
    \tilde\epsilon_{n-\alpha m}\big(S\!P^{(\alpha)}_\kappa\big) =S\!P^{(\alpha)}_\kappa(1^{n+m})
    =\prod_{(i,j)\in\kappa}\frac{n-\alpha m+\alpha (j-1)-(i-1)}{\alpha a_\kappa(i,j)+l_\kappa(i,j)+1}.
\end{equation}
We here emphasize that $S\!P^{(\alpha)}_\kappa(1^{n+m})$ might be zero, which is different from the Jack polynomials. For instance, if $n=m$ and $\alpha=1$, then
$$S\!P^{(1)}_\kappa(1^{n+n})=0\  \ \text{for\ any}\quad|\kappa|>0.$$

In the finite variable case, Lassalle \cite{lassalle} defined generalized binomial coefficients $\binom{\kappa}{\sigma}$ by the following series expansion in the algebra of symmetric polynomials:
\begin{equation}\label{binomformula}
    \frac{P^{(\alpha)}_\kappa(x_1+1,\ldots,x_n+1)}{P^{(\alpha)}_\kappa(1^n)} = \sum_{\sigma\subseteq\kappa}
    \binom{\kappa}{\sigma}\frac{P^{(\alpha)}_\sigma(x_1,\ldots,x_n)}{P^{(\alpha)}_\sigma(1^n)}.
\end{equation}
This was recently lifted to the algebra of symmetric functions (i.e., infinitely many variables) \cite{dh}.  The  idea was to introduce a new indeterminate, called $p_{0}$,  and then  define the  homomorphism
 $\theta_\gamma:\Lambda_{\mathbb{F}(p_{0})}\to\Lambda_{\mathbb{F}(p_{0})}$ for any $\gamma\in\mathbb{F}$ by setting
\begin{equation}
    \theta_\gamma(p_r) = \sum_{j=0}^r\gamma^{r-m}\binom{r}{j}p_j,\qquad r\geq 1.\nonumber
\end{equation}

\begin{proposition}[Proposition 2.1, \cite{dh}]\label{binomforfunction}
For any partition $\kappa$, we have
\begin{equation}
    \frac{\theta_1(P^{(\alpha)}_\kappa)}{\epsilon_{p_0}(P^{(\alpha)}_\kappa)} = \sum_{\sigma\subseteq\kappa}\binom{\kappa}{\sigma}\frac{P^{(\alpha)}_\sigma}{\epsilon_{p_0}(P^{(\alpha)}_\sigma)}.\nonumber
\end{equation}
\end{proposition}

In a similar manner, we define the algebra homomorphism $\tilde\theta_\gamma:S\!\Lambda_{\mathbb{F}(p_{0,\alpha})}\to S\!\Lambda_{\mathbb{F}(p_{0,\alpha})}$ such that
\begin{equation}
    \tilde\theta_\gamma(p_{r,\alpha}) = \sum_{j=0}^r\gamma^{r-m}\binom{r}{j}p_{j,\alpha},\qquad r\geq 1.\nonumber
\end{equation}
By identifying $p_{0,\alpha}$ and $\varphi(p_{0})$, we get the following commutative diagram:
 \begin{equation}
    \begin{CD}
        \Lambda_{\mathbb{F}(p_{0})} @> \theta_1>> \Lambda_{\mathbb{F}(p_{0})} \\
        @V\varphi VV @VV\varphi V\\
        S\!\Lambda_{\mathbb{F}(p_{0,\alpha})} @>\tilde\theta_1>>S\!\Lambda_{\mathbb{F}(p_{0,\alpha})}
    \end{CD}\nonumber
\end{equation}
 It immediately follows from \eqref{evaluationjack}, \eqref{evaluationsuperjack} and Proposition \ref{binomforfunction} that
\begin{equation}
    \frac{\tilde\theta_1(S\!P^{(\alpha)}_\kappa)}{\tilde\epsilon_{p_{0,\alpha}}(S\!P^{(\alpha)}_\kappa)} = \sum_{\sigma\subseteq\kappa}\binom{\kappa}{\sigma}\frac{S\!P^{(\alpha)}_\sigma}{\tilde\epsilon_{p_{0,\alpha}}(S\!P^{(\alpha)}_\sigma)}.\nonumber
\end{equation}
Thus, we have the binomial formula for super Jack polynomials:
\begin{multline}\label{binomforsuperjack}
   S\!P^{(\alpha)}_\kappa(x_1+1,\ldots,x_n+1; y_1+1,\ldots,y_m+1) = \\
      \sum_{\sigma\subseteq\kappa}\binom{\kappa}{\sigma}
   \frac{S\!P^{(\alpha)}_\kappa(1^{n+m})}{S\!P^{(\alpha)}_\sigma(1^{n+m})}
   S\!P^{(\alpha)}_\sigma(x_1,\ldots,x_n; y_1,\ldots,y_m).
\end{multline}
where
\[ \frac{S\!P^{(\alpha)}_\kappa(1^{n+m})}{S\!P^{(\alpha)}_\sigma(1^{n+m})}:=\prod_{(i,j)\in\kappa - \sigma}\frac{n-\alpha m+\alpha (j-1)-(i-1)}{\alpha a_\kappa(i,j)+l_\kappa(i,j)+1}.\]

By summing up the  $(n+m)$ differential equations in the deformed system \eqref{PDEs1} and
\eqref{PDEs2}, we get a single  differential equation that preserves the space of super symmetric polynomials:
\be \mathcal{L}_{n,m}F=0.  \nonumber\ee
In the last equation, \begin{multline}
\mathcal{L}_{n,m}:=\sum_{i=1}^{n+m}\mathcal{S}\!\mathcal{L}_i-(n-\alpha m)ab
=D_{n,m}^{1}-D_{n,m}^{2}  +\left(c-\frac{1}{\alpha}(n-\alpha m - 1)\right)
E_{n,m}^{0}\\
-\left(a+b+1-\frac{1}{\alpha}(n-\alpha m - 1)\right)E_{n,m}^{1} -(n-\alpha m)ab, \end{multline}
where
\be
E_{n,m}^{l} := \sum_{i=1}^n t_i^l \frac{\partial}{\partial t_i}+\sum_{j=1}^m s_j^l \frac{\partial}{\partial s_j}\ee
and
\begin{multline}
D_{n,m}^{k} := \sum_{i=1}^n t_i^k \frac{\partial^2}{\partial t_i^2} -\frac{1}{\alpha}\sum_{j=1}^m s_j^k \frac{\partial^2}{\partial s_j^2}
 +
\frac{2}{\alpha} \sum_{1 \leq i \neq j \leq n }
\frac{t_i^k}{t_i-t_j} \frac{\partial}{\partial t_i}-2\sum_{1 \leq i \neq j \leq m }
\frac{s_i^k}{s_i-s_j} \frac{\partial}{\partial s_i}\\
-2\sum_{\substack{i=1}}^n\sum_{\substack{j=1}}^m
\frac{1}{t_i - s_j} \biggl(t_i^{k}
 \frac{\partial F}{\partial t_i}
+ \frac{1}{\alpha}s_j^{k} \frac{\partial F}{\partial s_j} \biggr)-k (1+\frac{1}{\alpha})\sum_{\substack{j=1}}^m s_j^{k-1} \frac{\partial F}{\partial s_j}.
\end{multline}

The action of these operators on super Jack polynomials can be calculated explicitly.  We use the notation $e_i$, $i\in\mathbb{N}$, for the sequence defined by $(e_i)_j = \delta_{ij}$, where $\delta_{ij}$ is the Kronecker delta. In addition,  if $\kappa$ is a partition,
\begin{equation*}
    \kappa^{(i)} = \kappa + e_i,\qquad \kappa_{(i)} = \kappa - e_i.
\end{equation*}
We also use the standard convention that  terms in summations containing  $\kappa_{(i)}$ or $\kappa^{(i)}$ are considered only if they are  partitions.

\begin{proposition}We have
\begin{subequations}
\begin{align}
\label{E0Action}
    E_{n,m}^{0}S\!P^{(\alpha)}_\kappa &=
    \sum_i \binom{\kappa}{\kappa_{(i)}}\frac{S\!P^{(\alpha)}_\kappa(1^{n+m})}{S\!P^{(\alpha)}_{\kappa_{(i)}}(1^{n+m})}S\!P^{(\alpha)}_{\kappa_{(i)}},\\
\label{E1Action}
   E_{n,m}^{1}S\!P^{(\alpha)}_\kappa &=
     |\kappa|SP^{(\alpha)}_\kappa,\\
     \label{D1Action}
   D_{n,m}^{1}S\!P^{(\alpha)}_\kappa &=
    \sum_i \binom{\kappa}{\kappa_{(i)}}\left(\kappa_i - 1 + \frac{1}{\alpha}(n-\alpha m-i)\right)
    \frac{S\!P^{(\alpha)}_\kappa(1^{n+m})}{S\!P^{(\alpha)}_{\kappa_{(i)}}(1^{n+m})}S\!P^{(\alpha)}_{\kappa_{(i)}},\\
\label{D2Action}
D_{n,m}^{2}S\!P^{(\alpha)}_\kappa&=e_{\kappa}(\alpha)S\!P^{(\alpha)}_\kappa,
\end{align}
\end{subequations}
where \be
e_{\kappa}(\alpha)=\sum_i \kappa_i\left(\kappa_i - 1 + \frac{2}{\alpha}(n-\alpha m - i)\right). \ee
\end{proposition}
\begin{proof}One either applies the homomorphism $\varphi_{n,m}$ on the formulas of Lemma 3.9 in \cite{dh} or proceeds as in the Jack polynomial case \cite{muirhead,kaneko} and uses the binomial formula for super Jack polynomials \eqref{binomforsuperjack}.\end{proof}

The following proposition readily follows from equations \eqref{E0Action}--\eqref{D2Action}.
\begin{proposition}\label{recurencerelation}
The formal series \begin{equation}\label{formalseries}
    F(t_1,\ldots,t_n, s_1,\ldots,s_m) = \sum_{\kappa} A_{\kappa}\frac{S\!P_{\kappa}^{(\alpha)}(t;s)}{h_{\kappa}^{(\alpha)}}
\end{equation}
satisfies the equation $\mathcal{L}_{n,m}F=0$ if and only if the coefficients satisfy the recurrence relations
\begin{multline}\label{recurrencerelationforsum}
    \sum_{\kappa^{(i)}\in H_{n,m}} \binom{\kappa^{(i)}}{\kappa}\left(c+\kappa_{i}-\frac{i-1}{\alpha}\right)\big(n- \alpha m +\alpha \kappa_{i}-i+1\big)
    A_{\kappa^{(i)}} \\
    = \left(e_{\kappa}(\alpha)+\big(a+b+1-\frac{n- \alpha m-1}{\alpha}\big)|\kappa|+(n- \alpha m)ab\right)A_{\kappa}.
\end{multline}
\end{proposition}

Let $p\leq n$ and $q\leq m$.  Moreover, let $F_{p,q}$ denote the function obtained  from the function $F$ above by reducing the number of variables:
\begin{align}\label{variate-reduced}
   F_{p,q}(t_1,\ldots,t_p, s_1,\ldots,s_q):&= F(t_1,\ldots,t_p,0, \ldots,0, s_1,\ldots,s_q, 0, \ldots,0)\\
   & = \sum_{\kappa\in H_{p,q}} A_{\kappa}\frac{SP_{\kappa}^{(\alpha)}(t_1,\ldots,t_p; s_1,\ldots,s_q)}{h_{\kappa}^{(\alpha)}}\nonumber.
\end{align}
To ensure the uniqueness of the solution of  $\mathcal{L}_{n,m}F=0$  from the recurrence relations \eqref{recurrencerelationforsum}, we make use of the following
 criterion, which was first introduced by Kaneko in the non-deformed case (see Corollary 4.6 in \cite{kaneko2}).

\begin{definition}[Kaneko's criterion] The formal series \eqref{formalseries} is said to satisfy Kaneko's criterion  if
$\mathcal{L}_{p, 0}F_{p, 0}=0$ for every $p\leq n$, and $\mathcal{L}_{n, q}F_{n, q}=0$ for every  $q\leq m$.
\end{definition}

\begin{lemma}\label{zerosolutionforsum}
Assume that $[c]_{\kappa}^{(\alpha)}\neq 0$ for any partition $\kappa \in H_{n,m}$. If  the formal series \eqref{formalseries} satisfies the Kaneko's criterion and
 $F(0,\ldots,0)=0$, then $$F(t_1,\ldots,t_n, s_1,\ldots,s_m)\equiv 0.$$
\end{lemma}

\begin{proof}
Let $ H_{n,m}=\bar{H}_{n,m}\bigcup \tilde{H}_{n,m}$, where $\bar{H}_{n,m}=\{\kappa \in H_{n,m}|\kappa_{n+1}=0\}$ and
$\tilde{H}_{n,m}=\{\kappa \in H_{n,m}|\kappa_{n+1}>0\}$. Without loss of generality, we assume $n\geq 1$. We will determine  $A_{\kappa}$ by induction from $A_{(0)}=0$.
Note that the coefficient of $A_{\kappa^{(i)}}$ on the left-hand side of \eqref{recurrencerelationforsum} is not zero except for the factor
$n- \alpha m +\alpha \kappa_{i}-i+1$, which is denoted by $g_{n,m}(\kappa_{i})$.

We first claim that $A_{\kappa}=0$ for every $\kappa \in \bar{H}_{n,m}$. It is sufficient to prove that $F_{p,0}=0$ by induction on $p$. The case $p=1$ immediately follows from
\eqref{recurrencerelationforsum} because  $g_{1,0}(\kappa_{1})=1+\alpha \kappa_{1}>0$. Assume $F_{p,0}=0$ for $p\leq n-1$, i.e. $A_{\sigma}=0$ if $\ell(\sigma)\leq n-1$.
For $A_{\kappa}$ with $\kappa_{n}=1$, put $\sigma=\kappa_{(n)}$ or $\sigma^{(n)}=\kappa$. Substituting this $\sigma$ into \eqref{recurrencerelationforsum}
 shows that $A_{\sigma^{(n)}}$ is a linear combination of $A_{\sigma}$
 and $A_{\sigma^{(i)}}, i<n$  because of $g_{n,0}(\sigma_{n})=1$, which implies $A_{\kappa}=A_{\sigma^{(n)}}=0$ with $\kappa_{n}=1$.
The case where $\kappa_{n}>1$ can be checked by induction on $\kappa_{n}$ since  $g_{n,0}(\kappa_{n})=1+\alpha \kappa_{n}>0$.

Next, we deal with the $A_{\kappa}$ for  $\kappa \in \tilde{H}_{n,m}$. Every  $\kappa \in H_{n,m}$ can be rewritten uniquely in the form
$\kappa=(\kappa_{1},\ldots,\kappa_{n})\cup \pi'$ where $\pi'=(\kappa_{n+1},\ldots, \kappa_{n+i},\ldots)$ with $\kappa_{n+i}\leq m$, i.e. $\pi=(\pi_{1},\ldots,\pi_{m})$.
We prove that $F_{n,q}=0$ by induction on $q$. If $q=0$ nothing is done. When $q=1$, for  $\kappa \in \tilde{H}_{n,m}$ with
$\kappa=(\kappa_{1},\ldots,\kappa_{n},1,0)$, i.e.  $\pi=(1)$, put $\sigma=\kappa_{(n+1)}$ or $\sigma^{(n+1)}=\kappa$. Substituting this $\sigma$ into \eqref{recurrencerelationforsum}
 shows that $A_{\sigma^{(n+1)}}$ is a linear combination of $A_{\sigma}$ with $\sigma \in \bar{H}_{n,m}$ since  $g_{n,1}(\sigma_{n+1})=-\alpha<0$,
 which implies $A_{\kappa}=A_{\sigma^{(n+1)}}=0$. The case where $\ell(\kappa)>n+1$,  i.e. $\pi_{1}>1$,  can be checked by induction on $\pi_{1}$.
 Assume $F_{n,q}=0$ for $q\leq m-1$, i.e. $A_{\sigma}=0$ if $ \sigma\in H_{n,m-1}$.
For $A_{\kappa}$ with $\pi_{m}=1$, put $\sigma=\kappa_{(n+1)}$ or $\sigma^{(n+1)}=\kappa$. Substituting this $\sigma$ into \eqref{recurrencerelationforsum}
 and noting that $g_{n,m}(\sigma_{n+1})=-\alpha$, one obtains   $A_{\sigma^{(n+1)}}=0$. The general case follows  by induction on
  $\pi_{m}$ since  $g_{n,m}(\kappa_{i})\neq 0$ for any $i> n+1$.
 \end{proof}

\begin{lemma}
\label{heridity} If $F(t_1,\ldots,t_n, s_1,\ldots,s_m)$ is one solution of the finite $(n+m)$-system, then for all $p\leq n$ and $q\leq m$,
$F_{p,q}(t_1,\ldots,t_p, s_1,\ldots,s_q)$ is one solution of the finite $(p+q)$-system.
\end{lemma}
\begin{proof}
Proceed by induction on $p$ and $q$. For instance, if $p=n-1$ and $q=m$,
consider the first $(n-1)$- equations of  \eqref{PDEs1} in the finite $(n+m)$-system. Suppose $t_{i}\neq 0, s_{j}\neq 0$, $i=1,\ldots,n-1, j=1,\ldots,m$,
let $t_{n} \to 0$, then we are done.
\end{proof}

Lemma  \ref{heridity} implies that if $F(t_1,\ldots,t_n, s_1,\ldots,s_m)$ is one solution of the finite $(n+m)$-system, then it satisfies the\textit{ Kaneko's
criterion}.  Returning to Lemma \ref{zerosolutionforsum},  we deduce the following.

\begin{lemma}\label{zerosolution}
Assume that $[c]_{\kappa}^{(\alpha)}\neq 0$ for any partition $\kappa \in H_{n,m}$. If
$F(t_1,\ldots,t_n, s_1,\ldots,s_m)$ is one solution of the finite $(n+m)$-system,  and
 $F(0,\ldots,0)=0$, then $$F(t_1,\ldots,t_n, s_1,\ldots,s_m)\equiv 0.$$
\end{lemma}

\begin{theorem}
\label{Thmforsum}
 ${}_2SF_{1}^{(\alpha)}(a,b; c; t_1,\ldots,t_n, s_1,\ldots,s_m)$ is the unique solution of the summed-up equation $\mathcal{L}_{n,m}F=0$
  subject to the following conditions:

(a) $F$  is symmetric in $t=(t_1,\ldots,t_n)$ and $s=(s_1,\ldots,s_m)$ separately and satisfies  the cancellation condition \eqref{cancellationproperty},

 (b) $F$ has a formal power series expansion at $(t,s)=(0,0)$ with $F(0,0)=1$,

 (c) $F$ satisfies the Kaneko's criterion: $\mathcal{L}_{p, 0}F_{p, 0}=0$ for every $p\leq n$ and $\mathcal{L}_{n, q}F_{n, q}=0$ for every  $q\leq m$, where
 $F_{p,q}$ is given by \eqref{variate-reduced}.
 \end{theorem}
\begin{proof}
It follows from Theorem \ref{hypergfissolution} or Remark \ref{remarkonfinitesystem} that ${}_2SF_{1}^{(\alpha)}(a,b; c; t_1,\ldots,t_n, s_1,\ldots,s_m)$ is one solution
 of  the finite $(n+m)$-system, certainly implying that it is one solution of the summed-up equation $\mathcal{L}_{n,m}F=0$. Uniqueness is an immediate consequence of
 Lemma \ref{zerosolution}.\end{proof}

We recall that every super symmetric solution of the deformed $(n+m)$-system is also a solution of the equation $\mathcal{L}_{n,m}F=0$.  Comparing Theorems \ref{hypergfissolution} and  \ref{Thmforsum},  we then arrive at the desired statement for the existence and uniqueness of the solution to the deformed system.
\begin{theorem}
\label{Thmforfinitesystem}
 ${}_2S\!F_{1}^{(\alpha)}(a,b; c; t_1,\ldots,t_n, s_1,\ldots,s_m)$ is the unique solution of  the finite deformed $(n+m)$-system
  subject to the following conditions:

(a) $F$  is symmetric in $t=(t_1,\ldots,t_n)$ and $s=(s_1,\ldots,s_m)$ separately and satisfies  the cancellation condition \eqref{cancellationproperty},

 (b) $F$ has a formal power series expansion at $(t,s)=(0,\ldots,0)$ with $F(0,\ldots,0)=1$.
 \end{theorem}

\begin{remark} As will be stated in the next section, the series ${}_2S\!F_1(a,b;c;t,s)$, as a function of $(t,s)$, converges in some  neighborhood of the origin in $\mathbb{C}^{n+m}$.  Thus, Theorem \ref{maintheorem} directly follows from Theorem \ref{Thmforfinitesystem}.  Similarly, Corollary \ref{maincoro} is a consequence Theorems \ref{superselbergintegral} and  \ref{Thmforfinitesystem} and the convergence of the super hypergeometric series in finitely many variables. \end{remark}

Similarly, for  the semi-infinite systems, i.e.
 $t=(t_{1}, \ldots,t_{n}, 0, \ldots), s=(s_{1}, s_{2}, \ldots)$ or $s=(t_{1}, t_{2}, \ldots), s=(s_{1}, \ldots,s_{m}, 0, \ldots)$, the unique solution can also be given.  Let us consider  $t=(t_{1}, \ldots,t_{n}, 0, \ldots)$ and $s=(s_{1}, s_{2}, \ldots)$.

\begin{theorem}
\label{Thmforone-side}
 ${}_2SF_{1}^{(\alpha)}(a,b; c; t,s)$ is the unique solution of the deformed semi-infinite system given by \eqref{PDEs1} and
\eqref{PDEs2} with $n<\infty$ and $m=\infty$, subject to the following conditions:

(a) $F$  is symmetric in $t=(t_{1}, \ldots,t_{n})$ and $s=(s_{1}, s_{2}, \ldots)$ separately and satisfies  the cancellation condition \eqref{cancellationproperty},

 (b) $F$ has a formal power series expansion at $(t,s)=(0,\ldots, 0)$ with $F(0,\ldots,0)=1$.
\end{theorem}

\begin{proof}Theorem \ref{hypergfissolution} says ${}_2SF_{1}^{(\alpha)}(a,b; c; t,s)$ is a solution of the deformed semi-infinite system.
Every solution  can be rearranged as a series of super Jack symmetric functions
  $$F(t,s) = \sum_{\kappa} \gamma_{\kappa}\frac{SP_{\kappa}^{(\alpha)}(t,s)}{h_{\kappa}^{(\alpha)}}.$$
It suffices to show that the coefficient $\gamma_\kappa$ must be equal to
$$A_\kappa=\frac{[a]_{\kappa}^{(\alpha)}  [b]_{\kappa}^{(\alpha)}}{[c]_{\kappa}^{(\alpha)}}$$
For any partition $\sigma$ and given $m>\max\{1,\sigma_{n+1}\}$, putting $s_{j}=0$ for $j>m$,
$F( t_1,\ldots,t_n, s_1,\ldots,s_m)=\sum_{k=0}^{\infty}\sum_{|\kappa|=k}\gamma_{\kappa}C_{\kappa}(y_{1},\ldots,y_{n},0,\ldots)$ clearly satisfies the finite
$(n+m)$-system. Since $\sigma\in H_{n,m}$, Theorem \ref{Thmforfinitesystem} implies that
$\gamma_\sigma=A_\sigma$, as expected.
\end{proof}

Let us end this section by a direct generalization of Theorem \ref{Thmforfinitesystem}, which comes from Theorem \ref{generaluniquenesstheorem} and a simple modification of Theorem \ref{Thmforsum}.

\begin{theorem}
\label{generaltheoremfords} Fix $n,m  \in\mathbb{N}_{0} $  and let $\alpha_{1},  \alpha_{2}, \beta_0, \beta_1, \gamma_0$ be complex numbers such that
 $\beta_{0}-(\frac{i-1}{\alpha}-\kappa_{i})\alpha_{1}$ are never zero for any $\kappa \in H_{n,m}$.
 Then the deformed system
\begin{multline}\label{generaldeformedsystem1}
t_i (\alpha_{1} + \alpha_{2}t_i) \frac{\partial^2 F}{\partial t_i^2} + (\beta_{0} + \beta_{1}t_i)
\frac{\partial F}{\partial t_i} +\gamma_{0}F   \\
+\frac{1}{\alpha}\sum_{k=1, k \ne i}^{n}
\frac{t_k}{t_i - t_k}
\biggl( (\alpha_{1} + \alpha_{2}t_i)
 \frac{\partial F}{\partial t_i}
- (\alpha_{1} + \alpha_{2}t_k) \frac{\partial F}{\partial t_k} \biggr)\\
-\sum_{k=1}^{m}
\frac{s_k}{t_i - s_k} \biggl((\alpha_{1} + \alpha_{2}t_i)
 \frac{\partial F}{\partial t_i}
+ \frac{1}{\alpha}(\alpha_{1} + \alpha_{2}s_k) \frac{\partial F}{\partial s_k} \biggr)
 = 0, \qquad 1\leq i \leq n,
\end{multline}
\begin{multline}\label{generaldeformedsystem2}
-\frac{1}{\alpha}s_j (\alpha_{1} + \alpha_{2}s_j) \frac{\partial^2 F}{\partial s_j^2} + \Big(\beta_{0} + \big(\beta_{1}-(1+\frac{1}{\alpha})\alpha_{2}\big)s_j\Big)
\frac{\partial F}{\partial s_j} +(-\alpha)\gamma_{0}F   \\
-\sum_{k=1, k \ne j}^{m}
\frac{s_k}{s_j - s_k}
\biggl( (\alpha_{1} + \alpha_{2}s_j)
 \frac{\partial F}{\partial s_j}
- (\alpha_{1} + \alpha_{2}s_k) \frac{\partial F}{\partial s_k} \biggr)\\
+\sum_{k=1}^{n}
\frac{t_k}{s_j - t_k} \biggl(\frac{1}{\alpha}(\alpha_{1} + \alpha_{2}s_j)
 \frac{\partial F}{\partial s_j}
+ (\alpha_{1} + \alpha_{2}t_k) \frac{\partial F}{\partial t_k} \biggr)
 = 0, \qquad 1\leq j \leq m,
\end{multline}
has a  unique  solution $F(t,s)$  subject to the following conditions:

(a) $F$  is symmetric in $t=(t_1,\ldots,t_n)$ and $s=(s_1,\ldots,s_m)$ separately and satisfies  the cancellation condition \eqref{cancellationproperty},

 (b) $F$ has a formal power series expansion at $(t,s)=(0,\ldots,0)$ with $F(0,\ldots,0)=1$.

 \end{theorem}

\section{Further properties of super hypergeometric series}
\label{sectproperties}
In this section, we establish some important properties of super hypergeometric series ${}_pS\!F_q$ with $p\leq 1$ and $q\leq 2$.

\subsection{Special cases}\label{sectspecialcases}

Let us return to Theorem \ref{generaltheoremfords} and consider special values for the parameters $\alpha_{1}$, $\alpha_{2}$, $\beta_0$, $\beta_1$, and $\gamma_0$.  We are concerned with the solutions of the deformed system   \eqref{generaldeformedsystem1} and \eqref{generaldeformedsystem2} of the form $$F(t,s)=\sum_\kappa c_\kappa S\!P_\kappa(t;s)\qquad \text{with}\qquad F(0,\ldots,0)=1.$$   We immediately see that ${}_1S\!F_{1}^{(\alpha)}(a; c; t,s)$ is the unique solution for  the deformed system when
$$(\alpha_{1},  \alpha_{2}, \beta_0, \beta_1, \gamma_0)=(1,0,c,-1,-a).$$
 Similarly, ${}_2SF_{0}^{(\alpha)}(a, b; t,s)$ and $ {}_0SF_{1}^{(\alpha)}(c; t,s)$ are the unique solutions for the respective cases
$$ (\alpha_{1},  \alpha_{2}, \beta_0, \beta_1, \gamma_0)=(0,-1,1,-a-b-1,-ab) \qquad \text{and}\qquad (\alpha_{1},  \alpha_{2}, \beta_0, \beta_1, \gamma_0)=(1,0,c,0,-1).$$

The other simple deformed systems  can be solved explicitly.  Indeed, one easily finds that the unique solution for the system with
$$(\alpha_{1},  \alpha_{2}, \beta_0, \beta_1, \gamma_0)=(0,0,1,-1,-a)\qquad \text{and}\qquad (\alpha_{1},  \alpha_{2}, \beta_0, \beta_1, \gamma_0)=(0,0,1,0,-1)$$ are respectively given by
$${}_1SF_{0}^{(\alpha)}(a;t,s)=\prod_{i=1}^{n}(1-t_{i})^{-a}\prod_{j=1}^{m}(1-s_{j})^{\alpha a},$$
and $${}_0SF_{0}^{(\alpha)}(t,s)=\prod_{i=1}^{n}e^{t_{i}}\prod_{j=1}^{m}e^{-\alpha s_{j}}.$$

\subsection{Transformations}

As pointed out in \cite{dh}, properties of the super Jack polynomials and generalized coefficients imply that the super series ${}_pS\!F_q$ enjoy a duality symmetry that has no equivalent in the usual case. This duality can also be understood as follows.  By interchanging \eqref{PDEs1} and \eqref{PDEs2},
we   can rewrite the deformed system  in the same form but with $\alpha$ replaced by  $1/\alpha$, and $a, b, c$ replaced by $-\alpha a, -\alpha b, -\alpha c$ respectively. The uniqueness then implies the following duality relation.

\begin{proposition}\label{dualityseries}We have  \begin{align*}
  {}_2S\!F_{1}^{(\alpha)}(a,b; c; t,s)&= {}_2S\!F_{1}^{(1/\alpha)}(-\alpha a, -\alpha b; -\alpha c; s,t),&
  {}_1S\!F_{0}^{(\alpha)}(a; t,s)&= {}_1S\!F_{0}^{(1/\alpha)}(-\alpha a; s,t),\\
  {}_1S\!F_{1}^{(\alpha)}(a; c; t,s)&= {}_1S\!F_{1}^{(1/\alpha)}(-\alpha a; -\alpha c; -\alpha s, -\alpha t),&   {}_0S\!F_{0}^{(\alpha)}(t,s)&= {}_0S\!F_{0}^{(1/\alpha)}( -\alpha s, -\alpha t),\\
  {}_0S\!F_{1}^{(\alpha)}(c; t,s)&= {}_0S\!F_{1}^{(1/\alpha)}(-\alpha c; \alpha^{2} s, \alpha^{2} t).
   \end{align*}\end{proposition}

The next properties generalize some results obtained by Yan \cite{yan} in the usual Jack polynomial case.   For convenience, we use the following  convention  $b+a s=(b+a s_{1}, b+a s_{2}, \ldots)$ and $ \frac{s}{1-s}=
( \frac{s_{1}}{1-s_{1}}, \frac{s_{2}}{1-s_{2}}, \ldots)$ for the complex number $a, b$ and the sequence $s=(s_{1}, s_{2}, \ldots)$.

\begin{proposition}
 We have the generalized Pfaff-Euler relations \begin{align}\label{Eulerrelation}
  &{}_2SF_{1}^{(\alpha)}(a,b; c; t,s)\nonumber\\
  &= \prod_{i=1}^{\iy}(1-t_{i})^{-a}\prod_{j=1}^{\iy}(1-s_{j})^{\alpha a}{}_2SF_{1}^{(\alpha)}(a,c-b; c; -\frac{t}{1-t},-\frac{s}{1-s})\nonumber\\
  &=\prod_{i=1}^{\iy}(1-t_{i})^{-b}\prod_{j=1}^{\iy}(1-s_{j})^{\alpha b}{}_2SF_{1}^{(\alpha)}(c-a,b; c; -\frac{t}{1-t},-\frac{s}{1-s})\nonumber\\
  &=\prod_{i=1}^{\iy}(1-t_{i})^{c-a-b}\prod_{j=1}^{\iy}(1-s_{j})^{-\alpha (c-a-b)}{}_2SF_{1}^{(\alpha)}(c-a,c-b; c; t,s),
   \end{align}
   and the generalized Kummer transformations
      \begin{align}\label{Kummerrelation}
  &{}_1SF_{1}^{(\alpha)}(a; c; t,s)= \prod_{i=1}^{\iy}e^{t_{i}}\prod_{j=1}^{\iy}e^{-\alpha s_{j}}{}_1SF_{1}^{(\alpha)}(c-a; c; -t,-s).
   \end{align}
\end{proposition}
\begin{proof}
Check directly that both sides of \eqref{Eulerrelation} satisfy the same deformed system and then make use of the uniqueness of Theorem \ref{hypergfissolution}. \eqref{Kummerrelation} follows from the third equality of \eqref{Eulerrelation} by setting
$t=t/b, s=s/b$ and letting $b\to \iy$.
\end{proof}

\subsection{Convergence and integral representation}

Up to now, the series were considered formally.  The next proposition shows that the  series in $S\!P_\kappa(t;s)$ can be also treated as complex-valued functions on  $ \mathbb{C}^{n+m}$.  The convergence of the super hypergeometric series actually follows the classical pattern.    The proof will be given in Appendix \ref{appendixSeries}. Here we emphasize that the part (3) of Proposition \ref{propconvergence} below  is slightly different from the usual Jack polynomial case (see  part (3) of Proposition 1 \cite{kaneko}, which claims that the series diverges at every nonzero point unless it terminates). Let us illustrate this point as follows: setting $\alpha=1, n=m=1$, we have $S\! P^{(1)}_{\kappa}(t_{1},s_{1})=0$ on the line $t_{1}=s_{1}$ for any $\kappa\neq (0)$, thus ${}_pS\!F^{(1)}_q (t_{1},s_{1})=1$ on the line $t_{1}=s_{1}$, which is not divergent even if $p>q+1$.

\begin{proposition} \label{propconvergence}
Let ${}_pS\!F^{(\alpha)}_q$ be the series in \eqref{hfpq} with $(t,s)\in \mathbb{C}^{n+m}$ and $\alpha>0$.  Let also $r_{1}=\max\{\alpha,\alpha^{-1}\}$ and $\|(t,s)\|=\max\{|t_1|,\ldots,|t_n|, |s_1|,\ldots,|s_m|\}$.

(1) If $p\leq q$, then the series ${}_pS\!F^{(\alpha)}_q$ converges absolutely for all $(t,s)\in \mathbb{C}^{n+m}$.

(2) If $p= q+1$, then ${}_pS\!F^{(\alpha)}_q$ converges absolutely for $\|(t,s)\|<1/\big(r^{2}_{1}(n+r_{1} m)\big)$.

(3) If $p>q+1$, then there does not exist a positive constant $\rho$ such that ${}_pS\!F^{(\alpha)}_q$  converges absolutely for $\|(t,s)\|<\rho$,   unless it terminates.
\end{proposition}

%\subsection{Series ${}_{2}S\!F^{(\alpha)}_{1}$ and ${}_{1}S\!F^{(\alpha)}_{0}$ }

As mentioned earlier, the series ${}_{p+1}S\!F^{(\alpha)}_{p}$ with $p=0,1$ are particularly important for us.  Proposition \ref{propconvergence} guarantees the convergence of such series in some neighborhood of the origin in $\mathbb{C}^{n+m}$.  In some instances however, the domain of convergence can be extended.

As an example, we first look at the binomial expansion of the deformed Cauchy product $\prod_{i,j}(1-x_iw_j)^{-\rho_j/\alpha}$.  We know from Proposition 3.1 in \cite{hallnas} that
\be \label{eqscauchy}\prod_{i=1}^\ell \frac{\prod_{j=1}^m(1-x_is_j)^{\phantom{1/\alpha}}}{\prod_{k=1}^n(1-x_it_k)^{1/\alpha}}=\sum_\kappa b_\kappa^{(\alpha)}P^{(\alpha)}_\kappa (x)\,S\! P^{(\alpha)}_{\kappa}(t;s),
\ee
where
\be \label{defb}b_\kappa^{(\alpha)}=\frac{1}{\alpha^{|\kappa|}}\frac{h^{(\alpha')}_{\kappa'}}{h^{(\alpha)}_{\kappa}}=\frac{1}{b_{\kappa'}^{(\alpha')}} \qquad\text{and}\qquad\alpha'=\frac{1}{\alpha} .\ee
 By using Stanley's evaluation formula for Jack polynomials (see eq. (10.20) of Chapter VI in \cite{macdonald}), which can be written as
 \begin{equation}\label{evaluationat1}
 P_{\kappa}^{(\alpha)}(1^{\ell})=
 \frac{\alpha^{|\kappa|} \lbrack \ell/\alpha\rbrack^{(\alpha)}_\kappa}{h_{\kappa'}^{(\alpha')}}
 ,\end{equation}
we arrive at the conclusion that
\be \label{eqscauchy2}\prod_{i=1}^\ell \frac{\prod_{j=1}^m(1-x_is_j)^{\phantom{1/\alpha}}}{\prod_{k=1}^n(1-x_it_k)^{1/\alpha}}={}_1\mathcal{S}\!\mathcal{F}^{(\alpha)}_0(\ell/\alpha;x;t,s)
\ee
 Proceeding like in Appendix \ref{appendixSeries}, one easily checks that the right-hand side of \eqref{eqscauchy} converges for all $(x,s,t)\in\mathbb{C}^{\ell+m+n}$ such that  $|x_is_j|<1$ and $|x_it_k|<1$, and so does ${}_1\mathcal{S}\!\mathcal{F}^{(\alpha)}_0$ in the same domain.  As a consequence we also have
\be \label{eqscauchy3}\frac{\prod_{j=1}^m(1-s_j)^{\ell\phantom{/\alpha}}}{\prod_{k=1}^n(1-t_k)^{\ell/\alpha}}={}_1{S}\!F^{(\alpha)}_0(\ell/\alpha;t,s)
\ee
which now converges in a larger domain than that given in Proposition \ref{propconvergence}.

We now provide another proof that the deformed Selberg integral $S _N(\lambda_1,\lambda_2,\lambda;t;s)$ defined in \eqref{skanekoint} is in fact given by a ${}_2S\!F_1$ series.
\begin{theorem} \label{othertheorem} Suppose $\lambda_1,\lambda_2>-1$, $\lambda>0$, $|s_j|<1$ and $|t_k|<1$.  Then
\[ S _N(\lambda_1,\lambda_2,\lambda;t;s)=S _N(\lambda_1,\lambda_2,\lambda)\,{}_2S\!F_1^{(\lambda)}\left(-N,-N+1-(\lambda_1+1)/\lambda,-2N+2-(\lambda_1+\lambda_2+2)/\lambda;t,s\right).
\]
\end{theorem}
\begin{proof} First set $\alpha=1/\lambda$ and let $\Theta$ denote the set of points $(x,s,t)$ in $\mathbb{C}^{N+m+n}$ such that $0\leq x_i\leq 1$, $|s_j|<1$ and $|t_k|<1$.  According to \eqref{eqscauchy2}, the series ${}_1\mathcal{S}\mathcal{F}^{(\alpha)}_0(N/\alpha;x;t,s)$ converges to $\prod_{i=1}^N \prod_{j=1}^m(1-x_is_j)\prod_{k=1}^n(1-x_it_k)^{-1/\alpha}$ everywhere in the domain $\Theta$.  The latter function is continuous and bounded on the same domain.  Application of Lebesgue's dominated convergence theorem in  Proposition \ref{propintsuperseries}  then proves that
\[S _N(\lambda_1,\lambda_2,1/\alpha;s;t)=S _N(\lambda_1,\lambda_2,\lambda){}_2S\!F_1^{(\alpha)}\left(\frac{N}{\alpha},\frac{N-1}{\alpha}+{\lambda_1+1},\frac{2N-2}{\alpha}+{\lambda_1+\lambda_2+2};t,s\right).
\]
The interchange of $s$ and $t$ together with the use of Proposition \ref{dualityseries} finally establishes the theorem.
\end{proof}

\subsection{Super Jacobi polynomials}\label{sectsuperjacobipoly}

The super Jacobi polynomials were introduced  in \cite{sv3}.  They can be defined the unique super polynomials $S\!J^{(\gamma,\eta)}_\kappa(t,s;\alpha)$ satisfying the two following conditions \cite{dh}:
\begin{align} (1)\qquad& S\!J^{(\gamma,\eta)}_\kappa=S\!P^{(\alpha)}_\kappa +\sum_{\mu\subset\kappa}c_{\kappa\mu}S\!P^{(\alpha)}_\mu \\
(2)\qquad &\left(D_{n,m}^1-D_{n,m}^2+(\gamma+1)E^0_{n,m}-(\gamma+\eta+2)E^1_{n,m} \right) S\!J^{(\gamma,\eta)}_\kappa=\epsilon^J_\kappa \,S\!J^{(\gamma,\eta)}_\kappa\end{align}
where each coefficient $c_{\kappa\mu}$ depends  upon $\alpha,\gamma,\eta$ and $p_0=n-\alpha m$.  The eigenvalue  is explicitly known:
\be \epsilon^J_\kappa=-\sum_i\kappa_i\left(\kappa_i-1+\frac{2}{\alpha}(p_0-i)\right)-(\gamma+\eta+2)|\kappa|.
\ee

Now, we have established  in the last section that ${}_2S\!F_1(a,b,c;t,s)$ is a solution of the form $\sum_{\kappa}d_{\kappa}S\!P^{(\alpha)}_\kappa(t,s)$ to the following  equation:
\be \left(D_{n,m}^1-D_{n,m}^2+\left(c-\frac{p_0-1}{\alpha}\right)E^0_{n,m}-\left(a+b+1-\frac{p_0-1}{\alpha}\right)E^1_{n,m} \right)F=abp_0\,F.\ee
Moreover, according to Lemma \ref{lemmapochhammer}, ${}_2S\!F_1(a,b,c;t,s)$ is a polynomial if
$$a=-N,\qquad \text{and}\qquad b=\frac{M}{\alpha}.$$
With this choice of parameters, the highest term of ${}_2S\!F_1(a,b,c;t,s)$ is  given by the largest possible partition in the fat hook $H_{n,m}$ whose fist part and length are respectively  equal to $N$ and $M$, which is
\be \kappa^\textrm{max}= (N^n,m^{M-n}).\ee
 Moreover, one easily verifies  that
$\epsilon^J_{\kappa^\textrm{max}}=abp_0$ whenever $1+\gamma+\eta=a+b-(p_0-1)/\alpha$, $a=-N$, and $b=M/\alpha$. Collecting all the preceding observations, we conclude the following.

\begin{proposition} Let $\kappa^{\textrm{max}}$ as above, $\eta=-\gamma-1+m-N+(M+1-n)/\alpha$ and $c=\gamma+1+(p_0-1)/\alpha$.  Then the super Jacobi polynomial  $S\!J^{(\gamma,\eta)}_{\kappa^\textrm{max}}(t,s;\alpha)$ is proportional to ${}_2S\!F^{(\alpha)}_1(-N,M/\alpha,c;t,s)$.
\end{proposition}

In principle,  ${}_2S\!F^{(\alpha)}_1(-N,M/\alpha,c;t,s)$ cannot be represented as a deformed Selberg integral since it would correspond to the case $\lambda_1=-M-\alpha N+\alpha-1$ which is negative.  However, as we pointed out at the end of Section \ref{section1}, the multiple contour integral $T_N(\lambda_1,\lambda_2,\lambda;t,s)$ is well defined for such values.

\begin{coro} Let $\lambda\in\mathbb{Z}_+$, $\lambda_1=-M-1-\lambda (N-1)$, $\lambda_2=-\lambda (N-m+\gamma)+M-n-1>-1$ and $\eta=-\gamma-1+m-N+(M+1-n)/\lambda$.  Then, the super Jacobi polynomial  $S\!J^{(\gamma,\eta)}_{\kappa^\textrm{max}}(t,s;\lambda)$ is proportional to
$$ \int_1^{(0^+)}\cdots \int_1^{(0^+)}\prod_{j=1}^N\frac{\prod_{i=1}^{n}(1-x_jt_i)}{\prod_{k=1}^{m}(1-x_js_k)^{\lambda}}D_{\lambda_1,\lambda_2,\lambda}(x) dx_1\cdots dx_N.$$
\end{coro}

\section{Applications to $\beta$-Ensembles of random matrices}
\label{sectapplication}

In this section we apply our results on super hypergeommetric series to the classical $\beta$-Ensembles of Random Matrix Theory, namely the Jacobi, Laguerre (Chiral), Hermite (Gaussian) and Circular $\beta$-Ensembles.
Note that in following paragraphs, the indeterminates $t=(t_1,\ldots,t_n)$ and $s=(s_1,\ldots,s_m)$ are always considered as finite sets of complex variables.  Unless otherwise stated, the series considered below are convergent.

\subsection{Jacobi $\beta$-ensemble}

The joint probability density function of the Jacobi $\beta$-Ensemble (J$\beta$E) is equal to
\begin{equation}\label{PDFforjacobi}
\frac{1}{S_{N}(\lambda_{1},\lambda_{2},\beta/2)}\prod_{i=1}^{N}x_{i}^{\lambda_{1}}(1-x_{i})^{\lambda_{2}} \,\prod_{1\leq j<k\leq
N}|x_{j}-x_{k}|^{\beta}.
\end{equation}
According to Corollary \ref{maincoro} (or equivalently Theorems  \ref{superselbergintegral} or \ref{othertheorem}), the average of ratios of characteristic polynomials over this ensemble is given by a particular $_2S\!F_1$ series.
\begin{proposition} Whenever $|t_i|<1$ and $|s_j|<1$,
\begin{multline}\label{Jacobiform1}
\left\langle\prod_{k=1}^{N}\Big(\prod_{i=1}^{n}(1-t_{i}x_{k})\prod_{j=1}^{m}(1-s_{j}x_{k})^{-\beta/2}\Big)\right\rangle_{\mathrm{J\beta E}}
=\\
{}_2S\!F_{1}^{(\beta/2)}(-N, -N+1-(2/\beta)(1+\lambda_1);-2N+2-(2/\beta)(2+\lambda_1+\lambda_2);t;s).
\end{multline}
\end{proposition}

Moreover in the case where $|t_i|>1$ and $|s_j|>1$, following Lemma \ref{inverselemma}, we have
\begin{multline}\label{Jacobiform2}
\left\langle\prod_{k=1}^{N}\Big(\prod_{i=1}^{n}(x_{k}-t_{i})\prod_{j=1}^{m}(x_{k}-s_{j})^{-\beta/2}\Big)\right\rangle_{\mathrm{J\beta E}}
= (-1)^{(n- \frac{\beta}{2}m)N}\prod_{i=1}^{n}t_{i}^{N}\prod_{j=1}^{m} s_{j}^{-\frac{\beta}{2}N}\\
\cdot {}_2S\!F_{1}^{(\beta/2)}(-N, -N+1-(2/\beta)(1+\lambda_1);-2N+2-(2/\beta)(2+\lambda_1+\lambda_2);\frac{1}{t};\frac{1}{s}),
\end{multline}
which satisfies the finite $(n+m)$-system with  $\alpha=\beta/2, a=-N, b=(2/\beta)(1+\lambda_1+\lambda_2+p_{0})+N-1, c=(2/\beta)(\lambda_1+p_{0})$ and $p_{0}=n-(\beta/2)m$.

It is tempting to conclude from the last observation that the expectation given on the left-hand side of \eqref{Jacobiform2} is equal to ${}_2S\!F_{1}^{(\beta/2)}(a,b;c;t,s)$, with $a,b,c$ equal to the values just given.  However, the uniqueness theorem does not apply for this case since the expectation is not analytic in $(t,s)=(0\ldots,0)$.  In order to extend the range of possible values for $t$ and $s$, we introduce another type of correlation
function:
 \begin{multline}\label{Jacobiform3}
{\left\langle \prod_{k=1}^{N}\frac{\prod_{i=1}^{n}(x_{k}-t_{i})}{\prod_{j=1}^{m}(x_{k}-s_{j})^{\beta/2}}\right\rangle_{{\mathscr{C}}\mathrm{J\beta E}}=
\frac{1}{Z_N}\int_{\mathscr{C}}\ldots\int_{\mathscr{C}}
 \prod_{k=1}^{N}\frac{\prod_{i=1}^{n}(x_{k}-t_{i})}{\prod_{j=1}^{m}(x_{k}-s_{j})^{\beta/2}} D_{\lambda_1,\lambda_2,\beta/2}(x)dx_1\cdots dx_N,}
\end{multline}
where it is assumed that $\lambda_1\not\in\mathbb{Z}$, $\lambda_2\in\mathbb{N}$, $\beta\in2\mathbb{N}$, and $\mathscr{C}$ stands for the contour introduced for the definition of the function $K^\mathscr{C}_N$ in \eqref{deformedKanekoC}.  Moreover,
\begin{multline} Z_N:=(e^{2\pi\mathrm{i}\lambda'_1}-1)^NS_N(\lambda_1',\lambda_2,\lambda)=\\
(2\pi\mathrm{i})^N(-1)^{N\lambda_1'}\prod_{j=0}^{N-1}
\frac{\Gamma(1+\lambda+j\lambda)  \Gamma(1+\lambda_{2}+j\lambda)}
{\Gamma(1+\lambda)\Gamma(2+\lambda_{1}+\lambda_{2}+(N+j-1)\lambda)\Gamma(-\lambda'_{1}-j\lambda)} \end{multline}
where $\lambda=\beta/2$ and $\lambda_1'=\lambda_1+n-\lambda m$.
Note that the normalization constant $Z_N$ has been evaluated thanks to the identity $(\sin \pi z)/\pi =1/(\Gamma(z)\Gamma(1-z))$;  it is such that the right-hand side of \eqref{Jacobiform3} equals  1 when $s_j=0$ and $t_i=0$ for all $i,j$.   This allows us to apply the Theorem \ref{maintheorem} to Corollary \ref{corodeformedKanekoC} and obtain the exact expression for the correlation function.
\begin{proposition}Let $\lambda_1$, $\lambda_2$, as $\lambda$ as described above.  Then for all $t_i, s_j\in\mathbb{C}$ such that $|t_i|<1$ and $|s_j|<1$,
 \begin{multline}\label{Jacobiform4}
\left\langle \prod_{k=1}^{N}\frac{\prod_{i=1}^{n}(x_{k}-t_{i})}{\prod_{j=1}^{m}(x_{k}-s_{j})^{\beta/2}}\right\rangle_{{\mathscr{C}}\mathrm{J\beta E}}\\=  {}_2S\!F_1^{(\beta/2)}\left(-N,N-m-1+2(\lambda_1+\lambda_2+n+1)/\beta,2(n+\lambda_1)/\beta-m;t;s\right)
\end{multline}
\end{proposition}

We know from Lemma \ref{lemmapochhammer} that the last series  terminates if for some positive integer $M$,
\be \label{condMsuperJacobi}
N-m-1+\frac{2}{\beta}(\lambda_1+\lambda_2+n+1)=\frac{2M}{\beta}.\ee  In such a case, the series also becomes a rational function in $\beta$, well defined for all $\beta>0$.  By analytic continuation, the restriction $\beta=2\mathbb{N}$ can then be removed from the definition of the correlation \eqref{Jacobiform3}. Moreover, according to our discussion in Section \ref{sectsuperjacobipoly},  we also conclude that when the condition \eqref{condMsuperJacobi} is satisfied, the correlation function is proportional  to a super Jacobi polynomial $\overline{S\!J}^{(\gamma,\eta)}_{\kappa^\textrm{max}}(t,s;\beta/2)$, where $\gamma=2(\lambda_1+1)/\beta-1$ and $\eta=2(\lambda_2+1)/\beta-1$.

As immediate consequences of formulas \eqref{Jacobiform1} and \eqref{Jacobiform4}, we compute the following scaled limits  at the left hard edge of the spectrum (see \cite{forrester} for more details on the spectrum edges in Random Matrix Theory).
\begin{coro} With the appropriate choice of parameters $\lambda_1$, $\lambda_2$ and $\beta$, we have
 \begin{align*}
\lim_{N\to \iy}&\left\langle\prod_{k=1}^{N}\Big(\prod_{i=1}^{n}(1+\frac{2t_{i}}{N}x_{k})\prod_{j=1}^{m}(1+\frac{2s_{j}}{N}x_{k})^{-\beta/2}\Big)\right\rangle_{\mathrm{J\beta E}}=\prod_{i=1}^{n}e^{t_{i}}\prod_{j=1}^{m}e^{-(\beta/2) s_{j}}.
\end{align*}
and
 \begin{align*}
\lim_{N\to \iy}&\left\langle\prod_{k=1}^{N}\Big(\prod_{i=1}^{n}(x_k-\frac{t_{i}}{N^2})\prod_{j=1}^{m}(x_k-\frac{s_{j}}{N^2})^{-\beta/2}\Big)\right\rangle_{\mathscr{C}\mathrm{J\beta E}}={}_0S\!F_1^{(\beta/2)}\left(2(n+\lambda_1)/\beta-m;-t,-s\right).
\end{align*}
\end{coro}

\subsection{Laguerre $\beta$-Ensemble}
The Laguerre $\beta$-Ensemble (L$\beta$E) is characterized by the following joint probability density function:
\begin{equation}\label{PDFforlaguerre}
\frac{1}{W_{\lambda_{1}, \beta,N}}\prod_{i=1}^{N}x_{i}^{\lambda_{1}}e^{-\beta x_{i}/2} \,\prod_{1\leq j<k\leq
N}|x_{j}-x_{k}|^{\beta}.
\end{equation}
The normalization constant is easily derived from the Selberg integral \cite{mehta,forrester}:
\[W_{\lambda_{1}, \beta,N}=(\frac{2}{\beta})^{(1+\lambda_{1})N+\beta N(N-1)/2}\prod_{j=0}^{N-1}
\frac{\Gamma(1+\beta/2+j\beta/2) \Gamma(1+\lambda_{1}+j\beta/2)}
{\Gamma(1+\beta/2)}.\]

The density given in \eqref{PDFforlaguerre} can be obtained from the J$\beta$E density \eqref{PDFforjacobi} by making the following limiting change of variables: \be \label{changeJtoL}
x_{i}\longmapsto x_{i}/L, \qquad \lambda_{2}\longmapsto \beta L/2,\qquad L\longmapsto\infty. \ee
By applying the same limit transformation to the correlation \eqref{Jacobiform3} and assuming that the variables  $s_{i}$ and $t_j$ transform as $x_j$,  we get the following formula involving $\mathscr{C}$, which now denotes a Hankel type contour such that each variable $x_k$ starts at $\infty+\mathrm{i}0^+$, then encircles   all the variables $s_j$ in the positive direction without crossing the interval $[0,\infty)$, and ends at the point $\infty+\mathrm{i}0^-$.

\begin{proposition} We have
\begin{equation}\label{LaguerreC}
\left\langle \prod_{k=1}^{N}\frac{\prod_{i=1}^{n}(x_{k}-t_{i})}{\prod_{j=1}^{m}(x_{k}-s_{j})^{\beta/2}}\right\rangle_{{\mathscr{C}}\mathrm{L\beta E}} =  {}_1S\!F_1^{(\beta/2)}\left(-N;2(n+\lambda_1)/\beta-m;t;s\right),
\end{equation}
where   the normalization is chosen so that  the right-hand side of \eqref{LaguerreC} is equal to unity when $t_i=0$ and $s_j=0$ for all $i,j$.
\end{proposition}

The scaled limit at the hard edge is easily computed.
\begin{coro} We have
\begin{equation*}
\lim_{N\to\infty}\left\langle \prod_{k=1}^{N}\frac{\prod_{i=1}^{n}(x_{k}-t_{i}/N)}{\prod_{j=1}^{m}(x_{k}-s_{j}/N)^{\beta/2}}\right\rangle_{{\mathscr{C}}\mathrm{L\beta E}} =  {}_0S\!F_1^{(\beta/2)}\left(2(n+\lambda_1)/\beta-m;t;s\right).
\end{equation*}
\end{coro}

Now, when applying the same transformation \eqref{changeJtoL} to the expectation values given in  equations \eqref{Jacobiform1} and \eqref{Jacobiform2},   one formally gets
\begin{multline}\label{Laguerreform1}
\left\langle\prod_{k=1}^{N}\Big(\prod_{i=1}^{n}(1-t_{i}x_{k})\prod_{j=1}^{m}(1-s_{j}x_{k})^{-\beta/2}\Big)\right\rangle_{\mathrm{L\beta E}}
=\\
{}_2SF_{0}^{(\beta/2)}(-N, -N+1-(2/\beta)(1+\lambda_1);-t,-s),
\end{multline}
and
\begin{multline}\label{Laguerreform2}
\left\langle\prod_{k=1}^{N}\Big(\prod_{i=1}^{n}(x_{k}-t_{i})\prod_{j=1}^{m}(x_{k}-s_{j})^{-\beta/2}\Big)\right\rangle_{\mathrm{L\beta E}}
=(-1)^{(n-\frac{\beta}{2}m)N}\prod_{i=1}^{n}t_{i}^{N}\prod_{j=1}^{m} s_{j}^{-\frac{\beta}{2}N}\\
{}_2SF_{0}^{(\beta/2)}(-N, -N+1-(2/\beta)(1+\lambda_1);-\frac{1}{t},-\frac{1}{s}).
\end{multline}
As we already pointed out, if $a$ and $b$ are
not equal respectively equal to $-N$ and $ M/\alpha$ for some $N,M\in \mathbb{N}_{0}$, the  power series for ${}_2S\!F_{0}^{(\alpha)}(a,b;t,s)$ is in general divergent (see Proposition \ref{propconvergence}, part (3)).  Thus, apparently, formulas   \eqref{Laguerreform1} and  \eqref{Laguerreform2} do not make sense.  However, as will be shown below, there is an integral representation for ${}_2S\!F_{0}^{(\beta/2)}$ that allows us to circumvent the convergence difficulty.

Let us first consider the expectation of a Jack polynomial in the L$\beta$E, which can be computed by making use of formula \eqref{jackintegral}:
\begin{align}
\left\langle P_{\kappa}^{(\alpha)}(x)\,\right\rangle_{\mathrm{L\beta E}}
= P_{\kappa}^{(\alpha)}(1^{N})\,\alpha^{|\kappa|} \lbrack \lambda_1+1+(N-1)/\alpha\rbrack^{(\alpha)}_\kappa
 .\end{align}
On the right-hand side, we make the change of variables $x_{i}=u y_{i}$, where $0<u<\infty$ and $\sum_{i} y_{i}=1$, which leads to
 \begin{multline}\label{fixedjackintegral}
C_{\Delta}\int_{\Delta_{N}}P_{\kappa}^{(\alpha)}(y)\prod_{i=1}^{N}y_{i}^{\lambda_{1}} \,\prod_{ j<k}|y_{j}-y_{k}|^{2/\alpha}d \sigma_{N}=
\frac{\Gamma(c_{0})}{\Gamma(c_{0}+|\kappa|)}P_{\kappa}^{(\alpha)}(1^{N}) \lbrack \lambda_1+1+(N-1)/\alpha\rbrack^{(\alpha)}_\kappa.
\end{multline}
In this equation,  $\Delta_{N}$ denotes the set of points in the hyperplane $\sum_{i} y_{i}=1$ with $ y_{i}\geq 0$,  $d \sigma_{N}$ stands for the Lebesgue measure  on  $\Delta_{N}$ such that the left-hand side of \eqref{fixedjackintegral} equals  to one when $\kappa=(0)$, while $C_{\Delta}=\alpha^{c_{0}}\Gamma(c_{0})/W_{\lambda_{1}, 2/\alpha,N}$ and $c_{0}=\lambda_1 N+N(N-1)/\alpha +N$.

Following the method exposed in \cite{bo}, which is concerned with the Jack polynomial case, we now define a new type of super series
\begin{equation}
   {}_2{S\!\hat{F}}^{(\alpha)}_{1}(a,b;c;t;s) =
   \sum_{k=0}^{\infty}\sum_{|\kappa|=k} \frac{1}{h_{\kappa}^{(\alpha)}}
   \frac{
   \lbrack a\rbrack^{(\alpha)}_\kappa \lbrack b\rbrack^{(\alpha)}_\kappa}{(c)_{|\kappa|}} S\!P_{\kappa}^{(\alpha)}(t,s),\qquad a, b, c\in \mathbb{C},\, c \neq 0, -1, -2, \ldots .
\end{equation}
Then, proceeding as for Propositions \ref{propintsuperseries} or \ref{gammasuperselbergintegral} and using \eqref{fixedjackintegral}, we get
\begin{multline}
C_{\Delta}\int_{\Delta_{N}}\prod_{k=1}^{N}\frac{\prod_{i=1}^{n}(1-x_{k}t_{i})}{\prod_{i=1}^{m}(1-x_{k}s_{j})^{1/\alpha}}\prod_{i=1}^{N}x_{i}^{\lambda_{1}} \,\prod_{ j<k}|x_{j}-x_{k}|^{2/\alpha}d \sigma_{N}=\\{}_2 S\!\hat{F}^{(1/\alpha)}_{1}(-N, -N+1-\alpha (1+\lambda_{1});\lambda_1 N+N(N-1)/\alpha +N;-\frac{1}{\alpha}t,-\frac{1}{\alpha}s).
\end{multline}
Furthermore,
\begin{multline}\label{integralforlaguerre}
\left\langle\prod_{k=1}^{N}\Big(\prod_{i=1}^{n}(1-t_{i}x_{k})\prod_{j=1}^{m}(1-s_{j}x_{k})^{-\beta/2}\Big)\right\rangle_{\mathrm{L\beta E}}
=\frac{1}{\Gamma(c_{0})}\int_{0}^{\infty} d u\\u^{c_{0}-1} e^{-u}
{}_2S\!\hat{F}_{1}^{(\beta/2)}(-N, -N+1-(\beta/2)(1+\lambda_1);c_{0};-u t,-u s),
\end{multline}
where $c_{0}=\lambda_1 N+(2/\beta)N(N-1) +N$.

At this point, it should be stressed that the right-hand side of \eqref{integralforlaguerre} is formally equal to that of \eqref{Laguerreform1}.    Indeed, if one expands ${}_2S\!\hat{F}_{1}^{(\beta/2)}$ and then applies the integral
transform $\int_0^\infty du\,u^{c_0}e^{-u}$ term by term , one recovers the formal power series defining ${}_2S\!F_{0}^{(\beta/2)}(-N, -N+1-(\beta/2)(1+\lambda_1); - t,-s)$.  This suggests a new definition for ${}_2S\!F_{0}^{(\alpha)}$:
\be \label{eqcorrec2F0} {}_2S\!F_{0}^{(\alpha)}(a,b;t;s):=\frac{1}{\Gamma(c_{0})}\int_{0}^{\infty}  u^{c_{0}-1} e^{-u}{}_2 S\!\hat{F}^{(\alpha)}_{1}(a,b;c_{0};t;s) d u,\ee The integral is well-defined at least when $t_{i}>0, s_{j}>0$ for all $i,j$, $a=-N$ and $b=-N+1-(\beta/2)(1+\lambda_1)$.   Our new definition thus resolves the convergence problem. In what follows, the notation ${}_2S\!F_{0}^{(\alpha)}$ is used in the sense of \eqref{eqcorrec2F0}.

Notice that the function ${}_2 S\!\hat{F}^{(\alpha)}_{1}(a,b;c;t;s)$
 shares some properties with the standard super hypergeometric functions ${}_2 S\!F^{(\alpha)}_{1}(a,b;c;t;s)$.

\begin{proposition}

(i)  The power series  defining ${}_2 S\!\hat{F}^{(\alpha)}_{1}(a,b;c;t,s)$  converges in the polydisk $\{|t_i|<1,|s_j|<1, i=1,\ldots,n;j=1,\ldots,m\}$.

\noindent(ii) We have the duality relation  \begin{align}
  {}_2S\!\hat{F}_{1}^{(\alpha)}(a,b; c; t,s)= {}_2S\!\hat{F}_{1}^{(1/\alpha)}(-\alpha a, -\alpha b;  c; -s/\alpha,-t/\alpha).
   \end{align}
\end{proposition}

\begin{proof} (i)  One compares the series
 ${}_2 S\!\hat{F}^{(\alpha)}_{1}(a,b;c;t,s)$  with the series
 $${}_1S\!F^{(\alpha)}_{0}(a;t,s)=\sum_{k=0}^{\infty}\sum_{|\kappa|=k} \frac{\lbrack a\rbrack^{(\alpha)}_\kappa}{h_{\kappa}^{(\alpha)}} S\!P_{\kappa}^{(\alpha)}(t,s).$$
By virtue of the formula $${}_1S\!F_{0}^{(\alpha)}(a;t,s)=\prod_{i=1}^{n}(1-t_{i})^{-a}\prod_{j=1}^{m}(1-s_{j})^{\alpha a},$$
 the series on the left-hand side converges in the polydisk in question.  Now, given that the ratio $ \lbrack b\rbrack^{(\alpha)}_\kappa /(c)_{|\kappa|}$
 has at most polynomial growth in $|\kappa|$, the former series also converges in
the same polydisk.

 (ii) One does the same calculation as in the proof of Proposition 7.2 in \cite{dh}.
\end{proof}

We return to our main subject. The above computations allow to give the following exact expressions for the averages of ratios of characteristic polynomials in L$\beta$E.
\begin{proposition}\label{proLaguerre} Let ${}_2S\!F_{0}^{(\beta/2)}$ denote the function defined by \eqref{eqcorrec2F0}.  Then,
\begin{multline}\label{Laguerreform3}
\left\langle\prod_{k=1}^{N}\Big(\prod_{i=1}^{n}(1-t_{i}x_{k})\prod_{j=1}^{m}(1-s_{j}x_{k})^{-\beta/2}\Big)\right\rangle_{\mathrm{L\beta E}}
=\\
{}_2S\!F_{0}^{(\beta/2)}(-N, -N+1-(2/\beta)(1+\lambda_1);-t;-s),
\end{multline}
and
\begin{multline}
\left\langle\prod_{k=1}^{N}\Big(\prod_{i=1}^{n}(x_{k}-t_{i})\prod_{j=1}^{m}(x_{k}-s_{j})^{-\beta/2}\Big)\right\rangle_{\mathrm{L\beta E}}
=(-1)^{(n-\frac{\beta}{2}m)N}\prod_{i=1}^{n}t_{i}^{N}\prod_{j=1}^{m} s_{j}^{-\frac{\beta}{2}N}\\
{}_2S\!F_{0}^{(\beta/2)}(-N, -N+1-(2/\beta)(1+\lambda_1);-\frac{1}{t};-\frac{1}{s}).
\end{multline}
\end{proposition}

\begin{coro}We have
 \begin{align*}
\lim_{N\to \iy}&\left\langle\prod_{k=1}^{N}\Big(\prod_{i=1}^{n}(1+\frac{t_{i}}{N^{2}}x_{k})\prod_{j=1}^{m}(1+\frac{s_{j}}{N^{2}}x_{k})^{-\beta/2}\Big)\right\rangle_{\mathrm{L\beta E}}=\prod_{i=1}^{n}e^{t_{i}}\prod_{j=1}^{m}e^{-(\beta/2) s_{j}}.
\end{align*}
\end{coro}
\begin{proof} Combining \eqref{integralforlaguerre} and \eqref{Laguerreform3}, we find
\begin{multline*}{}_2S\!F_{0}^{(\beta/2)}(-N, -N+1-(2/\beta)(1+\lambda_1);\frac{t}{N^{2}},\frac{s}{N^{2}})=\\
\frac{(2N^{2}/\beta)^{c_{0}}}{\Gamma(c_{0})}\int_{0}^{\infty} d u u^{c_{0}-1} e^{-(2N^{2}/\beta)u}
{}_2S\!\hat{F}_{1}^{(\beta/2)}(-N, -N+1-(\beta/2)(1+\lambda_1);c_{0};(2/\beta)u t,(2/\beta)u s).
\end{multline*}
Since, as $N\to \infty$,
$${}_2S\!\hat{F}_{1}^{(\beta/2)}(-N, -N+1-(\beta/2)(1+\lambda_1);c_{0};(2/\beta)u t,(2/\beta)u s)\longrightarrow {}_0S\!F_{0}^{(\beta/2)}(u t,u s),$$
and $$\frac{(2N^{2}/\beta)^{c_{0}}}{\Gamma(c_{0})} u^{c_{0}-1} e^{-(2N^{2}/\beta)u}\longrightarrow \delta(u-1),$$
 we get the desired result from the fact that ${}_0S\!F_{0}^{(\beta/2)}(t,s)=\prod_{i=1}^{n}e^{t_{i}}\prod_{j=1}^{m}e^{-(\beta/2)s_{j}}$.
\end{proof}

The second integral in Proposition \ref{proLaguerre} can be used to give expressions for some spacing distributions at the
hard edge of the spectrum. Now, by definition  (see \cite{forrester} for more details)
\begin{align}\label{spacedistribution}
E_{N,\beta}(0;&(0,s);x^{a}e^{-\beta x/2}):=\frac{1}{W_{a, \beta,N}}\int_{[s,\iy)^{N}}\prod_{i=1}^{N}x_{i}^{a}e^{-\beta x_{i}/2} \,\prod_{1\leq j<k\leq
N}|x_{j}-x_{k}|^{\beta}d^{N}x\nonumber\\
&=\frac{e^{-N\beta s/2}}{W_{a, \beta,N}}\int_{[0,\iy)^{N}}\prod_{i=1}^{N}(x_{i}+s)^{a}e^{-\beta x_{i}/2} \,\prod_{1\leq j<k\leq
N}|x_{j}-x_{k}|^{\beta}d^{N}x,
\end{align}
where the second equality follows from the change of variables $x_{i}\mapsto x_{i}+s$. From this, the
distribution of the smallest eigenvalue is given by
\begin{align}\label{smallestdistribution}
&p^{(N)}_{\beta}(0;s;a)=-\frac{d}{d s}E_{N,\beta}(0;(0,s);x^{a}e^{-\beta x/2})\nonumber\\
&=\frac{N e^{-N\beta s/2}}{W_{a, \beta,N}}s^{a}\int_{[0,\iy)^{N-1}}\prod_{i=1}^{N-1}x_{i}^{\beta}(x_{i}+s)^{a}e^{-\beta x_{i}/2} \,\prod_{1\leq j<k\leq
N-1}|x_{j}-x_{k}|^{\beta}d^{N-1}x,
\end{align}
where the second equality follows from the  differentiation of the first equality in \eqref{spacedistribution} and then from the
change of variables $x_{i}\mapsto x_{i}+s$.

For $a\in \mathbb{N}_{0}$,  the final integrals in both \eqref{spacedistribution} and \eqref{smallestdistribution} can be expressed as   (non super) hypergeometric functions \cite[Proposition 13.2.6]{forrester}. With the aid of integral \eqref{Laguerreform2} they can also
be written in terms of the super hypergeometric functions, but this time,  for more general values of $a$ . %%For convenience we write $$L_{\beta}=\{a=n-(\beta/2)m|n,m\in \mathbb{N}_{0}, a>-1\}.$$

\begin{proposition}For $a=n-(\beta/2)m>-1$ with $n,m\in \mathbb{N}_{0}$, we have
\begin{align*}
E_{N,\beta}(0;(0,s);x^{a}e^{-\beta x/2})=\frac{e^{-N\beta s/2}}{W_{a, \beta,N}}s^{(n-\frac{\beta}{2}m)N}{}_2S\!F_{0}^{(\beta/2)}\left(-N, -N+1-{2}/{\beta};\Big(\frac{1}{s}\Big)^{n};\Big(\frac{1}{s}\Big)^{m}\right)
\end{align*}
and
\begin{align*}
p^{(N)}_{\beta}(0;s;a)=\frac{W_{\beta, \beta,N}}{W_{a, \beta,N}}N e^{-N\beta s/2}s^{(n-\frac{\beta}{2}m)N}{}_2S\!F_{0}^{(\beta/2)}\left(-N+1, -N-{2}/{\beta};\Big(\frac{1}{s}\Big)^{n};\Big(\frac{1}{s}\Big)^{m}\right).
\end{align*}
\end{proposition}

\subsection{Gaussian $\beta$-Ensemble}
The Gaussian $\beta$-Ensemble (G$\beta$E) refers to a set of real random variables $x=(x_1,\ldots,x_N)$ whose joint probability density function is given
\begin{equation}\label{PDFforgauss}
\frac{1}{G_{\beta,N}}\prod_{i=1}^{N}e^{-\beta x^{2}_{i}/2} \,\prod_{1\leq j<k\leq
N}|x_{j}-x_{k}|^{\beta},
\end{equation}
where the normalization \cite{mehta,forrester}
\[G_{\beta,N}=\beta^{-N/2-\beta N(N-1)/4}(2\pi)^{N/2}\prod_{j=0}^{N-1}
\frac{\Gamma(1+\beta/2+j\beta/2)}
{\Gamma(1+\beta/2)}.\]

 From our perspective, the expectation values of ratios of products of characteristic polynomials in the Gaussian $\beta$-Ensemble  seem more difficult to tackle.  We nevertheless have the following.
\begin{proposition}
\begin{equation}\label{Gaussform2}
\left\langle\prod_{k=1}^{N}\Big(\prod_{i=1}^{n}(x_{k}-t_{i})\prod_{j=1}^{m}(x_{k}-s_{j})^{-\beta/2}\Big)\right\rangle_{\mathrm{G\beta E}}
=I_{n,m,N}
\end{equation}
where \begin{equation*}{I_{n,m,N}=
\lim_{L\to \iy}(-L)^{p_{0}N}{}_2S\!F_{1}^{(\beta/2)}\left(-N, \frac{2}{\beta}(1+p_0)+2L^{2};\frac{2}{\beta}p_0+L^{2};\frac{1}{2}(1-\frac{t}{L}),\frac{1}{2}(1-\frac{s}{L})\right),}
\end{equation*}
and $p_0 =n-(\beta/2)m$.
\end{proposition}
\begin{proof}
In \eqref{Jacobiform2}, we set \be \label{changingvariables} x_{k}\mapsto \frac{1}{2}(1-\frac{x_{k}}{L}),\quad
t_{i}\mapsto \frac{1}{2}(1-\frac{t_{i}}{L}),\qquad s_{j}\mapsto \frac{1}{2}(1-\frac{s_{j}}{L}),\qquad \lambda_{1}=\lambda_{2}= \beta L^{2}/2\ee and
take the limit $L \to \iy$. Note that under the above change of variables the density in the J$\beta$E  becomes the density in the G$\beta$E. Finally, the properties of the Gamma function imply
\begin{align*}
\lim_{L\to \iy}\frac{S_{N}(\beta L^{2}/2+p_{0},\beta L^{2}/2,\beta/2)}{S_{N}(\beta L^{2}/2,\beta L^{2}/2,\beta/2)}=2^{-p_{0}N}.
\end{align*}
and the result follows.
\end{proof}

One easily verifies that $I_{n,m,N}$ does not reduce to a single
super hypergeometric series. Nonetheless, it  satisfies a deformed holonomic system of differential equations, which is obtained by substituting \eqref{changingvariables} into \eqref{PDEs1} and \eqref{PDEs2},  but with
 $\alpha=\beta/2$, $a=-N$, $b=(2/\beta)(1+\lambda_1+\lambda_2+p_{0})+N-1$, $c=(2/\beta)(\lambda_1+p_{0})$.

\begin{proposition}Let  $\alpha=\beta/2$, $a=-N$, and $b=2$.  Then, $I_{n,m,N}$ satisfies the deformed system
\begin{multline}\label{PDEs1forgauss}
\frac{\partial^2 F}{\partial t_i^2} -b t_i \frac{\partial F}{\partial t_i} - abF  +
\frac{1}{\alpha}\sum_{ k=1, k \ne i}^{n}
\frac{1}{t_i - t_k} \biggl( \frac{\partial F}{\partial t_i}- \frac{\partial F}{\partial t_k} \biggr)\\
-\sum_{ k=1 }^{m}
\frac{1}{t_i - s_k} \biggl( \frac{\partial F}{\partial t_i}
+ \frac{1}{\alpha} \frac{\partial F}{\partial s_k} \biggr)
 = 0,  i=1, \dots, n,
\end{multline}
\begin{multline}\label{PDEs2forgauss}
-\frac{1}{\alpha}\frac{\partial^2 F}{\partial s_j^2} -b s_j \frac{\partial F}{\partial s_j} - (-\alpha)abF
-\sum_{ k=1, k \ne j}^{m}
\frac{1}{s_j - s_k} \biggl(
 \frac{\partial F}{\partial s_j}
-  \frac{\partial F}{\partial s_k} \biggr)\\
+\sum_{k=1}^{n}
\frac{1}{s_j - t_k} \biggl(\frac{1}{\alpha}\frac{\partial F}{\partial s_j}+ \frac{\partial F}{\partial t_k} \biggr)
 = 0, j=1, \dots, m,
\end{multline}
and the cancellation property
 \begin{equation}
 \left(\frac{\partial F}{\partial t_i}+\frac{1}{\alpha}
 \frac{\partial F}{\partial s_j} \right)_{t_i=s_j}=0.
\end{equation}
\end{proposition}

\begin{remark}
By interchanging \eqref{PDEs1forgauss} and \eqref{PDEs2forgauss},
we   can rewrite the deformed system  in the same form, except that $\alpha$, $a$, and $b$ are respectively replaced by  $1/\alpha$, $-\alpha a$, and $-\alpha b$. This shows that the solution of the system still enjoys  a duality symmetry.
\end{remark}

\subsection{Circular $\beta$-Ensemble}
By making use of the general identity (see Eq.(1.16)--(1.18) in  \cite{fw} or \cite{forrester} for more details about this argument)
\begin{align}\label{gerneralidentity}
&\int_{[0,1]^{N}} (x_{1}\cdots x_{N})^{\zeta-1}f(x_{1},\ldots, x_{N}) d^{N}x\nonumber\\
&=\frac{1}{(2\sin\pi\zeta)^{N}}\int_{[-\pi,\pi]^{N}} e^{i(\theta_{1}+\cdots+ \theta_{N})\zeta}f(-e^{i\theta_{1}},\ldots, -e^{i\theta_{N}}) d^{N}\theta,
\end{align}
valid for $f$ a Laurent polynomial and $\mathrm{Re}(\zeta)$ large enough so that the left-hand side
exists,  the Selberg integral formula \eqref{Selbergformula} can be transformed into its trigonometric form
\begin{align}S_{N} (\lambda_{1},\lambda_{2},\lambda)=(-1)^{N+N(N-1)\lambda/2}\frac{1}{(2\sin\pi b)^{N}}M_{N}(a,b,\lambda),
\end{align}
where $\lambda_{1}:=-b-(N-1)\lambda-1, \lambda_{2}:=a+b$ and the Morris integral
\begin{align}\label{morrisintegral}
M_{N}(a,b,\lambda):&=
\int_{[-\pi,\pi]^{N}} \prod_{k=1}^{N} e^{\frac{1}{2}i(a-b)\theta_{k}}|1+e^{i\theta_{k}}|^{a+b}
\prod_{j<k} |e^{i\theta_{j}}-e^{i\theta_{k}}|^{2\lambda} d^{N}\theta\\
&=(2\pi)^{N}\prod_{j=0}^{N-1}
\frac{\Gamma(1+\lambda+j\lambda)\Gamma(1+a+b+j\lambda)  }
{\Gamma(1+\lambda)\Gamma(1+a+j\lambda)\Gamma(1+b+j\lambda)}\end{align}
for $a, b, \lambda\in \mathbb{C}$ such that $\mathrm{Re} (a+b+1)>0, \mathrm{\lambda}>-\min\{ 1/N, \mathrm{Re} (a+b+1)/(N-1)\}$.

The above formulas readily imply that
Theorem \ref{superselbergintegral} and formula \eqref{Jacobiform1} have a trigonometric counterpart.

\begin{proposition}
 \begin{multline}\label{morriscorrelationintegral}
\frac{1}{ M_{N}(a,b,1/\alpha)}\int_{[-\pi,\pi]^{N}} \left(\prod_{k=1}^{N}\prod_{k=1}^{N}e^{\frac{1}{2}i(a-b)\theta_{k}}|1+e^{i\theta_{k}}|^{a+b}\prod_{p=1}^{n}
(1+t_{p}e^{i\theta_{k}})\right.\\
\left. \prod_{q=1}^{m}(1+s_{q}e^{i\theta_{k}})^{-1/\alpha}\right)
\prod_{j<k} |e^{i\theta_{j}}-e^{i\theta_{k}}|^{2/\alpha} d^{N}\theta=
{}_2S\!F_{1}^{(1/\alpha)}(-N, \alpha b;-N+1-\alpha (1+a);t;s).\end{multline}
\end{proposition}

Now, the joint probability density function for the Circular $\beta$-Ensemble (C$\beta$E) is equal to
\begin{equation}\label{PDFforcircular}
\frac{1}{C_{\beta,N}}\prod_{1\leq j<k\leq
N}|e^{i\theta_{j}}-e^{i\theta_{k}}|^{\beta},
\end{equation}
where the normalization constant is $C_{\beta,N}=(2\pi)^{N}
 {\Gamma(1+N\beta/2)}
{(\Gamma(1+\beta/2))^{-N}}$ \cite{mehta,forrester}.  More generally, when $a=b$,$t=0$ and $s=0$, the integrand on the left-hand side of \eqref{morriscorrelationintegral} is real  and  referred to as defining
the Circular Jacobi $\beta$-Ensemble (CJ$\beta$E) \cite{forrester}.  The  formula below directly follows from  \eqref{morriscorrelationintegral}.
\begin{proposition}
\begin{multline}\label{circularjacobiform1}
\left\langle\prod_{k=1}^{N}\Big(\prod_{i=1}^{n}(1+t_{i}e^{i\theta_{k}})\prod_{j=1}^{m}(1+s_{j}e^{i\theta_{k}})^{-\beta/2}\Big)\right\rangle_{\mathrm{CJ\beta E}}
=\\
{}_2S\!F_{1}^{(\beta/2)}(-N, (2/\beta) b;-N+1-(2/\beta)(1+b);t,s).
\end{multline}
\end{proposition}

\section{A more general deformed Selberg integral}

We end the article by explicitly computing a more general Selberg integral than $S_N(\lambda_1,\lambda_2,\lambda;t;s)$.  The approach here differs from the previous sections in the sense that no deformed holonomic system for this integral is yet known  and the series considered are kept at the formal level.

We first need to recall and then generalize  some concepts  exposed in Section \ref{superfunction}.
Let $x=(x_{1},  \ldots, x_{n})$, $y=(y_{1},  \ldots, y_{m})$,  and $\gamma\in \mathbb{F}$.  The  deformed power sums, which are defined as  $$p_{r,\gamma}(x,y)=p_{r}(x)-\gamma p_{r}(y),  \qquad \ r\in \mathbb{N},$$
generate an algebra denoted by $\mathcal{N}_{\gamma}(x,y)$.
There is a natural  homomorphism   $\varphi^{(\gamma)}_{n,m}\,:\,\Lambda(z) \to \mathcal{N}_{\gamma}(x,y)$ given by
$$\varphi^{(\gamma)}_{n,m}(p_{r})=p_{r,\gamma}(x,y).$$
We define the $\gamma$-generalized super Jack polynomial $S\!P_{\kappa}^{(\alpha,\gamma)}(x,y)$  as the image  of the Jack symmetric function $P_{\kappa}^{(\alpha)}$ under $\varphi^{(\gamma)}_{n,m}$, that is,
\be \label{gammasuperjackp}
S\!P_{\kappa}^{(\alpha,\gamma)}(x, y)=\varphi^{(\gamma)}_{n,m}(P_{\kappa}^{(\alpha)}).
\ee
We remark in particular that if $\gamma=r\alpha$ for some $r\in \mathbb{N}$, then
\[S\!P_{\kappa}^{(\alpha,\gamma)}(x_{1},  \ldots, x_{n}, y_{1},  \ldots, y_{n})=S\!P_{\kappa}^{(\alpha)}(x_{1},  \ldots, x_{n}, z_{1},  \ldots, z_{rn})\]
whenever  $z_{1}=  \cdots= z_{r}=y_{1}$, $\ldots$, $z_{r(n-1)+1}=  \cdots= z_{rn}=y_{n}$.

We actually know little about the $\gamma$-super Jack polynomials. For instance, we have not identified the kernel of
$\varphi^{(\gamma)}_{n,m}$ in the $\mathbb{C}$-span of the all Jack symmetric functions $P^{\alpha}_\lambda$.  This does not cause any problem however for the formal definition of the $\gamma$-super  hypergeometric series.

\begin{definition}\label{gammaHFpqDef}
Fix $p,q\in\mathbb{N}_0$ and let $a_1,\ldots,a_p, b_1,\ldots,b_q$ be complex numbers such that $(i-1)/\alpha-b_j\notin\mathbb{N}_0$ for all $i\in\mathbb{N}_0$.
We then define the hypergeometric series
\begin{equation}\label{gammahfpq}
    {}_pS\!F^{(\alpha,\gamma)}_q(a_1,\ldots,a_p;b_1,\ldots,b_q;x,y): = \sum_{k=0}^{\infty}\sum_{|\kappa|=k} \frac{\lbrack a_1\rbrack^{(\alpha)}_\kappa\cdots\lbrack a_p\rbrack^{(\alpha)}_\kappa}{\lbrack b_1\rbrack^{(\alpha)}_\kappa
    \cdots\lbrack b_q\rbrack^{(\alpha)}_\kappa}\frac{S\!P_{\kappa}^{(\alpha,\gamma)}(x,y)}{h_{\kappa}^{(\alpha)}}.
\end{equation}
\end{definition}

The following proposition  is a formal generalization of Theorems \ref{superselbergintegral} and \ref{othertheorem}.
\begin{proposition} \label{gammasuperselbergintegral} As formal power series,
\begin{multline*}
\frac{1}{S_{N}(\lambda_{1},\lambda_{2},1/\alpha)}\int_{[0,1]^{N}}
 \prod_{i=1}^{N}\frac{\prod_{j=1}^{n}(1-x_{i}t_{j})}{\prod_{k=1}^{m}(1-x_{i}s_{k})^\gamma}\, D_{\lambda_{1},\lambda_{2},1/\alpha}(x)d^{N}x
\\={}_2S\!F_{1}^{(1/\alpha,\gamma)}(-N, -N+1-\alpha(1+\lambda_1);-2N+2-\alpha(2+\lambda_1+\lambda_2);t;s).
\end{multline*}
\end{proposition}
\begin{proof}
First, we recall identity (5.4) in Chaper VI of \cite{macdonald} for the Jack symmetric functions:
$$ \prod_{i,j=1}^{\iy}(1-x_{i}z_{j})=
\sum_{\kappa} (-1)^{|\kappa|}P_{\kappa}^{(\alpha)}(x)P_{\kappa'}^{(1/\alpha)}(z).$$  Second, we act with $\varphi^{(\gamma)}_{n,m}$ on both sides with respect to $z$ and then restrict the number of variables $x$ by setting $x_i=0$ for all $i>N$.  This yields
the identity
\be \label{deformedidentity}\prod_{k=1}^{N}\frac{\prod_{i=1}^{n}(1-x_{k}t_{i})}{\prod_{i=1}^{m}(1-x_{k}s_{j})^{\gamma}}=
\sum_{\kappa} (-1)^{|\kappa|}P_{\kappa}^{(\alpha)}(x)S\!P_{\kappa'}^{(1/\alpha,\gamma)}(t,s),
\ee
where now $x=(x_1,\ldots,x_N)$.    Note that the left-hand side of \eqref{deformedidentity} immediately follows from
\be \log\prod_{k=1}^{\iy}\frac{\prod_{i=1}^{\iy}(1-x_{k}t_{i})}{\prod_{i=1}^{\iy}(1-x_{k}s_{j})^{\gamma}}=
\sum_{k\geq 1} \frac{-1}{k}p_{k}(x)\big(p_{k}(t)-\gamma p_{k}(s)\big).
\ee

Now, we know from formula \eqref{jackintegral} the value of the integration of a Jack polynomial $P^{(\alpha)}_\kappa(x)$ with respect the Selberg density.  But according to result (10.20) of Chapter VI in \cite{macdonald},
 \begin{align}\label{evaluationat1}
 P_{\kappa}^{(\alpha)}(1^{N})= \prod_{(i,j)\in\kappa}\frac{N+\alpha a^\prime_\kappa(i,j)-l^\prime_\kappa(i,j)}{1+\alpha a_\kappa(i,j)+l_\kappa(i,j)}
 =\frac{\alpha^{|\kappa|} \lbrack N/\alpha\rbrack^{(\alpha)}_\kappa}{h_{\kappa'}^{(1/\alpha)}}
 ,\end{align}
 where the second equality follows from \eqref{defhook}--\eqref{colengths}. Note that since
 \begin{equation*}
    [N/\alpha]^{(\alpha)}_\kappa = \prod_{1\leq i\leq \ell(\kappa)}\left(\frac{N}{\alpha}-\frac{i-1}{\alpha}\right)_{\kappa_i},
\end{equation*}
 we see from \eqref{evaluationat1} that $P_{\kappa}^{(\alpha)}(1^{N})=0$
 whenever $\ell(\kappa)>N$. Thus, the formula  \eqref{jackintegral} holds true for any partition $\kappa$.

Finally, combining \eqref{jackintegral}, \eqref{evaluationat1} and the property
\[[a]^{(\alpha)}_{\kappa}=(-\alpha)^{-|\kappa|}[-\alpha a]^{(1/\alpha)}_{\kappa'},\]
on can integrate term by term on the right-hand side of \eqref{deformedidentity} and get the desired formula.
\end{proof}

Similar calculations lead to the following trigonometric version of the generalized integral.
\begin{proposition}As formal power series,
 \begin{multline*}\label{gammamorriscorrelationintegral}
\frac{1}{ M_{N}(a,b,1/\alpha)}\int_{[-\pi,\pi]^{N}} \left(\prod_{p=1}^{n}
(1+t_{p}e^{i\theta_{k}})\prod_{q=1}^{m}(1+s_{q}e^{i\theta_{k}})^{-\gamma}\nonumber\right.\\
\left.
\prod_{k=1}^{N} e^{\frac{1}{2}i(a-b)\theta_{k}}|1+e^{i\theta_{k}}|^{a+b}\prod_{j<k} |e^{i\theta_{j}}-e^{i\theta_{k}}|^{2/\alpha} \right)d^{N}\theta
\\=
{}_2S\!F_{1}^{(1/\alpha,\gamma)}(-N, \alpha b;-N+1-\alpha (1+a);t;s).\end{multline*}
\end{proposition}

\begin{acknow}

The work of P.~D.\ was  supported by FONDECYT grant \#1090034 and by CONICYT through the Anillo de Investigaci\'on ACT56.
The work of D.-Z.~L.\ was  supported by FONDECYT grant \#3110108. D.-Z.~L. is also grateful to Zheng-Dong Wang for the hospitality during author's stay at Peking University in summer 2010.
\end{acknow}

\begin{appendix}
\section{Proof of Lemma \ref{symmetrylemma}} \label{Appendixsymmetrylemma}

It is sufficient to prove the lemma in the case where $i=1$.  Suppose that $\lambda_{k_{0}}(r)\neq 0$ for some $k_{0}\in \mathbb{N}_{0}$, writing $\lambda_{k_{0}}(r)=\sum_{I\in \mathfrak{L}}C_{I}r^{I}$, then
there exists $I_{0}=(i_{1},\ldots,i_{n_{0}},0,\ldots)$ such that $C_{I_{0}}\neq 0$. Taking $n=\max\{k_{0}+1,n_{0}\}$ and putting $y_{j}=0$ for any $j>n$, the equation (\ref{linearinfintesum}) reduces to \be
\sum_{k=0}^{n-1}\lambda_{k}(r_{1},\ldots,r_{n},0,\ldots) r^{(1)}_k =0
\ee
with $\lambda_{k_{0}}(r_{1},\ldots,r_{n},0,\ldots)\neq 0$, which leads a contradiction with the Lemma, p.371, James \cite{james}. The lemma, which plays a crucial role in the proof of uniqueness of solution of the system (\ref{holonomicsystemforfinite}), was first stated by James but its proof was not given. Next, we give one proof by induction.

Recall the James' lemma says that
\be \label{linearfintesum}
\sum_{k=0}^{n-1}\lambda_{k}(r_{1},\ldots,r_{n}) r^{(1)}_k(y_{2},\ldots,y_{n}) =0
\ee implies $\lambda_{0}(r_{1},\ldots,r_{n})=0, \lambda_{1}(r_{1},\ldots,r_{n})=0, \ldots, \lambda_{n-1}(r_{1},\ldots,r_{n})=0$.

For $n=1$, clearly $\lambda_{0}(r_{1})=0$. Assume that we have completed the proof when $n<m$. We consider the case where $n=m$.

Writing $\lambda_{k}(r_{1},\ldots,r_{m})=\sum_{I\in \mathfrak{L_{m}}}C^{(k)}_{I}r^{I}, k=0,1,\ldots,m-1$,  we claim: for $I=(i_{1},\ldots,i_{m})$, if $i_{m}=0$, then $C^{(k)}_{I}=0, k=0,1,\ldots,m-1$. In fact, When $i_{m}=0$ and $0\leq k<m-1$, putting $y_{m}=0$,  (\ref{linearfintesum}) reduces to
\be
\sum_{k=0}^{m-2}\lambda_{k}(r^{(m)}_{1},\ldots,r^{(m)}_{m-1},0) r^{(1,m)}_k(y_{2},\ldots,y_{m-1},0) =0.
\ee
By induction, we have $\lambda_{k}(r^{(m)}_{1},\ldots,r^{(m)}_{m-1},0)=0$, implying that $C^{(k)}_{I}=0, 0\leq k<m-1$. When $i_{m}=0$ and $k=m-1$, putting $y_{1}=0$,
since each nonzero term of $\lambda_{k}(r_{1},\ldots,r_{m})$ ($k<m-1$) contains the factor $r_{m}$, it must follow that $\lambda_{m-1}(r^{(1)}_{1},\ldots,r^{(1)}_{m-1},0) r^{(1)}_{m-1}=0$, which results in $C^{(m-1)}_{I}=0$.

For the coefficient $C^{(k)}_{I}$ with $i_{m}>0$, if $\{C^{(k)}_{I}\neq 0:I\in \mathfrak{L_{m}}, i_{m}>0, k=0,\ldots, m-1\}$ is not empty,
 denoting by $C^{(k_{0})}_{I_{0}}$ one nonzero coefficient such that the last component $i_{m}$ of $I_{0}=(i_{1},\ldots,i_{m})$ is the smallest of
$\{i_{m}:C^{(k)}_{I}\neq 0,I\in \mathfrak{L_{m}}, k=0,\ldots, m-1\}$. Rewrite the equation (\ref{linearfintesum}) as
\be
\sum_{k=0}^{m-1}\frac{\lambda_{k}(r_{1},\ldots,r_{m})}{r^{i_{m}}_{m}} r^{(1)}_k(y_{2},\ldots,y_{m}) =0,
\ee
using the above claim, we obtain $C^{(k_{0})}_{I_{0}}=0$, which is a contradiction.

\section{Convergence  of super hypergeometric series}\label{appendixSeries}

 We  study the convergence of the series given  by \eqref{hfpq}.   We assume that $\alpha>0$ and  $x=(x_1,\ldots,x_n), y=(y_1,\ldots,y_m)$.  An estimation of $SP_{\kappa}^{(\alpha)}(x,y)$ is first given by the following lemma.
\begin{lemma}\label{upbound}
Let $\|(x,y)\|=\max\{|x_1|,\ldots,|x_n|, |y_1|,\ldots,|y_m|\}$ and $r_{1}=\max\{\alpha, 1/\alpha\}$. There exists a positive constant $C_{n,m}$ depending only on $n, m$ such that
$$|SP_{\kappa}^{(\alpha)}(x,y)|\leq C_{n,m}\sqrt{h_{\kappa}^{(\alpha)}/{h^{'}}_{\kappa}^{(\alpha)}}\big(r_{1}(n+r_{1} m) \|(x,y)\|\big)^{|\kappa|},$$
where   $${h^{'}}_{\kappa}^{(\alpha)}=\prod_{(i,j)\in\kappa}\left(a_\kappa(i,j)+\frac{1}{\alpha}l_\kappa(i,j)+\frac{1}{\alpha}\right).$$
\end{lemma}

\begin{proof}
Put $|\kappa|=k>0$, and write $T=\|(x,y)\|$ and
$$SP_{\kappa}^{(\alpha)}=\sum_{|\sigma|=k}\chi_{\sigma}p_{\sigma,\alpha}.$$
If we apply Cauchy's inequality to the sum on the right-hand side
 we obtain $$|SP_{\kappa}^{(\alpha)}|^{2}\leq\Big(\sum_{|\sigma|=k}\chi^{2}_{\sigma}z_{\sigma}\alpha^{\ell(\sigma)}\Big)
 \Big(\sum_{|\sigma|=k}\frac{|p_{\sigma,\alpha}|^{2}}{z_{\sigma}\alpha^{\ell(\sigma)}}\Big),$$
 where $z_{\sigma}=(1^{\sigma_{1}}2^{\sigma_{2}}\cdots)\sigma_{1}!\sigma_{2}!\cdots$.

 On the other hand, by definition of super Jack polynomials we have
 $$P_{\kappa}^{(\alpha)}(z_{1},\ldots,)=\sum_{|\sigma|=k}\chi_{\sigma}p_{\sigma}(z_{1},\ldots,)$$
 and hence
 $$\sum_{|\sigma|=k}\chi^{2}_{\sigma}z_{\sigma}\alpha^{\ell(\sigma)}=h_{\kappa}^{(\alpha)}/{h^{'}}_{\kappa}^{(\alpha)}.$$
 Here we  use the scalar product of Jack functions,  see (10.16), section IV, \cite{macdonald}.

 Since \cite{macdonald} $$\sum_{|\sigma|=k}\frac{p_{\sigma}(z)}{z_{\sigma}}=h_{k}:=\mathrm{coefficent\  of\  } u^{k} \mathrm{\ in\ } \prod_{i\geq 1}(1-z_{i}u)^{-1}, $$
 after the action of $\varphi^{(-\alpha)}_{n,m}$ on the both sides we get
 $$\sum_{|\sigma|=k}\frac{p_{\sigma,-\alpha}(x,y)}{z_{\sigma}}=\mathrm{coefficent\  of\  } u^{k} \mathrm{\ in\ } \prod_{i= 1}^{n}(1-x_{i}u)^{-1}\prod_{j= 1}^{m}(1-x_{j}u)^{-\alpha}.$$
 Hence, it follows from

 $$\frac{(\alpha)_{d}}{d!}\leq r_{1}^{d} \ \mathrm{and}\ \binom{d+k}{k}\leq (d+1)k^{d}$$ that $$\sum_{|\sigma|=k}\frac{p_{\sigma,-\alpha}(|x|, |y|)}{z_{\sigma}}\leq \binom{n+m-1+k}{k}(r_{1}T)^{k}\leq (n+m)k^{n+m-1}(r_{1}T)^{k}.$$
 This gives
 \begin{align*}\sum_{|\sigma|=k}\frac{|p_{\sigma,\alpha}|^{2}}{z_{\sigma}\alpha^{\ell(\sigma)}}&\leq \big(r_{1}(n+m r_{1})T\big)^{k}\sum_{|\sigma|=k}\frac{p_{\sigma,-\alpha}(|x|, |y|)}{z_{\sigma}}\\
 &\leq (C_{n,m})^{2}\big(r_{1}(n+m r_{1})T\big)^{2k},\end{align*}
 where $(C_{n,m})^{2}:=\sup_{k}\{(n+m)k^{n+m-1}(n+m r_{1})^{-k}\}$.
\end{proof}

\begin{theorem}
(1) If $p\leq q$, then the series \eqref{hfpq} converges absolutely for all $(x,y)\in \mathbb{C}^{n+m}$.

(2) If $p= q+1$, then \eqref{hfpq} converges absolutely for $\|(x,y)\|<1/\big(r^{2}_{1}(n+r_{1} m)\big)$.

(3) If $p>q+1$, then there does not exist a positive constant $\rho$ such that \eqref{hfpq}  converges absolutely for $\|(x,y)\|<\rho$,  unless it terminates.
\end{theorem}
\begin{proof}
We compare the series \eqref{hfpq} with the hypergeometric series with one variable $z$
\begin{equation*}
    {}_pF_q(a_1,\ldots,a_p;b_1,\ldots,b_q;z) = \sum_{k=0}^{\infty}\frac{ (a_1)_k\cdots(a_p)_k}{(b_1)_k \cdots (b_q)_k}\frac{z^{k}}{k!},
\end{equation*}
which has radius of convergence $R=\infty$ if $p\leq q$, $R=1$ if $p= q+1$, and $R=0$ if $p>q+1$ unless it terminates.

Note first that
every  $\kappa \in H_{n,m}$ can be rewritten uniquely in the form
$\kappa=\kappa_{0}\cup \pi'$ where $\kappa_{0}=(\kappa_{1},\ldots,\kappa_{n})$ and $\pi'=(\kappa_{n+1},\ldots, \kappa_{n+i},\ldots)$ with $\kappa_{n+i}\leq m$, i.e. $\pi=(\pi_{1},\ldots,\pi_{m})$. Then we have $$\lbrack a_1\rbrack^{(\alpha)}_\kappa=\lbrack a_1\rbrack^{(\alpha)}_{\kappa_{0}}
\lbrack a_1-\frac{n}{\alpha}\rbrack^{(\alpha)}_{\pi'}=
\lbrack a_1\rbrack^{(\alpha)}_{\kappa_{0}}
\lbrack {n-\alpha a_1}\rbrack^{(1/\alpha)}_{\pi}.
$$
Write \begin{multline*}a_{ui}=a_{u}-(1/\alpha)(i-1), b_{vi}=b_{v}-(1/\alpha)(i-1), \tilde{a}_{uj}=n-\alpha a_{u}-\alpha(j-1),\\
 \tilde{b}_{vj}=n-\alpha b_{v}-\alpha(j-1), 1\leq u \leq p, 1\leq v \leq q, 1\leq i \leq n,1\leq j \leq m,\end{multline*}

Set $r_{2}=\min\{\alpha, 1/\alpha\}$, obviously  we have $r_{1} r_{2}=1$. Since
\begin{align*}&h_{\kappa}^{(\alpha)}{h^{'}}_{\kappa}^{(\alpha)}=\prod_{(i,j)\in\kappa}
\left(1+\kappa_{i}-j+\frac{1}{\alpha}(\kappa^{'}_{j}-i)\right)\left(\kappa_{i}-j+\frac{1}{\alpha}(\kappa^{'}_{j}-i+1)\right)\\
&=h_{\kappa_{0}}^{(\alpha)}{h^{'}}_{\kappa_{0}}^{(\alpha)}\prod_{(j,i)\in\pi}\left(1+\pi^{'}_{i}-j+\frac{1}{\alpha}(\pi_{j}-i)\right)
\left(\pi^{'}_{i}-j+\frac{1}{\alpha}(\pi_{j}-i+1)\right),
\end{align*}
it is easily verified that
$h_{\kappa}^{(\alpha)}{h^{'}}_{\kappa}^{(\alpha)}\geq r_{2}^{2|\kappa|}\left(\kappa_{1}!\cdots \kappa_{n}!\right)^{2}\left(\pi_{1}!\cdots \pi_{m}!\right)^{2}$.

By Lemma \ref{upbound}, we have
\begin{align*}&\sum_{k=0}^{\infty}\sum_{|\kappa|=k}\left| \frac{1}{h_{\kappa}^{(\alpha)}}\frac{\lbrack a_1\rbrack^{(\alpha)}_\kappa\cdots\lbrack a_p\rbrack^{(\alpha)}_\kappa}{\lbrack b_1\rbrack^{(\alpha)}_\kappa
    \cdots\lbrack b_q\rbrack^{(\alpha)}_\kappa}SP_{\kappa}^{(\alpha)}(x,y)\right|\\
&\leq C_{n,m} \cdot \prod_{i=1}^{n}\bigg\{\sum_{\kappa_{i}\geq 0}
\left| \frac{ (a_{1i})_{\kappa_{i}}\cdots (a_{pi})_{\kappa_{i}}}{(b_{1i})_{\kappa_{i}} \cdots (b_{qi})_{\kappa_{i}}}\right|
\frac{\big(r^{2}_{1}(n+r_{1} m) \|(x,y)\|\big)^{\kappa_{i}}}{\kappa_{i}!}\bigg\}\\
&\cdot \prod_{j=1}^{m}\bigg\{\sum_{\pi_{j}\geq 0}
\left| \frac{ (\tilde{a}_{1j})_{\pi_{j}}\cdots (\tilde{a}_{pj})_{\pi_{j}}}{(\tilde{b}_{1j})_{\pi_{j}} \cdots (\tilde{b}_{qj})_{\pi_{j}}}\right|
\frac{\big(r^{2}_{1}(n+r_{1} m) \|(x,y)\|\big)^{\pi_{j}}}{\pi_{j}!}\bigg\}.
\end{align*}
Thus (1) and (2) of the theorem are complete.

Notice the fact: \eqref{hfpq}   converges absolutely for $\|(x,y)\|<\rho$  implies that the series with $y=(0,\ldots,0)$ or $x=(0,\ldots,0)$  also converges absolutely for $\|x\|<\rho$ or $\|y\|<\rho$ respectively.  (3)  of the theorem follows from divergence of the series associated to Jack polynomials, see part (3) of Proposition 1, Kaneko \cite{kaneko}.
\end{proof}

\end{appendix}

\end{document}